\documentclass[11pt,reqno]{article}

\usepackage{amssymb,amscd,amsmath,amsthm,color}
\usepackage{authblk}
\usepackage[hidelinks]{hyperref}
\usepackage{pdfpages}

\theoremstyle{definition} 
\newtheorem{thm}{Theorem}[section]
\newtheorem{prop}[thm]{Proposition}
\newtheorem{lemma}[thm]{Lemma}
\newtheorem{cor}[thm]{Corollary}
\newtheorem{definition}[thm]{Definition}
\newtheorem{remark}[thm]{Remark}
\newtheorem{example}[thm]{Example}

\numberwithin{equation}{section}

\def\cH{\mathcal{H}}
\def\cK{\mathcal{K}}
\def\cB{\mathcal{B}}
\def\cA{\mathcal{A}}
\def\cE{\mathcal{E}}
\def\BH{\cB(\cH)}
\def\BK{\cB(\cK)}
\def\<{\langle}
\def\>{\rangle}
\def\bR{\mathbb{R}}
\def\eps{\varepsilon}
\def\bC{\mathbb{C}}
\def\BS{\mathrm{BS}}
\def\a{\mathbf{a}}
\def\b{\mathbf{b}}
\def\cM{\mathcal{M}}
\def\ffi{\varphi}
\def\bw{\mathbf{w}}
\def\bx{\mathbf{x}}
\def\by{\mathbf{y}}
\def\bbz{\mathbf{z}}
\def\bv{\mathbf{v}}
\def\bu{\mathbf{u}}
\def\HS{\mathrm{HS}}
\def\DH{\mathcal{D}(\cH)}
\def\BKM{\mathrm{BKM}}

\def\A{\mathcal{A}}
\def\B{\mathcal{B}}

\def\E{\mathcal{E}}
\def\F{\mathcal{F}}

\def\I{\mathcal{I}}

\def\K{\mathcal{K}}

\def\M{\mathcal{M}}

\def\P{\mathcal{P}}

\def\X{\mathcal{X}}
\def\Y{\mathcal{Y}}

\def\E{\mathcal{E}}
\def\F{\mathcal{F}}

\def\K{\mathcal{K}}

\def\M{\mathcal{M}}

\def\norm#1{\left\Vert #1 \right\Vert}

\def\vfi{\varphi}
\def\hil{{\mathcal H}}
\def\kil{{\mathcal K}}
\def\A{{\mathcal A}}

\def\I{\mathcal{I}}

\def\M{\mathcal{M}}
\def\P{\mathcal{P}}

\def\X{{\mathcal X}}
\def\Y{{\mathcal Y}}

\def\half{\frac{1}{2}}
\def\iff{\Longleftrightarrow}
\def\imp{\Longrightarrow}
\def\ep{\varepsilon}
\def\bN{\mathbb{N}}
\def\bC{\mathbb{C}}
\def\bR{\mathbb{R}}

\def\bM{\mathbb{M}}
\def\bz{\left(}
\def\jz{\right)}
\def\inv{^{-1}}

\def\map{\Phi}
\def\mapp{\Psi}

\def\what{\widehat}

\def\rho{\varrho}

\def\down{^{\downarrow}}
\def\sa{\mathrm{sa}}
\def\diag{\mathrm{diag}}
\def\ol{\overline}
\def\sa{\mathrm{sa}}

\def\egy{\mathbf{\mathrm{1}}}

\def\meas{\mathrm{meas}}
\def\pro{\mathrm{pr}}
\def\vN{\mathrm{vN}}

\def\ONB{\mathrm{ONB}}
\def\wtilde{\widetilde}
\def\crn{[\rho/\sigma]}
\def\povm{\mathrm{POVM}}

\newcommand{\ki}{\emph}

\newcommand{\s}{\mbox{ }}
\newcommand{\ds}{\mbox{ }\mbox{ }}
\newcommand{\inner}[2]{\left\langle #1 , #2\right\rangle}
\newcommand{\abs}[1]{\left| #1 \right|}

\newcommand{\diad}[2]{|#1\rangle\langle #2|}
\newcommand{\pr}[1]{\diad{#1}{#1}}

\newcommand{\fix}[1]{\F_{#1}}

\newcommand{\maxdiv}[1]{\widehat S_{#1}}
\newcommand{\per}[1]{P_{#1}}
\newcommand{\conn}[1]{\tau_{#1}}

\DeclareMathOperator{\id}{id}
\DeclareMathOperator{\supp}{supp}
\DeclareMathOperator{\ran}{ran}

\DeclareMathOperator{\Tr}{Tr}

\DeclareMathOperator{\spec}{spec}
\DeclareMathOperator{\Sp}{spec}

\renewcommand\theenumi{(\roman{enumi})}

\renewcommand{\thefootnote}{\alph{footnote}}

\begin{document}
\allowdisplaybreaks

\centerline{\LARGE Different quantum $f$-divergences}
\medskip

\centerline{\LARGE  and the reversibility of quantum operations}
\bigskip
\bigskip
\centerline{\Large
Fumio Hiai$^{1,}$\footnote{{\it E-mail address:} hiai.fumio@gmail.com} and
Mil\'an Mosonyi$^{2,}$\footnote{{\it E-mail address:} milan.mosonyi@gmail.com}}

\medskip
\begin{center}
$^1$\,Tohoku University (Emeritus), \\
Hakusan 3-8-16-303, Abiko 270-1154, Japan
\end{center}

%

\begin{center}
$^2$\,Mathematical Institute, Budapest University of Technology and Economics,\\
Egry J.~u.~1, 1111 Budapest, Hungary
\end{center}

\medskip
\begin{abstract}
The concept of classical $f$-divergences gives a unified framework to construct and study measures of 
dissimilarity of probability distributions; special cases include the relative entropy and the R\'enyi divergences. Various quantum versions of this concept, and more narrowly, the concept of R\'enyi divergences, have been introduced in the literature with applications in quantum information theory; most notably Petz' quasi-entropies (standard $f$-divergences), Matsumoto's maximal $f$-divergences, measured $f$-divergences, 
and sandwiched and $\alpha$-$z$-R\'enyi divergences. 

In this paper we give a systematic overview of the various concepts of quantum $f$-divergences, 
with a main focus on their monotonicity under quantum operations, and the implications of the preservation 
of a quantum $f$-divergence by a quantum operation. In particular, we compare the standard 
and the maximal 
$f$-divergences regarding their ability to detect the reversibility of quantum operations. We also show 
that these two quantum $f$-divergences are strictly different for non-commuting operators unless $f$ is a polynomial, and obtain some analogous partial results for the relation between the measured and the standard $f$-divergences.

We also study the monotonicity of the $\alpha$-$z$-R\'enyi divergences under the special class 
of bistochastic maps that leave one of the arguments of the R\'enyi divergence invariant, and determine domains of the parameters $\alpha,z$ where monotonicity holds, and where the preservation of the 
$\alpha$-$z$-R\'enyi divergence implies the reversibility of the quantum operation.


\end{abstract}

\newpage
{\baselineskip=14pt
\tableofcontents
}

\section{Introduction}
\renewcommand{\thefootnote}{\arabic{footnote}}
\setcounter{footnote}{0}

Quantum divergences give measures of dissimilarity of quantum states (or, more generally, positive semidefinite operators on a Hilbert space). 
While from a purely mathematical point of view, any norm on the space of operators would do this job, 
for information theoretic applications it is often more beneficial to consider other types of divergences, that are more naturally linked
to the given problems. Undisputably the most important such divergence is \ki{Umegaki's
relative entropy} \cite{Umegaki}, defined for two positive operators $\rho,\sigma$ as\footnote{In the Introduction we assume all positive operators to be invertible for simplicity; the precise definitions for not necessarily invertible positive semidefinite operators will be given later in the paper.}
\begin{align}\label{Umegaki intro}
S(\rho\|\sigma):=\Tr \rho(\log \rho-\log \sigma).
\end{align}
The operational significance of this quantity was established in \cite{HP,ON}, as an optimal error exponent in the hypothesis testing problem of Stein's lemma. 
Moreover, the relative entropy serves as a parent quantity to many other measures of information and correlation, like
the von Neumann entropy, the conditional entropy and the coherent information, the mutual information, the Holevo capacity, and more, each of which quantifies an optimal achievable rate in a 
certain quantum information theoretic problem; see, e.g., \cite{Wilde_book}.

The relative entropy and its derived quantities mentioned above appear in the so-called first order versions of coding theorems, typically as the optimal exponent of some operational quantity 
(e.g., the coding rate or the compression rate) under the assumption that a certain error probability vanishes in the asymptotic treatment of the problem. 
In a more detailed analysis of these problems, one can 
try to give a quantitative description of the interplay between the relevant error probability and the operational quantity of interest (e.g., the coding rate) 
by fixing the asymptotic rate of one and optimzing the rate of the other. As it turns out, in every case when such a quantification has been found, it is given in terms of
two different families of divergences:
the \ki{(conventional) R\'enyi divergences}
\begin{align}\label{conventional Renyi intro}
D_\alpha(\rho\|\sigma):={1\over\alpha-1}\log\frac{\Tr \rho^\alpha \sigma^{1-\alpha}}{\Tr\rho},
\end{align}
or the recently discovered \ki{sandwiched R\'enyi divergences} \cite{Renyi_new,WWY13}
\begin{equation}\label{F-1.1}
D_\alpha^*(\rho\|\sigma):={1\over\alpha-1}\log
{\Tr(\sigma^{1-\alpha\over2\alpha}\rho \sigma^{1-\alpha\over2\alpha})^\alpha\over\Tr \rho};
\end{equation}
see, e.g., \cite{ANSzV,CMW,Hayashibook,Hayashi,HT14,MO13,MO_cq,Nagaoka_Hoeffding}.
Both families are defined for any $\alpha>0,\,\alpha\ne 1$, and 
the values for $\alpha\in\{0,1,+\infty\}$ can be obtained by taking the respective limit in $\alpha$.
In particular, the limit for $\alpha\to 1$ gives $\frac{1}{\Tr\rho}S(\rho\|\sigma)$.
It is important to note that these two families coincide for commuting $\rho$ and $\sigma$.
A two-parameter unification of these two families is given by the so-called 
$\alpha$-$z$-R\'enyi divergences, introduced in 
\cite{AD,JOPP} as
\begin{equation}\label{F-1.2}
D_{\alpha,z}(\rho\|\sigma):={1\over\alpha-1}\log
{\Tr(\sigma^{1-\alpha\over2z}\rho^{\alpha\over z}\sigma^{1-\alpha\over2z})^z\over\Tr \rho},
\qquad\alpha,z>0,\ \alpha\ne1.
\end{equation}
The previous two families are embedded as $D_{\alpha,1}=D_{\alpha}$
and $D_{\alpha,\alpha}=D_\alpha^*$ for every $\alpha$.

In the classical case, both the relative entropy and the R\'enyi divergences can be expressed as 
$f$-divergences, introduced by Csisz\'ar \cite{Csiszar_fdiv} and Ali and Silvey \cite{AS} for two probability distributions $p,q$ on a finite set $\X$ and a convex function $f:\,(0,+\infty)\to\bR$ as 
\begin{align}\label{classical f-div intro}
S_f(p\|q):=\sum_{x\in\X}q(x)f\bz\frac{p(x)}{q(x)}\jz.
\end{align}
The relative entropy corresponds to $f(t):=\eta(t):=t\log t$, while the R\'enyi divergences can be expressed as $D_{\alpha}(p\|q)=\frac{1}{\alpha-1}\log S_{f_{\alpha}}(p\|q)$, $f_{\alpha}(t):=\mathrm{sign}(\alpha-1)t^{\alpha}$. Moreover, various other divergences for probability distributions can be cast in this form; among others, the variational distance and the $\chi^2$-divergence.
An advantage of this general formulation is that important properties of the various divergences, like joint convexity and monotonicity under stochastic maps, can be derived from \eqref{classical f-div intro} and the convexity of $f$, thus providing a unified framework to study the different divergences.

Motivated by the success of the classical $f$-divergences, various quantum generalizations of the concept have been put forward in the literature. The closest in properties to the classical version are probably 
the \ki{standard $f$-divergences}, that are a special case of Petz' quasi-entropies \cite{P85,P86} (see also \cite{HMPB}), 
and are defined as
\begin{align}\label{standard fdiv intro}
S_f(\rho\|\sigma):=\Tr \sigma^{1/2}f(L_\rho R_{\sigma^{-1}})(\sigma^{1/2}),
\end{align}
where $L_\rho$ and $R_{\sigma^{-1}}$ are the left and the right multiplication operators by 
$\rho$ and $\sigma\inv$, respectively.
The choices $f=\eta$ and $f=f_{\alpha}$ give rise to the Umegaki relative entropy \eqref{Umegaki intro}
and the conventional R\'enyi divergences \eqref{conventional Renyi intro}, just as in the classical case. 
An alternative version, that coincides with the above for commuting $\rho$ and $\sigma$, has been 
introduced by Petz and Ruskai in \cite{PR} as 
\begin{align*}
\maxdiv{f}(\rho\|\sigma):=\Tr \sigma f(\sigma^{-1/2}\rho \sigma^{-1/2}).
\end{align*}
It has been shown recently by Matsumoto \cite{Ma} that this notion of quantum $f$-divergence is maximal among the monotone quantum $f$-divergences, and, moreover, it can be expressed 
in the form of a natural optimization of the $f$-divergences of classical distribution functions that can be mapped into the given quantum operators (see Section \ref{sec:fdiv intro} for details). Hence, following Matsumoto's terminology, we will refer to them as \ki{maximal $f$-divergences}.

The relative entropy and the standard and the sandwiched R\'enyi divergences take strictly positive values on pairs of unequal quantum states, 
supporting their interpretation as measures of distinguishability; for the standard $f$-divergences the same holds for every strictly convex $f$ with 
the normalization $f(1)=0$ \cite[Proposition A.4]{HMPB}. For any measure $D$ of distinguishability of states, it is natural to assume that 
stochastic operations do not increase the distinguishability, i.e., the monotonicity inequality 
\begin{align}\label{mon ineq}
D(\map(\rho)\|\map(\sigma))\le D(\rho\|\sigma)
\end{align}
holds for any states (or, more generally, positive operators) $\rho,\sigma$, and quantum operation $\map$.
For physical applications, the latter is usually defined as a completely positive and trace-preserving (CPTP) map,
although from a purely mathematical point it is also interesting to study monotonicity under maps with weaker 
positivity properties \cite{HMPB,M-HR,P85,P86}. The monotonicity inequality is also called the data-processing inequality 
in information theory, and it is often considered as a primary requirement for a quantum quantity to be called a divergence.
It is well-known that the standard R\'enyi divergences satisfy monotonicity exactly when 
$\alpha\in[0,2]$ \cite{HMPB,LR,P86,TCR}, 
and the sandwiched R\'enyi divergences when $\alpha\in[1/2,+\infty]$ \cite{Beigi,CFL,FL13,Hi3,Renyi_new,WWY13};
this gives a further insight into why one needs two separate families of R\'enyi divergences in the quantum case.
Domains of the parameters $\alpha,z$ where the $\alpha$-$z$-R\'enyi divergences satisfy monotonicity have been determined in 
\cite{CFL,Hi3} (see also \cite[Theorem 1]{AD}), but a complete characterization of all $\alpha,z$ values for which monotonicity holds is still missing.

As with any inequality, it is natural to ask when the monotonicity inequality \eqref{mon ineq} holds as an equality, i.e., 
when does a quantum operation preserve the distinguishability of two states (as measured by a certain quantum divergence).
It is clear that this is the case for any monotone divergence whenever $\map$ is reversible on $\{\rho,\sigma\}$ in the sense that there exists
a quantum operation $\mapp$ such that $\mapp(\map(\rho))=\rho$ and $\mapp(\map(\sigma))=\sigma$. 
It is a highly non-trivial observation with far-reaching consequences that for a large class of divergences the converse is also true. 
This line of research was initiated by Petz \cite{Petz_sufficiency2,Petz_sufficiency}, who showed this converse for the relative entropy and the standard R\'enyi divergence with parameter $1/2$,
and determined a canonical reversion map. His results were later extended to standard R\'enyi divergences
with other parameter values \cite{JP,JP_survey}, and more general standard $f$-divergences in \cite{HMPB,Je}.
Various other, mainly algebraic, characterizations of the preservation of the relative entropy 
were given, e.g., in \cite{mon_revisited,Ruskai_sufficiency}.
In \cite{HJPW}, a structural characterization of the equality case of the strong subadditivity of entropy (a special case of the monotonicity of the relative entropy) was presented, 
which was used to give a constructive description of quantum Markov states. 
This was later extended in \cite{MP} to a structural characterization of triples $(\map,\rho,\sigma)$
such that $\map$ is reversible on $\{\rho,\sigma\}$.
Also, the equality case in the
joint convexity (another special instance of monotonicity) of various quasi-entropies was clarified in \cite{JR}.
The above characterizations are all related to quantum $f$-divergences of the form \eqref{standard fdiv intro},
in particular, mainly to the standard R\'enyi relative entropies \eqref{conventional Renyi intro}.
Very recently, an algebraic characterization of the preservation of the sandwiched R\'enyi divergences \eqref{F-1.1}
with parameter values $\alpha>1/2$
was given in \cite{LRD}, based on the variational formula of \cite{FL13}. 
Moreover, in \cite{Jencova_rev16} it was shown that the preservation of a sandwiched R\'enyi divergence
with $\alpha>1$ implies reversibility. This was based on the complex interpolation method in non-commutative $L_p$ spaces, following the approach of \cite{Beigi}.

In this paper we give a systematic overview of the various concepts of quantum $f$-divergences, 
with a main focus on their monotonicity under quantum operations, and the implications of the preservation 
of a quantum $f$-divergence by a quantum operation. 
After summarizing the necessary preliminaries in Section \ref{sec:preliminaries}, we give a detailed overview 
of the standard and the maximal $f$-divergences in Section \ref{sec:fdiv}. Unlike in previous works, we define these $f$-divergences
for operator convex functions on $(0,+\infty)$ that need not have a finite limit from the right at $0$, and establish the relevant continuity properties to make sense of the definition.
In the introduction of the maximal $f$-divergences in Section \ref{sec:maximal f-div}, we deviate from Matsumoto's treatment in that we take the notion of the operator perspective as our starting point.
To define the maximal 
$f$-divergences for not necessarily invertible operators, we establish the extension of the operator perspective for certain settings with non-invertible operators in Propositions \ref{prop:persp extension} and 
\ref{prop:persp extension2}, that seems to be new and probably interesting in itself. 
It is easy to see, as we show in Proposition \ref{P-3.8}, that even with this more general definition, the 
standard $f$-divergences are monotone under the same class of positive trace-preserving maps 
as considered before in \cite{HMPB}, while the maximal $f$-divergences are monotone under arbitrary positive maps, as follows from standard facts in matrix analysis.

We summarize the known characterizations for the preservation of the standard $f$-divergen\-ces by positive trace-preserving maps in Theorems \ref{T-3.12} and \ref{thm:decomposition}.
Theorem \ref{T-3.12} contains a slight extension as compared to previous results, as we show that ordinary positivity of the reversion map 
(as opposed to a stronger positivity criterion in \cite[Theorem 5.1]{HMPB}) is sufficient for 
the preservation of any $f$-divergence; this is possible due to the recent developments in this direction in 
\cite{Beigi,M-HR}.
In Theorem \ref{T-3.21}, we give a slight extension of Matsumoto's prior results on the characterization of the preservation of the maximal $f$-divergences by quantum operations. In particular, we remove a technical restriction on the function $f$ in \cite[Lemma 12]{Ma}, and show that the preservation of any maximal 
$f$-divergence with a non-linear operator convex function $f$ implies the preservation of any other maximal 
$f$-divergence. In particular, the choice $f_2(t)=t^2$ implies that the preservation of a 
maximal 
$f$-divergence with any non-linear operator convex function $f$ is equivalent to the preservation 
of the standard $f$-divergence $S_{f_2}$ (as $S_{f_2}=\what S_{f_2}$), which in turn is known not to imply 
reversibility, as was shown in \cite[Remark 5.4]{HMPB}. Hence, we conclude that the preservation of the 
maximal $f$-divergences has strictly weaker consequences than the preservation of the standard 
$f$-divergences. We discuss this difference in more detail in Section \ref{sec:reversibility comparison}. In particular, we give (in Example \ref{E-4.4}) a simple explicit construction  
for a channel $\map$ and two states $\rho,\sigma$ on $\bC^3$ such that $\map$ 
preserves all the maximal $f$-divergences of $\rho$ and $\sigma$, but does not preserve any of their standard $f$-divergences whenever $f$ satisfies some mild technical condition. 
On the other hand, we show in Proposition \ref{P-4.6} that for unital qubit channels, preservation of the maximal $f$-divergences is equivalent to the preservation of the standard $f$-divergences,
and we show in Proposition \ref{C-3.22} that the same holds whenever the outputs of the channel commute with each other.

Section \ref{sec:comparison} is devoted to the comparison of three different notions of quantum $f$-divergences: the standard $f$-divergence, the maximal $f$-divergence and the measured (minimal) $f$-divergence. 
In Section \ref{sec:maxdiv vs mod div} we use Matsumoto's reverse tests and the characterization of the preservation of standard $f$-divergences to show
that for non-commuting states, their maximal $f$-divergences are strictly larger than their standard $f$-divergences
for all operator convex functions with a large enough support of their representing measure in a canonical integral representation (given in \cite[Theorem 8.1]{HMPB}).
Moreover, for qubit operators this condition can be dropped, as we show in Proposition \ref{P-4.5}.
Section \ref{sec:reversibility comparison} is devoted to the comparison of the standard and the maximal 
$f$-divergences regarding their ability to detect the reversibility of quantum operations, as explained above.
Finally, in Section \ref{sec:meas fdiv}, we discuss the measured $f$-divergences, and show that for any pair of non-commuting operators, their measured $f$-divergence is strictly smaller than their standard $f$-divergence,
provided again some technical conditions on the size of the support of the representing measure of $f$ are satisfied.
We also review, and give a slight extension of recent results on the ordering of the standard, the sandwiched, the measured, and the regularized measured R\'enyi divergences,
in Proposition \ref{prop:Renyi relations}. We close this section by a Pinsker inequality on the projectively measured $f$-divergences, given in Proposition \ref{Pinsker}. 

In the last section, Section \ref{sec:Renyi rev}, we consider the behaviour of the $\alpha$-$z$-R\'enyi divergences under bistochastic maps that 
leave one of the arguments of the R\'enyi divergence invariant, and determine domains of $\alpha,z$ values where monotonicity holds, and where the preservation of the 
$\alpha$-$z$-R\'enyi divergence implies the reversibility of the quantum operation.
This setup contains dephasing maps, i.e., (block-)diagonalization of one operator in a basis in which the other operator is already (block-)diagonal, or, more generally, conditional expectations onto a subalgebra that contains one of the arguments of the R\'enyi divergence. A particular example is the pinching by the 
eigenprojectors of the second argument of the R\'enyi divergence; the behaviour of the sandwiched R\'enyi divergences ($z=\alpha$ case) under these maps played an important role in establishing their operational 
significance in quantum state discrimination \cite{MO13}.
The $\alpha,z$ values where we establish monotonicity contain domains where the monotonicity of the $\alpha$-$z$-R\'enyi divergences is either not known or does not hold for general maps. The analysis 
of the implications of the preservation of the $\alpha$-$z$-R\'enyi divergences is completely new, as this has only been carried out so far for the standard R\'enyi divergences \cite{HMPB,JP,JP_survey,Petz_sufficiency}, and, very recently, for the sandwiched R\'enyi divergences for a part of the parameter range where they are monotone \cite{Jencova_rev16}.

We give supplementary material and some longer proofs in Appendices \ref{sec:persp properties}--\ref{sec:extension proof}.

\section{Preliminaries}
\label{sec:preliminaries}

\subsection{Notations}

Throughout the paper, $\hil,\kil$ will denote finite-dimensional Hilbert spaces. For any 
finite-dimensional Hilbert space $\hil$, $\BH$ will denote the algebra of
linear operators on $\cH$, and $\B(\hil)_{\sa}$ the real subspace of self-adjoint operators in $\B(\hil)$.
The identity operator on $\cH$ is denoted by $I_\cH$ (or simply $I$). 
The spectrum of an operator $X\in\B(\hil)$ is denoted by $\Sp(X)$.

We write $\BH_+$
for the set of positive linear operators on $\cH$. We write $\rho>0$ when $\rho\in\BH_+$ is invertible,
and denote the set of invertible positive operators by $\B(\hil)_{++}$.
For $\rho\in\BH_+$ with spectral decomposition $\rho=\sum_{a\in\Sp(\rho)}aP_a$, we define its real powers by 
$\rho^t:=\sum_{a\in\Sp(\rho),\,a>0}a^tP_a$, $t\in\bR$. In particular,
$\rho^{-1}$ stands for the generalized inverse of $\rho$, and $\rho^0$ is the support
projection of $\rho$, i.e., the projection onto the support of $\rho$. 

The usual trace functional on
$\BH$ is denoted by $\Tr$. We always consider $\BH$ as the Hilbert space with the
{\it Hilbert-Schmidt inner product}
$$
\<X,Y\>_{\HS}:=\Tr X^*Y,\qquad X,Y\in\BH.
$$
For a linear operator $\rho\in\BH$, the {\it left multiplication} $L_\rho$ and the {\it right
multiplication} $R_\rho$ are the linear operators on $\BH$ defined by
$$
L_\rho X:=\rho X,\quad R_\rho X:=X\rho,\qquad X\in\BH.
$$
If $\rho,\sigma\in\BH_+$, then both $L_\rho$ and $R_\rho$ are positive operators on the Hilbert space $\BH$,
which are commuting, i.e., $L_\rho R_\sigma=R_\sigma L_\rho$.

\subsection{Operator convex and operator monotone functions}
\label{sec:op convex}

In the rest of the paper, unless otherwise stated, we always assume that $f:\,(0,+\infty)\to\bR$ is a continuous function such that the limits 
\begin{align*}
f(0^+):=\lim_{x\searrow 0}f(x)\ds\ds\text{and}\ds\ds
f'(+\infty):=\lim_{x\to+\infty}\frac{f(x)}{x}
\end{align*}
exist in $\bR\cup\{\pm\infty\}$, and they are not both infinity with opposite signs. These assumptions are obviously satisfied when $f$ is convex, in which case 
the limits exist in $(-\infty,+\infty]$, and if $f$ is a differentiable convex function then in fact $f'(+\infty)=\lim_{x\to+\infty}f'(x)$.

A function $f:(0,+\infty)\to\bR$ is called an {\it operator convex} function if the operator
inequality
$$
f(tA+(1-t)B)\le tf(A)+(1-t)f(B),\qquad0\le t\le1
$$
holds for every $A,B\in\BH_{++}$ of any (even infinite-dimensional) $\cH$, where $f(A)$ etc.\ are
defined via usual functional calculus. Also, a function $h:(0,+\infty)\to\bR$ is said to be
{\it operator monotone} if $A\le B$ implies $h(A)\le h(B)$ for every $A,B\in\BH_{++}$ of any
$\cH$. For the general theory of operator monotone and operator convex functions, see, e.g.,
\cite{Bhatia,Hiai_book}. For the rest of the paper, we will mainly follow the convention that
$h$ denotes an operator monotone function, and $f$ an operator convex, or at least convex,
function.


Operator monotone and operator convex functions can be decomposed to simpler functions via
integral representations, a few of which  we recall here for later use.
Every non-negative operator monotone function  $h$ on $(0,\infty)$ can
be uniquely written as
\begin{align}\label{F-2.2}
h(x)=a+bx+\int_{(0,+\infty)}{x(1+s)\over x+s}\,d\nu_h(s),\qquad x\in(0,+\infty),
\end{align}
with $a=h(0^+)$, $b=h'(+\infty)=\lim_{x\to+\infty}h(x)/x$, and a finite positive measure
$\nu_h$ on $(0,+\infty)$ (see \cite[Theorem 2.7.11]{Hiai_book}). 

When $f:\,(0,+\infty)\to\bR$ is operator convex,
it can be written \cite{LR} (see also \cite[(5.2)]{FHR} for a more general form) as
\begin{equation}\label{F-2.3}
f(x)=f(1)+f'(1)(x-1)+c(x-1)^2
+\int_{[0,+\infty)}{(x-1)^2\over x+s}\,d\lambda(s),\quad x\in(0,+\infty),
\end{equation}
with $c\ge0$ and a positive measure $\lambda$ on $[0,+\infty)$ satisfying
$\int_{[0,+\infty)}(1+s)^{-1}\,d\lambda(s)<+\infty$. When $f(0^+)<+\infty$, and hence $f$ extends
by continuity to an operator convex function on $[0,+\infty)$, an alternative integral
representation can be obtained \cite[Theorem 8.1]{HMPB} as
\begin{equation}\label{F-2.4}
f(x)=f(0^+)+ax+bx^2+\int_{(0,+\infty)}\biggl({x\over1+s}-{x\over x+s}\biggr)\,d\mu_f(s),
\qquad x\in(0,+\infty),
\end{equation}
with $a\in\bR$, $b\ge0$ and a positive measure $\mu_f$ on $(0,+\infty)$ satisfying
$\int_{(0,+\infty)}(1+s)^{-2}\,d\mu_f(s)<+\infty$. In the more restrictive case when
$f(0^+)<+\infty$ and $f'(+\infty)<+\infty$, yet another integral representation was given in
\cite[Theorem 8.4]{HMPB} as
\begin{equation}\label{F-2.5}
f(x)=f(0^+)+f'(+\infty)x-\int_{(0,+\infty)}{x(1+s)\over x+s}\,d\nu(s)
\end{equation}
with a finite positive measure $\nu$ on $(0,+\infty)$. Note that the coefficients $c,a,b$
and the representing measures $\lambda,\mu_f,\nu$ are uniquely determined by $f$ in each of the
above integral representations. We make the dependence of $\mu$ on $f$ explicit in \eqref{F-2.4}
for the convenience of later references. Moreover, the representing measures in the above are
explicitly related to each other. Indeed, for $f$ with expression \eqref{F-2.3}, $f(0^+)<+\infty$
if and only if $\int_{[0,+\infty)}s^{-1}\,d\lambda(s)<+\infty$ (in particular,
$\lambda(\{0\})=0$), and in this case, the relation $(1+s)^{-2}\,d\mu_f(s)=s^{-1}\,d\lambda(s)$
holds (the proof of this is left to the reader). Also, for $f$ with expression \eqref{F-2.4}
(hence $f(0^+)<+\infty)$, $f'(+\infty)<+\infty$ if and only if $b=0$ and
$\int_{(0,+\infty)}(1+s)^{-1}\,d\mu_f<+\infty$, and in this case, $d\nu(s)=(1+s)^{-1}\,d\mu_f(s)$
(see the proof of \cite[Theorem 8.4]{HMPB}). Thus, the support of the representing measure for
$f$ is independent of the possible choice of the above integral expressions.

\subsection{Non-commutative perspectives and operator connections}

For any function $\vfi:(0,+\infty)\to\bR$, 
its \ki{perspective} $P_{\vfi}:\,(0,+\infty)\times (0,+\infty)\to\bR$ is defined by
\begin{align*}
P_{\vfi}(x,y):=y\vfi\bz\frac{x}{y}\jz,\ds\ds\ds x,y\in (0,+\infty).
\end{align*}
By definition, $\vfi(x)=P_{\vfi}(x,1)$ for all $x\in(0,+\infty)$, and
the \ki{transpose} $\widetilde\vfi$ of $\vfi$ is defined as
\begin{align*}
\widetilde\vfi(y):=P_{\vfi}(1,y)=y\vfi\bz\frac{1}{y}\jz,\ds\ds\ds y\in (0,+\infty).
\end{align*}
Thus, $\vfi$ and $\widetilde\vfi$ can be considered as marginals of the two-variable function $P_f$.

When $f$ is as at the beginning of the previous section, we can extend $P_f$ to $[0,+\infty)\times[0,+\infty)$ by
\begin{align}\label{persp ext}
P_f(x,y):=\lim_{\eps\searrow0}(y+\eps)f\bz\frac{x+\eps}{y+\eps}\jz
=\begin{cases}
yf(xy^{-1}), & \text{if $x,y>0$}, \\
yf(0^+), & \text{if $x=0$}, \\
xf'(+\infty), & \text{if $y=0$},
\end{cases}
\end{align}
with the convention $0\cdot\infty:=0$.
It is straightforward to see that
\begin{equation}\label{F-2.1}
\widetilde f(0^+)=f'(+\infty),\qquad\widetilde f'(+\infty)=f(0^+).
\end{equation}

It is well-known that the transpose $\widetilde h$ of a non-negative operator monotone function $h$ on
$(0,+\infty)$ is operator monotone again. Similarly, the transpose $\widetilde f$ of an operator
convex function $f$ on $(0,+\infty)$ is operator convex again. For these assertions, see
Propositions \ref{P-A.1} and \ref{P-A.2} of Appendix \ref{sec:persp properties}.

For a function $\ffi$ on $(0,+\infty)$, its
\ki{non-commutative} (or \ki{operator}) \ki{perspective} $\per{\vfi}$ is defined as the two-variable operator function
\begin{equation}\label{F-2.6}
\per{\vfi}:\,(A,B)\in\BH_{++}\times\BH_{++}\longmapsto B^{1/2}\ffi(B^{-1/2}AB^{-1/2})B^{1/2}
\end{equation}
for every finite-dimensional Hilbert space $\hil$. The following simple observation will
be useful:

\begin{lemma}\label{lemma:transpose perspective}
Let $\vfi:\,(0,+\infty)\to\bR$ be any function and $\widetilde \vfi$ be the transpose of $\vfi$.
For every $A,B\in\B(\hil)_{++}$, 
\begin{align*}
\per{\widetilde \vfi}(A,B)=\per{\vfi}(B,A).
\end{align*}
\end{lemma}

\begin{proof}
By definition,
\begin{align*}
\per{\widetilde \vfi}(A,B)
&=B^{1/2}\widetilde \vfi(B^{-1/2}AB^{-1/2})B^{1/2}\\
&=B^{1/2}(B^{-1/2}AB^{-1/2})\vfi(B^{1/2}A\inv B^{1/2})B^{1/2}\\
&=AB^{-1/2}\vfi(XX^*)XA^{1/2}=AB^{-1/2}X\vfi(X^*X)A^{1/2}\\
&=A^{1/2}\vfi(A^{-1/2}BA^{-1/2})A^{1/2}=\per{\vfi}(B,A),
\end{align*}
where $X:=B^{1/2}A^{-1/2}$.
\end{proof}

The following are basic properties of operator perspectives. The proof of (1) is due to
\cite{ENG,Effros,EH}. We give a small extension of the next lemma in Appendix
\ref{sec:persp properties}.

\begin{lemma}\label{lemma:persp properties}
Let $\vfi:\,(0,+\infty)\to\bR$.
\begin{enumerate}
\item[(1)]\label{Effros} 
$\per{\vfi}$ is jointly operator convex 
on $\BH_{++}\times\BH_{++}$ for every finite-dimensional Hilbert space $\hil$ if and only if $\vfi$ is operator convex. 

\item[(2)]\label{KuboAndo}
$\per{\vfi}$ is monotone non-decreasing in both of its arguments 
on $\BH_{++}\times\BH_{++}$ for every finite-dimensional Hilbert space $\hil$ if and only $\vfi$ is a non-negative operator monotone function.
\end{enumerate}
\end{lemma}

Assume that $h$ is a non-negative operator monotone function on $(0,+\infty)$, extended by
continuity to $[0,\infty)$.  Then $(A,B)\mapsto \per{h}(B,A)$ gives an \ki{operator connection},
that we denote by $\conn{h}$, i.e., $A\,\conn{h}\,B=\per{h}(B,A)$ (notice the reversed order of
$A$ and $B$). The general theory of operator connections was developed in an axiomatic way by
Kubo and Ando \cite{KA}. The operator connection $\tau_h$ is extended to pairs of not necessarily
invertible positive operators as
\begin{equation}\label{def-connect}
A\,\tau_h\,B:=\lim_{\eps\searrow0}(A+\eps I)\,\tau_h\,(B+\eps I),\ds\ds\ds A,B\in\B(\hil)_+,
\end{equation}
and it is called an \ki{operator mean} when $h$ further satisfies $h(1)=1$. A main result of
\cite{KA} says that the correspondence $h\leftrightarrow\tau_h$ is an order isomorphism between
the non-negative operator monotone functions and the operator connections. Although
$(A,B)\mapsto A\,\tau_h\,B$ is continuous for decreasing sequences in $\BH_{+}$, it is not
necessarily so for general sequences. Nevertheless, we have the following slightly more general convergence property (whenever $\cH$ is a finite-dimensional Hilbert space). This is
easily seen from the joint monotonicity and the definition \eqref{def-connect} of $\tau_h$.

\begin{lemma}\label{lemma:op mon extension}
Let $h:\,(0,+\infty)\to\bR$ be a non-negative operator monotone function. For any
$A,B\in\B(\hil)_{+}$, and any sequences $A_n,B_n\in\B(\hil)_+$ such that $A\le A_n\to A$ and
$B\le B_n\to B$, the sequence $A_n\,\tau_h\,B_n=\per{h}(B_n,A_n)$ converges to $A\,\tau_h\,B$.
\end{lemma}

When $h$ is a non-negative operator monotone function on $(0,+\infty)$, it admits a unique
integral representation, given in \eqref{F-2.2}, which in turn yields
\begin{align}\label{F-2.8}
A\,\tau_h\,B=aA+bB+\int_{(0,+\infty)}A\,\tau_{h_s}\,B\,d\nu_h(s),\qquad A,B\in\BH_+,
\end{align}
where $h_s(x):=x(1+s)/(x+s)$. In other notation, $A\,\tau_{h_s}\,B={1+s\over s}\{(sA):B\}$,
where $A:B$ is the parallel sum of $A,B\in\BH_+$ (see \cite{KA}).
We say that the operator connection $\conn{h}$ is non-linear if $h$ is non-linear (i.e., the
measure $\nu_h$ is non-zero). 

When $f$ is an operator convex function on $(0,+\infty)$, the extension of its perspective to 
$\B(\hil)_+\times\B(\hil)_+$ is a non-trivial problem, that we will discuss in detail in
Section \ref{sec:maximal f-div}.

\subsection{Monotone metrics}
Let $\DH$ denote the set of invertible density operators on $\cH$, which is a smooth
Riemannian manifold whose tangent space at any foot point is identified with
$$
\BH_{\sa}^0:=\{X\in\BH_{\sa}:\Tr X=0\}.
$$
Let $\kappa:(0,+\infty)\to(0,+\infty)$ be an operator monotone decreasing function such that
$x\kappa(x)=\kappa(x^{-1})$, $x>0$. Since $h(x):=\kappa(x^{-1})=x\kappa(x)$, $x>0$, is
operator monotone, the integral expression \eqref{F-2.2} of $h$ gives that of $\kappa$ as
\begin{equation}\label{F-2.10}
\kappa(x)={a\over x}+b+\int_{(0,+\infty)}{1+s\over x+s}\,d\nu_h(s)
=b+\int_{[0,+\infty)}{1+s\over x+s}\,\nu_\kappa(s),
\end{equation}
where $\nu_\kappa:=\nu_h+a\delta_0$. Associated with the function $\kappa$, a Riemannian
metric on $\DH$ is defined by
$$
\<X,\Omega_\sigma^\kappa(Y)\>_\HS,\qquad X,Y\in\BH_{\sa}^0,\ \sigma\in\DH,
$$
where
\begin{equation}\label{F-2.11}
\Omega_\sigma^\kappa:=R_{\sigma^{-1}}\kappa(L_\sigma R_{\sigma^{-1}}).
\end{equation}
This class of Riemannian metrics are called {\it monotone metrics} since the class was characterized by Petz \cite{Petz_monotone} with the monotonicity property
$$
\bigl\<\Phi(X),\Omega_{\Phi(\sigma)}^\kappa(\Phi(X))\bigr\>_\HS
\le\<X,\Omega_\sigma^\kappa(X)\>_\HS,\qquad X\in\BH_{\sa}^0,\ \sigma\in\DH,
$$
for every trace-preserving map $\Phi:\,\BH\to\BK$ such that $\Phi^*$ is a Schwarz contraction.
See also \cite{HR} for monotone Riemannian metrics. The description of
$\Omega_\sigma^\kappa$ in \eqref{F-2.11} is from \cite{HR}, that coincides with
$\bigl[f(L_\sigma R_\sigma^{-1})R_\sigma\bigr]^{-1}$ in Petz' representation in
\cite[Theorem 5]{Petz_monotone} for an operator monotone function $f(x)=1/\kappa(x)$, $x>0$,
and the condition $x\kappa(x)=\kappa(x^{-1})$, $x>0$, is equivalent to $f=\wtilde f$.

\subsection{Positive maps}

For a linear map $\map:\,\B(\hil)\to\B(\kil)$, where $\hil$ and $\kil$ are finite-dimensional
Hilbert spaces, the {\it adjoint} map $\map^*:\BK\to\BH$ is defined in terms of the
Hilbert-Schmidt inner products as
\begin{align*}
\<\Phi(X),Y\>_{\HS}=\<X,\Phi^*(Y)\>_{\HS},\qquad X\in\BH,\ Y\in\BK.
\end{align*}
The map $\Phi$ is said to be {\it positive} if $\Phi(A)\in\BK_+$ for all $A\in\BH_+$, and
\ki{$n$-positive}, for some $n\in\bN$, if 
$\id_n\otimes\map:\,\B(\bC^n)\otimes\B(\hil)\to\B(\bC^n)\otimes\B(\kil)$ is positive, where
$\id_n$ is the identity map on $\B(\bC^n)$. 
A map $\map$ is said to be \ki{completely positive} if it is $n$-positive for all $n\in\bN$. It
is easy to see that $\map$ is $n$-positive if and only if $\map^*$ is $n$-positive, and 
$\map$ is \ki{trace-preserving} (i.e., $\Tr\map(X)=\Tr X,\,X\in\B(\hil)$) if and only if $\map^*$
is \ki{unital} (i.e., $\map^*(I_{\kil})=I_{\hil}$).
A trace-preserving completely positive (CPTP) map is called a \ki{quantum channel} (or simply a channel).
We say that a positive map $\map$ is \ki{bistochastic} if it is both unital and trace-preserving.
The following is from \cite[Theorem 2.1]{Choi_Schwarz}:

\begin{lemma}\label{Choi inequality}
Let $\map:\,\B(\hil)\to\B(\kil)$ be a unital positive linear map, let $A\in\B(\hil)$ be
self-adjoint, and $f$ be an operator convex function defined on an interval containing
$\spec(A)$. Then 
\begin{align*}
f\bz\map(A)\jz\le\map\bz f(A)\jz.
\end{align*}
\end{lemma}
\medskip

The \ki{multiplicative domain} $\M_{\map}$ of a linear map $\map:\,\B(\hil)\to\B(\kil)$ is
defined as
\begin{align}\label{multdom}
\M_{\map}:=\left\{X\in\B(\hil):\,\map(XY)=\map(X)\map(Y),\s \map(YX)=\map(Y)\map(X),\,Y\in\B(\hil)\right\}.
\end{align}
Obviously, $\M_{\map}$ is an algebra, and if $\map$ is positive then it is also closed under the adjoint, and the restriction of $\map$ onto $\M_{\map}$ is a $^*$-homomorphism. In particular,
we have the following:
\begin{lemma}\label{lemma:mult domain fcalculus}
For any unital positive map $\map$ and any normal element $A$ in $\M_{\map}$, $\map(A)$ is also
normal, and for any function $\ffi$ on $\spec(A)\cup\spec(\map(A))$, we have
\begin{align*}
\ffi(\map(A))=\map(\ffi(A)).
\end{align*}
\end{lemma}

We say that a linear map $\map:\,\B(\hil)\to\B(\kil)$ is a \ki{Schwarz contraction} if it satisfies the \ki{Schwarz inequality}
\begin{align*}
\map(X)^*\map(X)\le\map(X^*X),\ds\ds\ds X\in\B(\hil).
\end{align*}
Obviously, every Schwarz contraction is positive, and it is known that every unital $2$-positive map is a Schwarz contraction, while the converse is not true.
If $\map$ is a Schwarz contraction, then its multiplicative domain can be characterized as 
\begin{align}\label{multdom2}
\M_{\map}=\left\{X\in\B(\hil):\,\map(XX^*)=\map(X)\map(X)^*,\s
\map(X^*X)=\map(X)^*\map(X)\right\};
\end{align}
see \cite[Lemma 3.9]{HMPB} for a proof.

The \ki{fixed point set} $\fix{\map}$ of a linear map $\map:\,\B(\hil)\to\B(\hil)$ is defined as 
\begin{align*}
\fix{\map}:=\left\{X\in\B(\hil):\,\map(X)=X\right\}.
\end{align*}
The same proof as that  of, e.g., \cite[Lemma 3.4]{BJKW} or \cite[Theorem 1\,(i)]{Je} yields the
following:
\begin{lemma}\label{lemma:fix-mult}
Let $\map:\,\B(\hil)\to\B(\hil)$ be a Schwarz contraction. If $\fix{\map^*}$ contains an element
of $\B(\hil)_{++}$, then $\fix{\map}$ is a $C^*$-subalgebra of $\M_{\map}$.
\end{lemma}

\begin{remark}
In general, $\fix{\map}$ need not be an algebra, and there is no inclusion between $\fix{\map}$ and $\M_{\Phi}$ in either direction. 
We give some examples illustrating these in Appendix \ref{sec:examples} and Example \ref{E-4.3}.
\end{remark}

\section{The standard and the maximal $f$-divergences}
\label{sec:fdiv}

\subsection{Introduction to $f$-divergences}
\label{sec:fdiv intro}

Given two probability density functions (or, more generally, positive functions) $\rho,\sigma$
on a finite set $\X$, their $f$-divergence $S_f(\rho\|\sigma)$, corresponding to a convex
function $f:\,(0,+\infty)\to\bR$, was defined by Csisz\'ar \cite{Csiszar_fdiv} as 
\begin{align}\label{cl fdiv def1}
S_f(\rho\|\sigma):=\sum_{x\in\X}\sigma(x)f\bz\frac{\rho(x)}{\sigma(x)}\jz.
\end{align}
(For simplicity, in this section we assume that both $\rho$ and $\sigma$ are strictly positive, whether they denote functions or operators.) Most divergence measures used in classical information theory can be written in this form; for instance, $f(t):=t\log t$ yields the relative entropy (Kullback-Leibler divergence), 
$f_{\alpha}(t):=\mathrm{sgn}(\alpha-1)t^{\alpha},\,\alpha\in(0,+\infty)\setminus\{1\}$, correspond to the R\'enyi divergences, and 
$f(t):=|t-1|$ gives the variational distance. All $f$-divergences are easily seen to be jointly convex in their variables, and monotone non-increasing under the joint action of a stochastic map on their arguments. Moreover, when $f$ is strictly convex, a stochastic map preserves the $f$-divergence of 
$\rho$ and $\sigma$ if and only if it is reversible on $\{\rho,\sigma\}$, i.e., there exists a stochastic map $\Psi$ such that 
$\Psi(\map(\rho))=\rho$ and $\Psi(\map(\sigma))=\sigma$ (see, e.g., \cite[Proposition A.3]{HMPB}).

To motivate the definition of the different quantum $f$-divergences, let us recall the GNS representation theorem, that says that for every 
positive linear functional $\sigma$ on a $C^*$-algebra $\A$, there exists a Hilbert space $\hil_{\sigma}$, a vector $\Omega_{\sigma}\in\hil$, and 
a representation $\pi_{\sigma}$ of $\A$ on $\hil$ such that $\sigma(a)=\inner{\Omega_{\sigma}}{\pi_{\sigma}(a)\Omega_{\sigma}}$ for all $a\in\A$.
In the classical case described above, $\rho$ and $\sigma$ define positive linear functionals on the commutative $C^*$-algebra $\bC^{\X}$,
which we denote by the same symbols, and GNS representations can be given by choosing $\hil=l^2(\X)$ (with respect to the counting measure), 
$\Omega_{\rho}=(\sqrt{\rho(x)})_{x\in\X},\,\Omega_{\sigma}=(\sqrt{\sigma(x)})_{x\in\X}$, and 
$\pi(a):=M_a:\,b\mapsto ab$ (with pointwise multiplication) for any $a,b\in\bC^{\X}$. Then the  operator
$S:=M_{\rho^{1/2}\sigma^{-1/2}}$ changes the representing vector of $\sigma$ to that of $\rho$, i.e., 
$S\Omega_{\sigma}=\Omega_{\rho}$, and we have 
\begin{align*}
S_f(\rho\|\sigma)=\inner{\Omega_{\sigma}}{f(\Delta_{\rho/\sigma})\Omega_{\sigma}},
\end{align*}
where $\Delta_{\rho/\sigma}:=SS^*=S^*S=M_{\rho/\sigma}$ is the Radon-Nikodym derivative. This reformulation of \eqref{cl fdiv def1} will be useful to extend the notion of $f$-divergences to the quantum setting.

In the general finite-dimensional case, when $\A\subset\B(\hil)$ for some finite-dimensional Hilbert space $\hil$, positive linear functionals can be identified with positive elements of $\A$ through $\rho(a)=\Tr D_{\rho }a$, where $D_{\rho}$ is the density operator of $\rho$. For the rest, we will 
use the same notation $\rho$ also for its density operator. Given two positive operators $\rho,\sigma\in\A$ (we assume again for simplicity that they are both invertible), the GNS representations can be given by choosing $\hil:=(\A,\inner{.}{.}_{\HS})$, $\Omega_{\rho}:=\rho^{1/2}$, $\Omega_{\sigma}:=\sigma^{1/2}$, and $\pi(a):=L_a:\,b\mapsto ab$, $a,b\in\A$.
The question is now how to define the Radon-Nikodym derivative, i.e., the non-commutative analogues of the operators $S$ and $\Delta_{\rho/\sigma}$.
One option is to choose $S:=L_{\rho^{1/2}}R_{\sigma^{-1/2}}$, so that $\Delta_{\rho/\sigma}:=SS^*=S^*S=L_{\rho}R_{\sigma\inv}$ becomes the \ki{relative modular operator}. The corresponding quantum $f$-divergence is 
\begin{align}\label{qe def}
S_f(\rho\|\sigma)&:=\Tr\sigma^{1/2}f\bz L_{\rho}R_{\sigma\inv}\jz\sigma^{1/2}
=\inner{I}{\per{f}\bz L_{\rho},R_{\sigma}\jz I}_{\HS},
\end{align}
that was defined and investigated by Petz  (in a more general form) under the name quasi-entropy \cite{P85,P86}. Note that the choice
$S:=L_{\sigma^{-1/2}}R_{\rho^{1/2}}$ results in the same expression.
Petz' analysis was extended in \cite{HMPB}, and we give further extensions in Section \ref{Petz f div} below.

Another option is to choose $S:=R_{\sigma^{-1/2}\rho^{1/2}}$, and $\Delta_{\rho/\sigma}:=SS^*=R_{\sigma^{-1/2}\rho\sigma^{-1/2}}$ (the so-called \ki{commutant Radon-Nikodym derivative}), resulting in the $f$-divergence
\begin{align}\label{maxdiv def}
\maxdiv{f}(\rho\|\sigma)&:=\Tr\sigma^{1/2}f\bz \sigma^{-1/2}\rho\sigma^{-1/2}\jz\sigma^{1/2}
=\inner{I}{\per{f}(\rho,\sigma)I}_{\HS}.
\end{align} 
A special case of this, corresponding to the function $f(t):=t\log t$, has been studied by Belavkin and Staszewski \cite{BS} as a quantum extension of the Kullback-Leibler divergence. The above general form was introduced in \cite{PR}. Matsumoto \cite{Ma} showed that this $f$-divergence is maximal among the monotone quantum $f$-divergences, and analyzed the preservation of this $f$-divergence by quantum operations. We will review and extend some of his results in Sections \ref{sec:maximal f-div} and \ref{sec:comparison}.
Note that the definitions $S:=L_{\rho^{1/2}\sigma^{-1/2}},\,\Delta_{\rho/\sigma}:=S^*S$;
$S:=R_{\rho^{1/2}\sigma^{-1/2}},\,\Delta_{\rho/\sigma}:=S^*S$; and
$S:=L_{\sigma^{-1/2}\rho^{1/2}},\,\Delta_{\rho/\sigma}:=SS^*$ 
all result in the same $f$-divergence (although with the latter two $S\Omega_{\sigma}=\Omega_{\rho}$ does not hold).

Another natural definition would be to choose $S:=R_{\sigma^{-1/2}\rho^{1/2}}$ and
$\Delta_{\rho/\sigma}:=S^*S$, leading to the $f$-divergence
\begin{align}\label{bad fdiv def}
\widetilde S_f(\rho\|\sigma):=\Tr\sigma^{1/2}f\bz \rho^{1/2}\sigma^{-1}\rho^{1/2}\jz\sigma^{1/2}.
\end{align}
In general, however, $\widetilde S_f$, unlike the other two versions $S_f$ and  $\widehat S_f$ above, 
is not monotone under CPTP maps, nor it is jointly convex in its arguments, 
as we show in Appendix \ref{sec:counterex}.
Thus, $\widetilde S_f$ is not a proper quantum divergence for general operator convex functions $f$, and hence we don't consider this version further in the paper.

A different and more operational approach is to define quantum $f$-divergences directly from classical ones.
There seems to be two natural ways to do so, namely, to consider the \ki{maximal $f$-divergence}, introduced by Matsumoto \cite{Ma} as
\begin{align}\label{maxdiv def2}
S_f^{\max}(\rho\|\sigma):=\inf\{&S_f(p\|q):\,p,q\in\B(\kil)_+\text{ are commuting, $\dim\kil<+\infty$, and }\\
&\ds\map(p)=\rho,\,\map(q)=\sigma\ \text{for some CPTP map}\ \map:\,\B(\kil)\to\B(\hil)\}
\nonumber
\end{align} 
(denoted by $D_f^{\max}$ in \cite{Ma}) and the \ki{measured} (or \ki{minimal}) \ki{$f$-divergence}
\begin{align}\label{meas fdiv def}
S_f^{\min}(\rho\|\sigma):=S_f^{\meas}(\rho\|\sigma):=
\sup\{&S_f(\map(\rho)\|\map(\sigma)):\,\map:\,\B(\hil)\to\B(\kil)\ \text{is CPTP}, \\
&\ds\dim\kil<+\infty,\ \text{and $\ran\map$ is commutative}\}. \nonumber
\end{align}
For a given (convex) function $f:\,(0,+\infty)\to\bR$, we say that a functional
$S_f^q$ is a quantum $f$-divergence if $S_f^q$ assigns a number in $(-\infty,+\infty]$ to any
pair $(\rho,\sigma)\in\B(\hil)_+\times\B(\hil)_+$ for any finite-dimensional Hilbert space,
such that
if $\rho$ and $\sigma$ commute then 
$S_f^q(\rho\|\sigma)=S_f(\{\rho(x)\}_{x\in\X}\|\{\sigma(x)\}_{x\in\X})$, where
$\{\rho(x)\}_{x\in\X}$ and $\{\sigma(x)\}_{x\in\X}$ are the diagonal elements of $\rho$ and
$\sigma$ in  an orthonormal basis in which both of them are diagonal. 
We say that $S_f^q$ is monotone if it is monotone non-increasing under the action of CPTP maps
on both arguments of $S_f^q$. It is clear from the above definitions that 
\begin{align}\label{min-max properties}
S_f^{\min}(\rho\|\sigma)\le S_f^q(\rho\|\sigma)\le S_f^{\max}(\rho\|\sigma).
\end{align}
for any monotone quantum $f$-divergence $S_f^q$, which explains the names ``maximal'' and ``minimal''
for the definitions in \eqref{maxdiv def2} and \eqref{meas fdiv def}.

Matsumoto has shown that $S_f^{\max}(\rho\|\sigma)=\maxdiv{f}(\rho\|\sigma)$ for operator convex function
$f$ on $[0,+\infty)$, and for $\rho,\sigma$ such that $\rho^0\le\sigma^0$. 
For $S_f^{\meas}(\rho\|\sigma)$, no explicit general formula is known. We will analyze the relation of the 
$f$-divergences $\maxdiv{f}= S_f^{\max}$, $S_f$, and $S_f^{\meas}$ in Section \ref{sec:comparison}.

\subsection{Standard $f$-divergences}
\label{Petz f div}

Petz originally introduced his quasi-entropies \cite{P85,P86} by a more general formula than
\eqref{qe def}, as
\begin{align*}
S_f^K(\rho\|\sigma):=\<K\sigma^{1/2},f(L_\rho R_{\sigma\inv})(K\sigma^{1/2})\>_\HS
=\Tr\sigma^{1/2}K^* f\bz L_{\rho}R_{\sigma\inv}\jz(K\sigma^{1/2}),
\end{align*}
with $K$ an arbitrary operator, and $\sigma$ invertible. He proved the monotonicity
$$
S_f^K(\Phi(\rho)\|\Phi(\sigma))\le S_f^{\Phi^*(K)}(\rho\|\sigma)
$$
of these quantities under the joint action of the dual of unital Schwarz contractions for
operator monotone decreasing $f$ on $[0,+\infty)$ with $f(0)\le0$, and under the restriction onto a subalgebra for operator convex $f$.
His definition and results were extended in the $K=I$ case in \cite{HMPB}, in particular, for general positive operators $\rho,\sigma$. 

Below we give some further extensions, by only requiring the function $f$ to be defined on $(0,+\infty)$ (as opposed to $[0,+\infty)$ in \cite{HMPB}), while allowing the operators $\rho$ and $\sigma$ to have arbitrary supports.
Recall our convention stated in the first paragraph of Section \ref{sec:op convex}, that $f:\,(0,+\infty)\to\bR$ is a continuous function such that the limits 
$f(0^+):=\lim_{x\searrow 0}f(x)$ and $f'(+\infty):=\lim_{x\to+\infty}\frac{f(x)}{x}$ exist and their non-negative linear combinations make sense.

\begin{definition}\label{D-3.1}\rm
For $\rho,\sigma\in\BH_+$ let $\rho=\sum_{a\in\Sp(\rho)}aP_a$ and $\sigma=\sum_{b\in\Sp(\sigma)}bQ_b$ be the spectral
decompositions. When $\rho,\sigma>0$, we have
$$
f(L_\rho R_{\sigma^{-1}})=\sum_{a\in\Sp(\rho)}\sum_{b\in\Sp(\sigma)}f(ab^{-1})L_{P_a}R_{Q_b},
$$
and we define the (standard) {\it $f$-divergence} of $\rho$ and $\sigma$ as
\begin{align}\label{fdiv def}
S_f(\rho\|\sigma):=\bigl\<\sigma^{1/2},f(L_\rho R_{\sigma^{-1}})\sigma^{1/2}\bigr\>_\HS
=\Tr \sigma^{1/2}f(L_\rho R_{\sigma^{-1}})(\sigma^{1/2}).
\end{align}
We extend $S_f(\rho\|\sigma)$ to general $\rho,\sigma\in\BH_+$ as
\begin{equation}\label{F-3.6}
S_f(\rho\|\sigma):=\lim_{\eps\searrow0}S_f(\rho+\eps I\|\sigma+\eps I).
\end{equation}
\end{definition}

\begin{prop}\label{P-3.2}
For every $\rho,\sigma\in\BH_+$ the limit in \eqref{F-3.6} exists, and we have
\begin{align}
S_f(\rho\|\sigma)
&=\sum_{a,b}P_f(a,b)\Tr P_aQ_b\label{standard fdiv1}\\
&=\sum_{a,b}P_f(a\Tr P_aQ_b,b\Tr P_aQ_b)\label{standard fdiv2}\\
&=\sum_{a>0}\sum_{b>0}bf(ab^{-1})\Tr P_aQ_b+f(0^+)\Tr(I-\rho^0)\sigma+f'(+\infty)\Tr \rho(I-\sigma^0)\label{F-3.7}
\end{align}
with the convention $(+\infty)0=0$. In particular, \eqref{F-3.6} coincides with \eqref{fdiv def} for invertible $\rho,\sigma$.
\end{prop}

\begin{proof}
Since $\rho+\eps I=\sum_a(a+\eps)P_a$ and $\sigma+\eps I=\sum_b(b+\eps)Q_b$, one has
$$
f(L_{\rho+\eps I}R_{(\sigma+\eps I)^{-1}})=\sum_{a,b}f((a+\eps)(b+\eps)^{-1})L_{P_a}R_{Q_b}
$$
so that
$$
S_f(\rho+\eps I\|\sigma+\eps I)=\sum_{a,b}(b+\eps)f((a+\eps)(b+\eps)^{-1})\Tr P_aQ_b.
$$
Using \eqref{persp ext},
one finds that
\begin{align*}
&\lim_{\eps\searrow0}S_f(\rho+\eps I\|\sigma+\eps I) \\
&\quad=\sum_{a,b}P_f(a,b)\Tr P_aQ_b\\
&\quad=\sum_{a,b>0}bf(ab^{-1})\Tr P_aQ_b+\sum_{b>0}bf(0^+)\Tr P_0Q_b
+\sum_{a>0}af'(+\infty)\Tr P_aQ_0 \\
&\quad=\sum_{a,b>0}bf(ab^{-1})\Tr P_aQ_b+f(0^+)\Tr(I-\rho^0)\sigma+f'(+\infty)\Tr \rho(I-\sigma^0),
\end{align*}
giving \eqref{standard fdiv1} and \eqref{F-3.7}. The equality of \eqref{standard fdiv1} and \eqref{standard fdiv2} is trivial.
\end{proof}

\begin{remark}
Note that the expression in \eqref{standard fdiv2} is the classical $f$-divergence \cite{Csiszar_fdiv} of the functions
$p(a,b):=a\Tr P_aQ_b$ and $q(a,b):=b\Tr P_aQ_b$, defined on $(\spec\rho)\times(\spec\sigma)$ (see \cite{HMPB} and 
\cite{NSz} for further details).
\end{remark}

\begin{cor}\label{cor:fdiv infty}
$S_f(\rho\|\sigma)=+\infty$ if and only if one of the following conditions holds:
\begin{enumerate}
\item
$f(0^+)=+\infty$ and $\sigma^0\nleq\rho^0$;
\item
$f'(+\infty)=+\infty$ and $\rho^0\nleq\sigma^0$.
\end{enumerate}
In all other cases, $S_f(\rho\|\sigma)$ is a finite number.
\end{cor}

\begin{example}\label{E-3.4}
The most relevant examples for applications are given by 
\begin{align*}
f_{\alpha}(x):=s(\alpha)x^{\alpha}\ds\text{for}\ds\alpha\in(0,+\infty),\ds\ds\ds\text{and}
\ds\ds\ds\eta(x):=x\log x,\ds\ds\ds\,x\ge 0,
\end{align*}
where $s(\alpha):=-1$ for
$0<\alpha<1$ and $s(\alpha):=1$ for $\alpha\ge1$. They give rise to 
\begin{align*}
S_{f_{\alpha}}(\rho\|\sigma)=
\begin{cases}
s(\alpha)\Tr\rho^{\alpha}\sigma^{1-\alpha},&\alpha\in(0,1]\text{ or }\rho^0\le\sigma^0,\\
+\infty,&\text{otherwise},
\end{cases}
\end{align*}
\begin{align}\label{Umegaki relentr}
S(\rho\|\sigma):=S_{\eta}(\rho\|\sigma)=\begin{cases}
\Tr\rho(\log\rho-\log\sigma),&\rho^0\le\sigma^0,\\
+\infty,&\text{otherwise},
\end{cases}
\end{align}
where $S(\rho\|\sigma)$ is the \ki{Umegaki relative entropy} \cite{Umegaki};
see \eqref{Umegaki intro}.
The quantities $S_{f_{\alpha}}$ define the standard \ki{R\'enyi divergences} as
\begin{align}\label{Renyi div as fdiv}
D_{\alpha}(\rho\|\sigma)
:=\frac{1}{\alpha-1}\log\bigl(s(\alpha)S_{f_{\alpha}}(\rho\|\sigma)\bigr)-\frac{1}{\alpha-1}\log\Tr\rho,
\qquad\alpha\in(0,+\infty)\setminus\{1\};
\end{align}
see \eqref{conventional Renyi intro}.
It is easy to see (by simply computing its second derivative)
that $\alpha\mapsto\log\bz s(\alpha)S_{f_{\alpha}}(\rho\|\sigma)\jz$ is convex, and hence $\alpha\mapsto D_{\alpha}(\rho\|\sigma)$ is 
increasing for any fixed $\rho,\sigma$; moreover,
\begin{align}\label{Renyi limit}
\lim_{\alpha\to 1}D_{\alpha}(\rho\|\sigma)=\sup_{\alpha\in(0,1)}D_{\alpha}(\rho\|\sigma)=\frac{1}{\Tr\rho}\,S(\rho\|\sigma).
\end{align}
(Although the function $f_\alpha$ is operator convex on $[0,\infty)$ only for
$0<\alpha\le2$, we shall use $S_{f_\alpha}$ for all $\alpha>0$. See also Example \ref{E-4.3}
below.)
\end{example}

\begin{remark}
In \cite{HMPB}, we assumed that $f$ is defined on $[0,+\infty)$, and we
defined $S_f(\rho\|\sigma)$ first for an invertible $\sigma$ as in \eqref{fdiv def}, and extended to non-invertible $\sigma$ as
$S_f(\rho\|\sigma):=\lim_{\ep\searrow 0}S_f(\rho\|\sigma+\ep I)$, which is slightly different from the above \eqref{F-3.6}.
However, when
$f(0^+)<+\infty$ so that $f$ can be extended to a continuous function on
$[0,+\infty)$, we see by expression \eqref{F-3.7} that the present definition is the same as
that in \cite[Definition 2.1]{HMPB}.
The extension of $S_f(\rho\|\sigma)$ to functions $f$ without the assumption $f(0^+)<+\infty$ 
is relevant, for instance, to the following symmetry property.
\end{remark} 

\begin{prop}\label{P-3.4}
Let $\widetilde f$ be the transpose of $f$. Then for every $\rho,\sigma\in\BH_+$,
$$
S_{\widetilde f}(\rho\|\sigma)=S_f(\sigma\|\rho).
$$
\end{prop}

\begin{proof}
The assertion follows immediately from expression \eqref{F-3.7} together with \eqref{F-2.1},
since $b\widetilde f(ab^{-1})=af(ba^{-1})$ for $a,b>0$.
\end{proof}

The next proposition shows that the continuity property that is incorporated in definition \eqref{F-3.6} can be extended to the case where the perturbation is not a constant multiple of the identity, but an arbitrary positive operator. This becomes important, for instance, when one studies the behavior of the $f$-divergences under the action of stochastic maps, in which case one might need to evaluate expressions like 
\begin{align*}
\lim_{\ep\searrow 0}S_f\bz \map(\rho+\ep I)\|\map(\sigma+\ep I)\jz=\lim_{\ep\searrow 0}S_f\bz \map(\rho)+\ep \map(I)\|\map(\sigma)+\ep \map(I)\jz,
\end{align*}
which does not reduce to \eqref{F-3.6} unless $\map$ is unital.

\begin{prop}\label{prop:standard fdiv cont}
Let $\rho,\sigma\in\BH_+$.
\begin{enumerate}
\item\label{fdiv cont i} 
Assume that both $f(0^+)$ and $f'(+\infty)$ are finite. Then
$$
S_f(\rho\|\sigma)=\lim_{n\to\infty}S_f(\rho_n\|\sigma_n)
$$
for any choice of sequences $\rho_n,\sigma_n\in\BH_+$ such that $\rho_n\to\rho,\sigma_n\to0$ as $n\to+\infty$.
\item\label{fdiv cont ii} 
Let $f$ be an operator convex function on $(0,+\infty)$ (with no restriction on
$f(0^+)$ and $f'(+\infty)$). Then
$$
S_f(\rho\|\sigma)=\lim_{n\to\infty}S_f(\rho+L_n\|\sigma+L_n)
$$
for any choice of a sequence $L_n\in\BH_+$ such that $\rho+L_n,\sigma+L_n>0$ for every $n$, and
$L_n\to0$ as $n\to+\infty$.
\end{enumerate}
\end{prop}

We give the proof of the above proposition, and further observations about the continuity properties of the standard $f$-divergences, in Appendix \ref{sec:cont}. We remark that in the proof of \ref{fdiv cont ii} of the above proposition, we will use the joint convexity property given in Proposition \ref{P-3.6} below.

\begin{remark}
Note that \ref{fdiv cont i} of the above proposition can be reformulated as follows: When $f$ is a continuous function on 
$(0,+\infty)$ such that both $f(0^+)$ and $f'(+\infty)$ are finite, then 
\begin{align*}
(\rho,\sigma)\mapsto S_f(\rho\|\sigma)\ds\ds\text{is continuous on}\ds\ds \B(\hil)_+\times\B(\hil)_+.
\end{align*}
\end{remark}

\medskip

The most important properties of $f$-divergences are their joint convexity and monotonicity
under stochastic maps when $f$ is operator convex. These properties follow immediately from the results of \cite{P86,HMPB}, 
even though our definition of $f$-divergences in this paper is slightly 
more general than in \cite{P86,HMPB}.

\begin{prop}\label{P-3.6}
Let $f:\,(0,+\infty)\to\bR$ be operator convex.
$S_f(\rho\|\sigma)$ is jointly convex in $\rho,\sigma\in\BH_+$, i.e., for every
$\rho_i,\sigma_i\in\BH_+$ and $\lambda_i\ge0$ for $1\le i\le k$,
\begin{align}\label{convexity}
S_f\Biggl(\sum_{i=1}^k\lambda_i\rho_i\Bigg\|\sum_{i=1}^k\lambda_i\sigma_i\Biggr)
\le\sum_{i=1}^k\lambda_iS_f(\rho_i\|\sigma_i).
\end{align}
\end{prop}
\begin{proof}
Immediate from \cite[Corollary 4.7]{HMPB} and definition \eqref{F-3.6}.
\end{proof}

\begin{remark}\label{rem:subadd}
It is clear from \eqref{F-3.7} that the $f$-divergences have the homogeneity property 
\begin{align*}
S_f(\lambda\rho\|\lambda\sigma)=\lambda S_f(\rho\|\sigma),\ds\ds\ds \lambda\ge0,\ds\rho,\sigma\in\B(\hil)_+.
\end{align*}
Hence, \eqref{convexity} is equivalent to the joint subadditivity
\begin{align*}
S_f\Biggl(\sum_{i=1}^k\rho_i\Bigg\|\sum_{i=1}^k\sigma_i\Biggr)
\le\sum_{i=1}^k S_f(\rho_i\|\sigma_i).
\end{align*}
In particular, it is not necessary that the $\lambda_i$'s sum up to $1$ in \eqref{convexity}.
\end{remark}

The monotonicity property of $f$-divergences, first shown by Petz \cite{P86} in a somewhat
restricted setting, was later extended in various ways, e.g., in \cite{LR,TCR,HMPB}. The
following is an easy adaptation of  \cite[Theorem 4.3]{HMPB} to the present setting. 

\begin{prop}\label{P-3.8}
Let $\Phi:\BH\to\BK$ be a trace-preserving linear map such that the adjoint $\Phi^*$ is a
Schwarz contraction (see Section 2.5). Then for every $\rho,\sigma\in\BH_+$, and every operator
convex function $f:\,(0,+\infty)\to\bR$,
\begin{align}\label{fdiv monotonicity}
S_f(\Phi(\rho)\|\Phi(\sigma))\le S_f(\rho\|\sigma).
\end{align}
\end{prop}

\begin{proof}
For $\eps>0$ let $f_\eps(x):=f(x+\eps)$, $x\ge0$. By \cite[Theorem 4.3]{HMPB} one has
$$
S_{f_\eps}(\Phi(\rho)\|\Phi(\sigma))\le S_{f_\eps}(\rho\|\sigma).
$$
Thanks to expression \eqref{F-3.7} it is straightforward to see that
$$
\lim_{\eps\searrow0}S_{f_\eps}(\rho\|\sigma)=S_f(\rho\|\sigma),
$$
and similarly $\lim_{\eps\searrow0}S_{f_\eps}(\Phi(\rho)\|\Phi(\sigma))=S_f(\Phi(\rho)\|\Phi(\sigma))$, so the
assertion follows.
\end{proof}

\begin{remark}\label{rem:con=mono}
As observed in \cite{LR} (more explicitly, in \cite[Appendix A]{TCR} and
\cite[Proposition E.2]{HP2}), it is known that for a general continuous function $f$ on
$(0,+\infty)$,  the $f$-divergence $S_f$ has the joint convexity property in Proposition
\ref{P-3.6} if and only if it has the monotonicity property under CPTP maps. Indeed, this fact
holds true for different types of quantum divergences; for example, the proof of the monotonicity
under CPTP maps for $D_{\alpha,z}$ given in \eqref{F-1.2} can be reduced to that of the joint
convexity/concavity of $(\rho,\sigma)\mapsto
\Tr(\sigma^{1-\alpha\over2z}\rho^{\alpha\over z}\sigma^{1-\alpha\over2z})^z$
(see \cite{FL13,AD}).
\end{remark}

\begin{remark}\label{rem:positive mon}
It is not known whether in Proposition \ref{P-3.8}, the assumption that $\Phi^*$ is a Schwarz
contraction can be weakened to simply requiring that $\map$ is positive. 
A non-trivial example is when $f(x):=f_2(x):=x^2$, giving the $f$-divergence
$S_{f_2}(\rho\|\sigma)=\Tr\rho^2\sigma\inv$.
Monotonicity of this $f$-divergence under trace-preserving positive maps is a consequence of a
stronger operator inequality (see, e.g., \cite[Lemma 3.5]{HMPB}).
Alternatively, this follows from the more general statement in Corollary \ref{P-3.20}, by noting
that $S_{f_2}=\maxdiv{f_2}$ (see Example \ref{E-3.18}).
More importantly, it has been pointed out recently in \cite{M-HR} that Beigi's proof for the
monotonicity of the sandwiched R\'enyi divergences \cite{Beigi} yields that the Umegaki
relative entropy \eqref{Umegaki relentr} is monotone under trace-preserving positive maps.
\end{remark}

As with any inequality, it is natural to ask when \eqref{fdiv monotonicity} holds with equality. This problem was first addressed by Petz, who 
considered it in the more general von Neumann algebraic framework \cite{Petz_sufficiency}. When translated to our finite-dimensional setting, his result, given in \cite[Theorem 3]{Petz_sufficiency},
 says that 
for a $2$-positive and trace-preserving $\map:\,\B(\hil)\to\B(\kil)$, and $\rho,\sigma\in\B(\hil)_{++}$,
\begin{align}\label{Petz reversibility}
S_{f_{1/2}}(\map(\rho)\|\map(\sigma))=S_{f_{1/2}}(\rho\|\sigma)\ds\iff\ds
\map_{\sigma}^*(\map(\rho))=\rho,
\end{align}
where $f_{1/2}(x):=-x^{1/2}$ with the corresponding $f$-divergence
$S_{f_{1/2}}(\rho\|\sigma)=-\Tr\rho^{1/2}\sigma^{1/2}$, and $\map_{\sigma}^*$ is the
adjoint of the map $\Phi_\sigma:\,\B(\hil)\to\B(\kil)$ defined by
\begin{equation}\label{F-3.13}
\map_{\sigma}(X)=\map(\sigma)^{-1/2}\map\bz\sigma^{1/2}X\sigma^{1/2}\jz\map(\sigma)^{-1/2},
\ds\ds\ds X\in\B(\hil).
\end{equation}
More explicitly, $\Phi_\sigma^*:\,\B(\kil)\to\B(\hil)$ is given as
\begin{align}\label{rverse-map}
\map_{\sigma}^*(Y):=\sigma^{1/2}\map^*\bz\map(\sigma)^{-1/2}Y\map(\sigma)^{-1/2}\jz
\sigma^{1/2},\ds\ds\ds Y\in\B(\kil).
\end{align}
Since it is easy to check that $\map_{\sigma}^*(\map(\sigma))=\sigma$, the second condition in
\eqref{Petz reversibility} yields the reversibility 
of $\map$ in the sense defined below, while reversibility implies the first condition in
\eqref{Petz reversibility} by a double application of the monotonicity inequality 
\eqref{fdiv monotonicity}.

By comparing (iii) of \cite[Theorem 3]{Petz_sufficiency} with (i) of \cite[Theorem 3.1]{mon_revisited}, one sees that the conditions in \eqref{Petz reversibility}
are further equivalent to the preservation of the Umegaki relative entropy
\begin{align*}
S(\map(\rho)\|\map(\sigma))=S(\rho\|\sigma).
\end{align*}
Moreover, it was stated in \cite[Theorem 2]{JP_survey} (albeit with an incorrect formulation and without a proof) that \eqref{Petz reversibility}
is also equivalent to the preservation of the $f_{\alpha}$-divergences for $0<\alpha<1$, where $f_{\alpha}(x):=x^{\alpha}$.

\begin{remark}
The notation of \cite{Petz_sufficiency,JP,JP_survey} corresponds to ours as 
\begin{align*}
\phi(\cdot)=\Tr\rho(\cdot),\ds\ds
\omega(\cdot)=\Tr\sigma(\cdot),\ds\ds
\alpha=\map^*,\ds\ds
\alpha_{\omega}^*=\map_{\sigma},
\end{align*}
where the first expressions are always from \cite{Petz_sufficiency}, and the second expressions
are our notations.
We remark that (v) and (vi) of \cite[Theorem 3]{Petz_sufficiency} are incorrectly stated
as  $\phi\circ\alpha_{\omega}^*=\phi$ and $\omega\circ\alpha_{\phi}^*=\omega$, respectively; they should be 
$\phi\circ\alpha\circ\alpha_{\omega}^*=\phi$ and $\omega\circ\alpha\circ\alpha_{\phi}^*=\omega$. This correction was given, e.g., in \cite[Theorem 3]{JP}.
\end{remark}

\begin{definition}\label{def:rev}
Let $\map:\,\B(\hil)\to\B(\kil)$ be a trace-preserving positive linear map and $\rho,\sigma\in\B(\hil)_+$. We say that $\Phi$ is {\it reversible} on the pair $\rho,\sigma$ if there
exists a trace-preserving positive linear map $\mapp:\,\B(\kil)\to\B(\hil)$ such that 
$$
\mapp(\map(\rho))=\rho,\ds\ds\ds
\mapp(\map(\sigma))=\sigma.
$$
\end{definition}

\begin{remark}
(1) Note that we only assume positivity of the reverse map $\mapp$ in the above definition, irrespective of the type of positivity of the map $\map$. The reason for this becomes clear from 
\ref{rev cond1}\,$\iff$\,\ref{rev cond2}\,$\iff$\,\ref{rev cond3} in 
Theorem \ref{T-3.12}, where we see that the reversibility condition for
$\map$ on the pair $\rho,\sigma$ is independent of the choice of the type of positivity for
the reverse map; the reversibility conditions with a simply positive reverse map and with a
completely positive one are equivalent.

(2) Note that the right-hand side of \eqref{Petz reversibility} states reversibility with the
reverse map $\map_{\sigma}^*$, except that $\Phi_\sigma^*$ is not necessarily trace-preserving
on the whole $\B(\kil)$. However, its restriction to
$\Phi(\sigma)^0\B(\kil)\Phi(\sigma)^0=\B(\Phi(\sigma)^0\kil)$ is trace-preserving, since
$\Phi_\sigma$ is unital as a map from $\B(\hil)$ to $\B(\Phi(\sigma)^0\kil)$, and it is
easy to extend $\Phi_\sigma^*|_{\Phi(\sigma)^0\B(\kil)\Phi(\sigma)^0}$ to a trace-preserving map
on $\B(\kil)$. We will benefit from this observation in the proof of
\ref{rev cond2}\,$\imp$\,\ref{rev cond3} of Theorem \ref{T-3.12}.

(3) It is easy to see that if $\Phi$ is $n$-positive for some $n\in\bN$ then so is $\Phi_\sigma^*$.
However, if $\map^*$ is a Schwarz contraction, that need not imply that $\map_{\sigma}$ is a 
Schwarz contraction, as was pointed out in \cite[Proposition 2]{Je}.
\end{remark}

A systematic study of the relation between reversibility and the preservation of
$f$-divergences was carried out in \cite{HMPB}, complemented later in \cite{Je} with some
further results. We summarize these results and give some slight extensions and modifications
in the following theorem.

\begin{thm}\label{T-3.12}
Let $\rho,\sigma\in\B(\hil)_+$ be such that $\rho^0\le\sigma^0$, and let $\map:\,\B(\hil)\to\B(\kil)$ be a $2$-positive trace-preserving linear map. Then the following 
\ref{rev cond1}--\ref{rev cond9} are equivalent:
\begin{enumerate}

\item \label{rev cond1}
$\map$ is reversible on $\{\rho,\sigma\}$ in the sense of Definition \ref{def:rev}, i.e.,
there exists a trace-preserving positive map 
$\Psi:\,\B(\kil)\to\B(\hil)$ such that $\Psi(\Phi(\rho))=\rho$,
$\Psi(\Phi(\sigma))=\sigma$.

\item \label{rev cond2}
There exists a trace-preserving map $\Psi:\,\B(\kil)\to\B(\hil)$ such that $\Psi^*$ satisfies the Schwarz inequality and $\Psi(\Phi(\rho))=\rho$,
$\Psi(\Phi(\sigma))=\sigma$.

\item \label{rev cond3}
There exist CPTP maps $\widetilde\Phi:\,\B(\hil)\to\B(\kil)$ and
$\widetilde\Psi:\,\B(\kil)\to\B(\hil)$ such that
$\widetilde\Phi(\rho)=\Phi(\rho)$, $\widetilde\Phi(\sigma)=\Phi(\sigma)$ and
$\widetilde\Psi(\Phi(\rho))=\rho$, $\widetilde\Psi(\Phi(\sigma))=\sigma$.

\item \label{rev cond4}
$S_f(\Phi(\rho)\|\Phi(\sigma))=S_f(\rho\|\sigma)$ for some operator convex function $f$ on
$(0,+\infty)$ such that $f(0^+)<+\infty$ and
\begin{equation}\label{F-3.15}
|\supp\mu_f|\ge
\big|\Sp\bigl(L_\rho R_{\sigma^{-1}}\bigr)\cup\Sp\bigl(L_{\Phi(\rho)}R_{\Phi(\sigma)^{-1}}\bigr)\big|,
\end{equation}
where $\mu_f$ is the measure from the integral representation given in \eqref{F-2.4}.

\item \label{rev cond5}
$S_f(\Phi(\rho)\|\Phi(\sigma))=S_f(\rho\|\sigma)$ for all operator convex functions $f$ on
$[0,+\infty)$.

\item \label{rev cond6}
 $\sigma^0\Phi^*(\Phi(\sigma)^{-z}\Phi(\rho)^{2z}\Phi(\sigma)^{-z})\sigma^0=\sigma^{-z}\rho^{2z} \sigma^{-z}$ for all $z\in\bC$.

\item\label{rev cond7}
 $\sigma^0\Phi^*(\Phi(\sigma)^{-1/2}\Phi(\rho)\Phi(\sigma)^{-1/2})\sigma^0=\sigma^{-1/2}\rho \sigma^{-1/2}$.

\item \label{rev cond8}
 $\Phi_\sigma^*(\Phi(\rho))=\rho$ (and also $\Phi_\sigma^*(\Phi(\sigma))=\sigma$ automatically).
 
\newcounter{szamlalo}
\setcounter{szamlalo}{\value{enumi}}
\end{enumerate}
\begin{enumerate}
\setcounter{enumi}{\value{szamlalo}}
\item \label{rev cond9} $\sigma^{-1/2}\rho \sigma^{-1/2}\in\fix{\Phi^*\circ\Phi_\sigma}$, the set of fixed points of
$\Phi^*\circ\Phi_\sigma$.
\setcounter{szamlalo}{\value{enumi}}
\end{enumerate}

Moreover, when we assume in addition that $\rho,\sigma$ are density operators with
invertible $\sigma$, the above \ref{rev cond1}--\ref{rev cond9} are also equivalent to
\begin{enumerate}
\setcounter{enumi}{\value{szamlalo}}
\item \label{rev cond10} $\bigl\<\Phi(\rho-\sigma),\Omega_{\Phi(\sigma)}^\kappa(\Phi(\rho-\sigma))\bigr\>_\HS
=\<\rho-\sigma,\Omega_\sigma^\kappa(\rho-\sigma)\>_\HS$ for some operator decreasing
function $\kappa:(0,+\infty)\to(0,+\infty)$ such that
$$
|\supp \nu_\kappa|\ge\big|\Sp(L_\sigma R_{\sigma^{-1}})\cup
\Sp\bigl(L_{\Phi(\sigma)}R_{\Phi(\sigma)^{-1}}\bigr)\big|,
$$
where $\Omega_\sigma^\kappa$ is given in \eqref{F-2.11} and $\nu_\kappa$ is the measure from the integral expression in \eqref{F-2.10}.
\end{enumerate}
\end{thm}

\begin{proof}
The equivalence of \ref{rev cond2}, \ref{rev cond4}, \ref{rev cond5}, and \ref{rev cond8} is in
\cite[Theorem 5.1]{HMPB}, and \ref{rev cond3}\,$\imp$\,\ref{rev cond2}\,$\imp$\,\ref{rev cond1}
is trivial. By Remark \ref{rem:positive mon}, \ref{rev cond1} yields that
\begin{align*}
S(\rho\|\sigma)=S(\mapp(\map(\rho))\|\mapp(\map(\rho)))\le S(\map(\rho)\|\map(\sigma))\le S(\rho\|\sigma)
\end{align*}
for $S=S_\eta$ with $\eta(x):=x\log x$. Since 
\begin{align*}
x\log x=\int_{(0,+\infty)}\biggl({x\over1+s}-{x\over x+s}\biggr)\,ds,
\end{align*}
we see that $\mu_f$ is the Lebesgue measure on $(0,+\infty)$, and hence 
\ref{rev cond1}\,$\imp$\,\ref{rev cond4} follows.

Next assume that \ref{rev cond2} holds, and consider the maps
$\Phi_0:\,\cB(\sigma^0\cH)=\sigma^0\BH\sigma^0\to\cB(\Phi(\sigma)^0\cK)=\Phi(\sigma)^0\BK\Phi(\sigma)^0$ and $\Psi_0:\,\cB(\Phi(\sigma)^0\cK)\to\cB(\sigma^0\cH)$ given by
$$
\Phi_0:=\Phi|_{\sigma^0\B(\hil)\sigma^0},\qquad
\Psi_0(Y):=\sigma^0\Psi(Y)\sigma^0,\quad Y\in\Phi(\sigma)^0\BK\Phi(\sigma)^0.
$$
Then it is easy to see that $(\map_0)^*$ and $(\map_0)_{\sigma}$ are unital $2$-positive maps,
and hence Schwarz contractions, and $(\mapp_0)^*$ is a Schwarz contraction; moreover,
\ref{rev cond2} is satisfied for $(\Phi_0,\rho,\sigma,\Psi_0)$ in
place of $(\Phi,\rho,\sigma,\Psi)$. Hence we can use \cite[Theorem 4]{Je} to conclude that
there exist CPTP maps $\widetilde\Phi_0:\,\cB(\sigma^0\cK)\to\cB(\Phi(\sigma)^0\cK)$ and
$\widetilde\Psi_0:\,\cB(\Phi(\sigma)^0\cK)\to\cB(\sigma^0\cK)$ such that
$\widetilde\Phi_0(\rho)=\Phi(\rho)$, $\widetilde\Phi_0(\sigma)=\Phi(\sigma)$ and
$\widetilde\Psi_0(\Phi_0(\rho))=\rho$, $\widetilde\Psi_0(\Phi_0(\sigma))=\sigma$. Define
CPTP maps $\widetilde\Phi:\,\BH\to\BK$ and $\widetilde\Psi:\,\BK\to\BH$ by
\begin{align*}
\widetilde\Phi(X)
&:=
\widetilde\Phi_0(\sigma^0X\sigma^0)+\pr{\psi_{\kil}}\cdot\Tr(I-\sigma^0)X,\qquad X\in\BH,\\
\widetilde\Psi(Y)
&:=
\widetilde\Psi_0(\Phi(\sigma)^0Y\Phi(\sigma)^0)
+\pr{\psi_{\hil}}\cdot\Tr(I-\Phi(\sigma)^0)Y,\qquad Y\in\BK,
\end{align*}
where $\psi_{\hil}\in\hil,\,\psi_{\kil}\in\kil$ are unit vectors.
Then \ref{rev cond3} holds for $\widetilde\Phi$ and $\widetilde\Psi$.

It was shown in \cite[Theorem 5.1]{HMPB} that \ref{rev cond4} implies 
\begin{align}\label{intermediate formula}
\sigma^0\map^*\bz\map(\sigma)^{-z}\map(\rho)^z\jz=\sigma^{-z}\rho^z,\ds\ds\ds z\in\bC,
\end{align}
which is condition (vi) of \cite[Theorem 5.1]{HMPB}. The proof of (vi)\,$\imp$\,(x) in
p.\,719 of \cite{HMPB} shows that 
this implies 
\begin{align*}
\sigma^0\map^*\bz\map(\sigma)^{-z}\map(\rho)^z Y\jz\sigma^0=\sigma^{-z}\rho^z\map^*(Y)\sigma^0
\end{align*}
for any $Y\in\B(\kil)$ and any $z\in\bC$. Hence we get \ref{rev cond6} by choosing 
$Y:=\map(\rho)^z\map(\sigma)^{-z}$ and using
$$
\map^*\bz\map(\rho)^z\map(\sigma)^{-z}\jz\sigma^0=\rho^z\sigma^{-z},\ds\ds\ds z\in\bC,
$$
which follows by taking the adjoint of both sides in \eqref{intermediate formula}.
The implication \ref{rev cond6}\,$\imp$\,\ref{rev cond7} is trivial. Even when $\Phi$ is only assumed to be positive, the equivalence \ref{rev cond7}\,$\iff$\,\ref{rev cond8} is a matter of straightforward computation. Thus, it has been shown that \ref{rev cond1}--\ref{rev cond8} are all equivalent.

It is clear that \ref{rev cond9} implies \ref{rev cond7}, and it is easily verified
by using Theorem \ref{thm:decomposition} that \ref{rev cond8} implies \ref{rev cond9}. Finally,
under the restriction of $\rho,\sigma$ to density operators, the equivalence
\ref{rev cond2}\,$\iff$\,\ref{rev cond10} was given in \cite[Proposition 4]{Je}.
\end{proof}

Note that when $\sigma$ is invertible, the equivalences
\ref{rev cond7}\,$\iff$\,\ref{rev cond8}\,$\iff$\,\ref{rev cond9} hold even when $\map$ is only
assumed to be positive.

Assume that $\map:\,\B(\hil)\to\B(\kil)$ is $2$-positive and trace-preserving and
$\sigma\in\BH_+$.  By the above theorem we have 
\begin{align*}
&\left\{\rho\in\B(\hil)_+:\,\rho^0\le\sigma^0\text{ and }S_f(\map(\rho)\|\map(\sigma))=S_f(\rho\|\sigma)\text{ for all operator convex }f:\,(0,+\infty)\to\bR\right\}\\
&\ds=
\left\{\rho\in\B(\hil)_+:\,\rho^0\le\sigma^0\text{ and }\map\text{ is reversible on }\{\rho,\sigma\}\right\}\\
&\ds=
\fix{\map_{\sigma}^*\circ\map}.
\end{align*}
In the above proof, we have used the following characterization
of $\mathcal{F}_{\Phi_\sigma^*\circ\Phi}$, due to \cite{HMPB,JP,MosonyiPhD,MP}:

\begin{thm}\label{thm:decomposition}
Let $\map:\,\B(\hil_1)\to\B(\hil_2)$ be a $2$-positive trace-preserving map, let $\sigma_1:=\sigma\in\B(\hil_1)_+\setminus\{0\}$, and 
$\sigma_2:=\map(\sigma)$. Then there exist decompositions
$\supp \sigma_m=\bigoplus_{k=1}^r\hil_{m,k,L}\otimes \hil_{m,k,R},\,m=1,2$, 
invertible density operators $\omega_{k}$ on $\hil_{1,k,R}$, 
unitaries $U_k\,:\hil_{1,k,L}\to\hil_{2,k,L}$,
and $2$-positive trace-preserving maps $\eta_k:\,\B(\hil_{1,k,R})\to\B(\hil_{2,k,R})$
such that $\omega_k$ is invertible on $\hil_{1,k,R}$, $\eta_k(\omega_{k})$ is invertible on $\hil_{2,k,R}$, and

\begin{align}
\fix{\map^*\circ\map_{\sigma}}&=\bigoplus_{k=1}^r\B(\hil_{1,k,L})\otimes I_{1,k,R},\label{dec1}\\
\fix{\map_{\sigma}\circ\map^*}&=\bigoplus_{k=1}^r\B(\hil_{2,k,L})\otimes I_{2,k,R},\label{dec2}\\
(\fix{\map_\sigma^*\circ\map})_+&=\bigoplus_{k=1}^r\B(\hil_{1,k,L})_+\otimes\omega_{k},\label{dec3}\\
\map(\rho_{1,k,L}\otimes \rho_{1,k,R})&=U_k \rho_{1,k,L}U_k^*\otimes \eta_k(\rho_{1,k,R}),\label{dec4}\\
\sigma^0\map^*(\rho_{2,k,L}\otimes \rho_{2,k,R})\sigma^0&=U_k^* \rho_{2,k,L}U_k\otimes \eta_k^*(\rho_{2,k,R}),
\label{dec5}
\end{align}
for all $\rho_{m,k,L}\in\B(\hil_{m,k,L}),\,\rho_{m,k,R}\in\B(\hil_{m,k,R})$.
\end{thm}

\begin{remark}
Note that the reversibility conditions \ref{rev cond1}--\ref{rev cond3} in Theorem \ref{T-3.12} are symmetric in 
$\rho$ and $\sigma$, while the rest of the equivalent characterizations of reversibility are not. To understand this, 
one should first note that deriving reversibility from the preservation of some $f$-divergence $S_f$ 
(i.e., the implication \ref{rev cond4}\,$\imp$\,\ref{rev cond1})
may only be possible 
if $S_f(\rho\|\sigma)<+\infty$, and the assumptions $\rho^0\le\sigma^0$ and $f(0^+)<+\infty$ guarantee this
(see Corollary \ref{cor:fdiv infty}).
If we assumed instead that $\sigma^0\le\rho^0$ and $f'(+\infty)<+\infty$ then \ref{rev cond4}\,$\imp$\,\ref{rev cond1}
would still hold; the proof of this can be reduced to the one with the original conditions, by using 
Proposition \ref{P-3.4} and noting that $|\supp\mu_{\wtilde f}|=|\supp\mu_f|$.
Of course, in this case $\rho$ and $\sigma$ have to be interchanged in points 
\ref{rev cond6}--\ref{rev cond10}.

There are two more ways to guarantee that $S_f(\rho\|\sigma)<+\infty$.
One is to assume that $\rho^0=\sigma^0$; it is easy to see that in this case we have the implication 
\ref{rev cond4}\,$\imp$\,\ref{rev cond1} even if we do not assume that $f(0^+)<+\infty$ or $f'(+\infty)<+\infty$;
one only has to note that in this case $\supp\mu_f$ in \ref{rev cond4} has to be replaced with $(\supp\lambda)\cap(0,+\infty)$, with $\lambda$ from \eqref{F-2.3}.
On the other hand, we do not know whether \ref{rev cond1} follows from \ref{rev cond4} if we assume that 
both $f(0^+)<+\infty$ and $f'(+\infty)<+\infty$, but we do not require any relation between the supports of $\rho$ and $\sigma$.
\end{remark}


\subsection{Maximal $f$-divergences}
\label{sec:maximal f-div}

In this section we consider in detail the quantum $f$-divergence introduced in
\eqref{maxdiv def}. This version of $f$-divergences was formerly treated in \cite{PR}, and more
recently it was studied in much detail by Matsumoto \cite{Ma}. While Matsumoto's definition,
referred to as the {\it maximal $f$-divergence}, 
is rather different
from that given here, it was shown in \cite[Lemma 4 and Theorem 5]{Ma} that the two definitions
coincide when $\rho^0\le \sigma^0$. Since our starting point here is the operator perspective function, 
we will use the notation $\what S_{f}$ for this family of $f$-divegences, as in \eqref{maxdiv def}, instead of the 
more operationally motivated notation $S_f^{\max}$ in \eqref{maxdiv def2}.

In this section we will always assume that $f$ is operator convex on $(0,+\infty)$. This is primarily to make sense of definition \eqref{F-3.18}; see Remark \ref{rem:op conv necessary}.

\begin{definition}\label{D-3.14}\rm
For invertible $\rho,\sigma\in\BH_+$ define
\begin{align}
\maxdiv{f}(\rho\|\sigma)&:=
\Tr\per{f}(\rho,\sigma)\nonumber\\
&\ =\Tr \sigma f(\sigma^{-1/2}\rho \sigma^{-1/2})
=\<\sigma^{1/2},f(\sigma^{-1/2}\rho \sigma^{-1/2})\sigma^{1/2}\>_\HS.\label{F-3.17}
\end{align}
For general $\rho,\sigma\in\BH_+$ let
\begin{equation}\label{F-3.18}
\maxdiv{f}(\rho\|\sigma):=\lim_{\eps\searrow0}\maxdiv{f}(\rho+\eps I\|\sigma+\eps I).
\end{equation}
\end{definition}

\begin{prop}\label{P-3.15}
\s
\begin{enumerate}
\item[(1)]\label{maxdiv existence}
For every $\rho,\sigma\in\BH_+$ the limit in \eqref{F-3.18} exists in $(-\infty,+\infty]$, and
it is equal to \eqref{F-3.17} for invertible $\rho,\sigma$. 
\item[(2)]\label{maxdiv convexity}
$\maxdiv{f}(\rho\|\sigma)$ is jointly convex in $\rho,\sigma\in\BH_+$.
\item[(3)]\label{maxdiv tranpose}
For every $\rho,\sigma\in\BH_+$,
\begin{align*}
\maxdiv{\widetilde f}(\rho\|\sigma)=\maxdiv{f}(\sigma\|\rho).
\end{align*}
\end{enumerate}
\end{prop}

\begin{proof}
The joint convexity of $(\rho,\sigma)\mapsto \maxdiv{f}(\rho\|\sigma)$ on
$\BH_{++}\times\BH_{++}$ follows from that of the perspective function $\per{f}$ given in
Lemma \ref{lemma:persp properties}\,(1). In particular, for every $\rho,\sigma\in\BH_+$ the
real function $t\mapsto \maxdiv{f}(\rho+tI\|\sigma+tI)$ is convex on $(0,+\infty)$, which implies
the existence of the limit in \eqref{F-3.18}, and that it is in $(-\infty,+\infty]$. The last
claim of (1) for invertible $\rho,\sigma$ is obvious, and (2) is immediate from definition
\eqref{F-3.18} and joint convexity on $\BH_{++}\times\BH_{++}$. For (3), applying Lemma
\ref{lemma:transpose perspective} to $\rho_{\ep}:=\rho+\ep I$, $\sigma_{\ep}:=\sigma+\ep I$,
$\ep>0$, taking the trace, and then the limit as $\ep\searrow 0$, we get the assertion by
\eqref{F-3.18}.
\end{proof}

\begin{remark}\label{rem:op conv necessary}
By Proposition \ref{P-A.1} note that the operator convexity of $f$ is a necessary and
sufficient condition for the joint convexity property of $\maxdiv{f}$ as stated in (2) above.
Although the details are not given here, we know that the joint
convexity of $S_f$ (equivalent to the monotonicity under CPTP maps,
see Remark 3.11) implies the operator convexity of $f$, whenever
$f$ is symmetric (i.e., $f=\widetilde f$) or both $f(0^+)$ and $f'(+\infty)$
are finite. However, it is still open whether this is true for a general
function $f$ on $(0,+\infty)$.
\end{remark}

Since $\maxdiv{f}$ arises as the trace of the operator perspective function $\per{f}$,
properties of the former can easily follow from those of the latter. From this point, it is
natural to study the properties of $\per{f}$ in further detail.
Below, we investigate to what extent the formula $\maxdiv{f}(\rho,\sigma)=\Tr\per{f}(\rho,\sigma)$
can be extended to  not necessarily invertible $\rho$ and $\sigma$. For this, we have to investigate whether the perspective function can be extended to not 
necessarily invertible operators. Note that this is not always possible in a natural way, as the following trivial example shows:
\begin{example}
Let $e_1,e_2$ be the canonical basis of $\bC^2$, and $\rho:=\pr{e_1}$, $\sigma:=\pr{e_2}$.
\begin{enumerate}
\item[(1)]
Let $f(x):=x^2$, which is operator convex with $f(0)=0$ and $f'(+\infty)=+\infty$. Then
$\lim_{\ep\searrow 0}\inner{e_1}{\per{f}(\rho+\ep I,\sigma+\ep I)e_1}=\lim_{\ep\searrow 0}(1+\ep)^2/\ep=+\infty$,
and hence $\ep\mapsto \per{f}(\rho+\ep I,\sigma+\ep I)$ does not have a limit as $\ep\searrow 0$.
\item[(2)]
Let $f(x):=1/x$, which is operator convex with $f'(+\infty)=0$ and $f(0^+)=+\infty$. Then 
$\lim_{\ep\searrow 0}\inner{e_2}{\per{f}(\rho+\ep I,\sigma+\ep I)e_2}=\lim_{\ep\searrow 0}(1+\ep)^2/\ep=+\infty$,
and hence $\ep\mapsto \per{f}(\rho+\ep I,\sigma+\ep I)$ does not have a limit as $\ep\searrow 0$.
\end{enumerate}
\end{example}

\begin{prop}\label{prop:persp extension}
Let $\rho,\sigma,\rho_n,\sigma_n\in\B(\hil)_+,\,n\in\bN$, be such that
$\lim_n\rho_n=\rho$ and $\lim_n\sigma_n=\sigma$. In the cases below, the limit
$\lim_{n\to\infty}\per{f}(\rho_n,\sigma_n)$ exists, independently of the choice of
$\rho_n,\sigma_n$, and it coincides with $\per{f}(\rho,\sigma)$  when both $\rho$ and $\sigma$
are invertible.
\begin{enumerate}
\item\label{op conv cond1}
If $f(0^+)<+\infty$, $f'(+\infty)<+\infty$, $\rho\le\rho_n$ and $\sigma\le\sigma_n$, then 
\begin{align*}
\lim_{n\to\infty}\per{f}(\rho_n,\sigma_n)=
f(0)\sigma+f'(+\infty)\rho-\sigma\,\tau_{h_f}\,\rho,
\end{align*}
where a non-negative operator monotone function $h_f$ on $[0,+\infty)$ is given by
$h_f(x):=\int_{(0,+\infty)}x(1+s)(x+s)^{-1}\,d\nu(s)$, $x\ge0$, with $\nu$ the representing
measure from \eqref{F-2.5}.

\item\label{op conv cond2}
If $f(0^+)<+\infty$ and $\sigma>0$, then 
\begin{align}\label{persp limit2}
\lim_{n\to\infty}\per{f}(\rho_n,\sigma_n)=
\sigma^{1/2}f\bz \sigma^{-1/2}\rho\sigma^{-1/2}\jz\sigma^{1/2}.
\end{align}
\item\label{op conv cond3}
If $f'(+\infty)<+\infty$ and $\rho>0$, then
\begin{align*}
\lim_{n\to\infty}\per{f}(\rho_n,\sigma_n)=
\rho^{1/2}\widetilde f\bz \rho^{-1/2}\sigma\rho^{-1/2}\jz\rho^{1/2},
\end{align*}
where $\widetilde f(x):=xf(x\inv)$ is the transpose of $f$.
\end{enumerate}
\end{prop}
\begin{proof}
\ref{op conv cond1}\enspace By \eqref{F-2.5}, $f(x)=f(0)+f'(+\infty)x-h_f(x),\,x\in(0,+\infty)$,
where  $h_f$ is a non-negative operator monotone function, and hence the assertion is immediate
from Lemma \ref{lemma:op mon extension}.

\ref{op conv cond2}\enspace
By the assumption, $f$ extends to a continuous function on $[0,+\infty)$, and thus
\eqref{persp limit2} follows from the continuity of functional calculus.

\ref{op conv cond3}\enspace
By Lemma \ref{lemma:transpose perspective}, we have 
$\per{f}(\rho_n,\sigma_n)=\per{\widetilde f}(\sigma_n,\rho_n)
=\rho_n^{1/2}\widetilde f\bz \rho_n^{-1/2}\sigma_n\rho_n^{-1/2}\jz\rho_n^{1/2}$
for every $n\in\bN$. Since $f'(+\infty)=\widetilde f(0^+)$, the assumption implies that
$\widetilde f$ extends to a continuous function on $[0,+\infty)$, and hence the assertion
follows as in (ii). 
\end{proof}

For applications, the assumptions $\sigma>0$ in \ref{op conv cond2} and $\rho>0$ in
\ref{op conv cond3} are too restrictive. However, we have the following:

\begin{prop}\label{prop:persp extension2}
Let $\rho,\sigma\in\B(\hil)_{+}$, and for every $n\in\bN$, let $K_n\ge 0$ be such that
$\rho+K_n>0$, $\sigma+K_n>0$, and $K_n\to0$.
\begin{enumerate}
\item\label{op conv cond4}
If $f(0^+)<+\infty$ and $\rho^0\le\sigma^0$, then 
\begin{align}\label{persp limit3}
\lim_{n\to\infty}\per{f}(\rho+K_n,\sigma+K_n)=
\sigma^{1/2}f\bz \sigma^{-1/2}\rho\sigma^{-1/2}\jz\sigma^{1/2}.
\end{align}
\item\label{op conv cond5}
If $f'(+\infty)<+\infty$ and $\sigma^0\le\rho^0$, then 
\begin{align}\label{persp limit5}
\lim_{n\to\infty}\per{f}(\rho+K_n,\sigma+K_n)=
\rho^{1/2}\widetilde f\bz \rho^{-1/2}\sigma\rho^{-1/2}\jz\rho^{1/2}.
\end{align}
\item If $\rho^0=\sigma^0$ then both \eqref{persp limit3} and \eqref{persp limit5} hold.
\end{enumerate} 
\end{prop}

Since the proof of the above proposition is rather lengthy, we defer it to Appendix
\ref{sec:extension proof}.
Now, we can extend the definition of $\per{f}$ to not necessarily invertible operators in the
following way:

\begin{definition}\label{def:persp extension}
Let $\rho,\sigma\in\B(\hil)_+$, and $f:\,(0,+\infty)\to\bR$ be an operator convex function such
that at least one of the following conditions is satisfied:
\begin{enumerate}
\item $\rho^0=\sigma^0$,
\item
$f(0^+)<+\infty$ and $f'(+\infty)<+\infty$,
\item
$f(0^+)<+\infty$ and $\rho^0\le\sigma^0$,
\item
$f'(+\infty)<+\infty$ and $\sigma^0\le\rho^0$.
\end{enumerate}
Then, we define $\per{f}(\rho,\sigma)$ as
\begin{align*}
\per{f}(\rho,\sigma):=\lim_{n\to\infty}\per{f}(\rho+K_n,\sigma+K_n),
\end{align*}
where $K_n\in\B(\hil)_+$ is any sequence such that $\rho+K_n,\sigma+K_n>0$ for every $n$ and 
$K_n\to0$.
\end{definition}

\begin{cor}\label{cor:maxdiv}
If any of the conditions in Definition \ref{def:persp extension} holds, then we have 
\begin{align*}
\maxdiv{f}(\rho\|\sigma)=\Tr\per{f}(\rho,\sigma).
\end{align*}
\end{cor}

In complete analogy with Corollary \ref{cor:fdiv infty}, we have the following:

\begin{prop}\label{prop:maxdiv infty}
$\maxdiv{f}(\rho\|\sigma)=+\infty$ if and only if one of the following conditions holds:
\begin{enumerate}
\item
$f(0^+)=+\infty$ and $\sigma^0\nleq\rho^0$;
\item
$f'(+\infty)=+\infty$ and $\rho^0\nleq\sigma^0$.
\end{enumerate}
In all other cases, $\maxdiv{f}(\rho\|\sigma)$ is a finite number.
\end{prop}
\begin{proof}
Assume that $f'(+\infty)=+\infty$ and $\rho^0\nleq\sigma^0$, so that there exists a unit vector
$\psi$ such that $\sigma^0\psi=0$ and $\inner{\psi}{\rho\psi}>0$. For all $\ep>0$,
\begin{align*}
&\Tr (\sigma+\ep I)^{1/2}f\bz(\sigma+\ep I)^{-1/2}(\rho+\ep I)(\sigma+\ep I)^{-1/2}
\jz(\sigma+\ep I)^{1/2}\\
&\ds\ds\ge 
\inner{\psi}{(\sigma+\ep I)^{1/2}f\bz(\sigma+\ep I)^{-1/2}(\rho+\ep I)(\sigma+\ep I)^{-1/2}
\jz(\sigma+\ep I)^{1/2}\psi}\\
&\ds\ds=\ep \inner{\psi}{f\bz(\sigma+\ep I)^{-1/2}(\rho+\ep I)(\sigma+\ep I)^{-1/2} \jz\psi}\\
&\ds\ds\ge\ep f\bz\inner{\psi}{(\sigma+\ep I)^{-1/2}(\rho+\ep I)(\sigma+\ep I)^{-1/2} \psi}\jz\\
&\ds\ds=\ep f\bz \ep\inv\inner{\psi}{(\rho+\ep I) \psi}\jz
=\ep f\bz \ep\inv\inner{\psi}{\rho\psi}+1\jz\\
&\ds\ds=
\frac{f\bz \ep\inv\inner{\psi}{\rho\psi}+1\jz}{\ep\inv\inner{\psi}{\rho\psi}+1}
\bz \inner{\psi}{\rho\psi}+\ep\jz,
\end{align*}
where the second inequality is due to Jensen's inequality. Since the last term converges to
$f'(+\infty)\inner{\psi}{\rho\psi}=+\infty$ as $\ep\searrow 0$,
$\maxdiv{f}(\rho,\sigma)=+\infty$. When $f(0^+)=+\infty$ and $\sigma^0\nleq\rho^0$, the previous
result combined with Proposition \ref{P-3.15}\,(3) yields immediately that
$\maxdiv{f}(\rho,\sigma)=+\infty$. 

Finiteness of $\maxdiv{f}(\rho,\sigma)$ in all other cases is immediate from Propositions
\ref{prop:persp extension} and \ref{prop:persp extension2}.
\end{proof}

\begin{prop}\label{prop:persp monotonicity}
Let $f:\,(0,+\infty)\to\bR$ be operator convex, and $\rho,\sigma\in\B(\hil)_+$ be such that at
least one of the conditions in Definition \ref{def:persp extension} holds.
Then, for any positive linear map 
$\map:\,\B(\hil)\to\B(\kil)$, we have
\begin{align}\label{persp monotonicity}
\per{f}(\map(\rho),\map(\sigma))\le\map\bz\per{f}(\rho,\sigma)\jz.
\end{align}
\end{prop}

\begin{proof}
The proof below is essentially the same as that of \cite[Proposition 2.5]{HP} (cf.~also
\cite[Lemma 3]{Ma}). By considering $\Phi$ as a map into
$\Phi(I)^0\BK\Phi(I)^0=\cB(\Phi(I)^0\cK)$, we can assume without loss of generality that
$\Phi(I)^0=I$. Let $\rho_n:=\rho+n^{-1}I$ and $\sigma_n:=\sigma+n^{-1}I$. 
Define $\map_{\sigma_n}:\BH\to\cB(\cK)$ by 
$\map_{\sigma_n}(X):=\map(\sigma_n)^{-1/2}\map(\sigma_n^{1/2} X\sigma_n^{1/2})\map(\sigma_n)^{-1/2}$, as in 
\eqref{F-3.13}. Then $\map_{\sigma_n}$ is a unital positive map, and Lemma \ref{Choi inequality}
yields
\begin{align*}
f(\Phi_{\sigma_n}(\sigma_n^{-1/2}\rho_n \sigma_n^{-1/2}))\le
\map_{\sigma_n}(f(\sigma_n^{-1/2}\rho_n \sigma_n^{-1/2})),
\end{align*}
which means that
\begin{align*}
\Phi(\sigma_n)^{1/2}f(\Phi(\sigma_n)^{-1/2}\Phi(\rho_n)\Phi(\sigma_n)^{-1/2})\Phi(\sigma_n)^{1/2}
\le\Phi(\sigma_n^{1/2}f(\sigma_n^{-1/2}\rho_n \sigma_n^{-1/2})\sigma_n^{1/2}),
\end{align*}
i.e., $P_f(\Phi(\rho_n),\Phi(\sigma_n))\le\Phi(P_f(\rho_n,\sigma_n))$. By now using Propositions
\ref{prop:persp extension}, \ref{prop:persp extension2} and Definition \ref{def:persp extension},
taking the limit $n\to\infty$ gives \eqref{persp monotonicity}.
\end{proof}

Now, the monotonicity of $\maxdiv{f}$ follows immediately:

\begin{cor}\label{P-3.20}
Let $\Phi:\BH\to\BK$ be a trace-preserving positive linear map. Then for every
$\rho,\sigma\in\BH_+$,
\begin{align}\label{maxdiv monotonicity}
\maxdiv{f}(\Phi(\rho)\|\Phi(\sigma))\le \maxdiv{f}(\rho\|\sigma).
\end{align}
\end{cor}

\begin{proof}
If any of the conditions in Definition \ref{def:persp extension} is satisfied, then 
\eqref{maxdiv monotonicity} is immediate from \eqref{persp monotonicity}. Otherwise 
$\maxdiv{f}(\rho,\sigma)=+\infty$, according to Proposition \ref{prop:maxdiv infty}, and thus
the assertion is trivial.
\end{proof}

\begin{remark}\label{remark:persp}
Let $\rho,\sigma\in\BH_+$ with $\rho^0\le\sigma^0$. For any function $\vfi:[0,\infty)\to\bR$,
one can define $P_\ffi(\rho,\sigma):=
\sigma^{1/2}\ffi(\sigma^{-1/2}\rho\sigma^{-1/2})\sigma^{1/2}$
simply via functional calculus. When $f$ is operator convex with $f(0^+)<+\infty$, this
definition is consistent with case (iii) of Definition \ref{def:persp extension} due to
\eqref{persp limit3}. When $\Phi:\BH\to\BK$ is a positive linear map, one can also define
$P_f(\Phi(\rho),\Phi(\sigma))$ in the same way since $\Phi(\rho)^0\le\Phi(\sigma)^0$. If the map
$\Phi_\sigma$ defined in \eqref{F-3.13} is considered as a map from $\cB(\sigma^0\cH)$ to
$\cB(\Phi(\sigma)^0\cK)$, then it is unital and positive, so one can apply Lemma
\ref{Choi inequality} to have
$$
f(\Phi_\sigma(\sigma^{-1/2}\rho\sigma^{-1/2}))\le\Phi_\sigma(f(\sigma^{-1/2}\rho\sigma^{-1/2})),
$$
which means \eqref{persp monotonicity}. Thus, Proposition \ref{prop:persp monotonicity}, if
restricted to this situation, follows in a simpler way without the convergence argument.
\end{remark}

\begin{remark}
When $h:(0,+\infty)\to\bR$ is non-negative and operator monotone, Proposition
\ref{prop:persp monotonicity} applied to $f:=-h$ shows that for any positive linear map
$\Phi:\BH\to\BK$ and for every $\rho,\sigma\in\BH_+$,
\begin{equation}\label{Ando inequality}
\Phi(\rho\,\tau_h\,\sigma)\le\Phi(\rho)\,\tau_h\,\Phi(\sigma).
\end{equation}
This inequality is essentially due to Ando \cite{An}, where it was proved only for the geometric
and the harmonic means in a similar way to the proof of Proposition \ref{prop:persp monotonicity}.
We will use this observation in the proof of \ref{max f rev cond4}\,$\imp$\,\ref{max f rev cond1}
in Theorem \ref{T-3.21} below.
\end{remark}
\medskip

By Corollary \ref{P-3.20} it is obvious that if $\Phi$ is reversible on $\{\rho,\sigma\}$
(see Definition \ref{def:rev}), then
$$
\maxdiv{f}(\Phi(\rho)\|\Phi(\sigma))=\maxdiv{f}(\rho\|\sigma)
$$
for all operator convex functions $f$ on $(0,+\infty)$. The next theorem presents several
equivalent conditions for the equality case of $\maxdiv{f}$ under $\Phi$. We note that the
implication \ref{max f rev cond1}\,$\imp$\,\ref{max f rev cond9} was shown in
\cite[Lemma 12]{Ma} under an additional assumption on the support of $\mu_f$, analogous to
\eqref{F-3.15}. Here we stress that assumption \eqref{F-3.15} on $f$ is essential in
\ref{rev cond4} of Theorem \ref{T-3.12} (see, e.g., \cite[Example 1]{Je}),  while $f$ in
\ref{max f rev cond1} of Theorem \ref{T-3.21} can be an arbitrary non-linear operator convex
function. The equivalence \ref{max f rev cond9}\,$\iff$\,\ref{max f rev cond7} was also pointed
out in \cite[Section 9.1]{Ma}. Moreover, we note that a variant of 
\ref{max f rev cond1}\,$\imp$\,\ref{max f rev cond9} in the case where $\rho^0\nleq\sigma^0$ was
given in \cite[Lemma 12]{Ma}.

\renewcommand\theenumi{(\alph{enumi})}
\begin{thm}\label{T-3.21}
Let $\rho,\sigma\in\B(\hil)_+$ be such that $\rho^0\le\sigma^0$, and let $\map:\,\B(\hil)\to\B(\kil)$ be a positive trace-preserving linear map. Then the following are equivalent:
\begin{enumerate}
\item\label{max f rev cond1}
$\maxdiv{f}(\Phi(\rho)\|\Phi(\sigma))=\maxdiv{f}(\rho\|\sigma)$ for some non-linear
operator convex function $f$ on $[0,+\infty)$.
\item\label{max f rev cond2}
$\maxdiv{f}(\Phi(\rho)\|\Phi(\sigma))=\maxdiv{f}(\rho\|\sigma)$ for all operator convex functions $f$ on
$[0,+\infty)$.
\item\label{max f rev cond3}  
$\Tr\Phi(\rho)^2\Phi(\sigma)^{-1}=\Tr \rho^2\sigma^{-1}$.
\item\label{max f rev cond9}  
$\per{\vfi}(\map(\rho),\map(\sigma))=\map\bz\per{\vfi}(\rho,\sigma)\jz$ for all functions $\vfi$
on $[0,+\infty)$.
\item\label{max f rev cond5} 
$\Phi(\sigma)\,\tau\,\Phi(\rho)=\Phi(\sigma\,\tau\,\rho)$ for all operator connections $\tau$.
\item\label{max f rev cond4} 
$\Phi(\sigma)\,\tau\,\Phi(\rho)=\Phi(\sigma\,\tau\,\rho)$ for some non-linear operator connection
$\tau$.
\item\label{max f rev cond6}
$\Phi(\rho \sigma^{-1}\rho)=\Phi(\rho)\Phi(\sigma)^{-1}\Phi(\rho)$.
\item\label{max f rev cond7}
$\map_{\sigma}(\crn^2)=\bz\map_{\sigma}(\crn)\jz^2$,\ds\ds where\ds\ds $\crn:=\sigma^{-1/2}\rho \sigma^{-1/2}$.
\setcounter{szamlalo}{\value{enumi}}
\end{enumerate}
If we further assume that $\map$ is $2$-positive, then the above are also equivalent to 
\begin{enumerate}
\setcounter{enumi}{\value{szamlalo}}
\item\label{max f rev cond8}
$\sigma^{-1/2}\rho \sigma^{-1/2}\in\cM_{\Phi_\sigma}$.
\end{enumerate}

\end{thm}
\renewcommand\theenumi{(\roman{enumi})}

\begin{proof}
We shall prove 
\begin{align}\label{maxdiv preservation}
\text{\ref{max f rev cond3}\,$\iff$\,\ref{max f rev cond6}\,$\iff$\,\ref{max f rev cond7}
\,$\imp$\,\ref{max f rev cond9}\,$\imp$\,\ref{max f rev cond2}\,$\imp$\,\ref{max f rev cond1}
\,$\imp$\,\ref{max f rev cond3} \ds and \ds
\ref{max f rev cond9}\,$\imp$\,\ref{max f rev cond5}\,$\imp$\,\ref{max f rev cond4}
\,$\imp$\,\ref{max f rev cond1}.}
\end{align}

First, note that $P_\vfi(\rho,\sigma)$ and $P_\vfi(\Phi(\rho),\Phi(\sigma))$ are defined in
the sense of Remark \ref{remark:persp}, and $P_h(\rho,\sigma)=\sigma\,\tau_h\,\rho$ when $h$ is
a non-negative operator monotone function on $[0,\infty)$ with the corresponding operator
connection $\tau_h$. Hence \ref{max f rev cond9}\,$\imp$\,\ref{max f rev cond2} is obvious by
Corollary \ref{cor:maxdiv}, and the implications
\ref{max f rev cond2}\,$\imp$\,\ref{max f rev cond1} and
\ref{max f rev cond9}\,$\imp$\,\ref{max f rev cond5}\,$\imp$\,\ref{max f rev cond4} are trivial.
We also remark (although not necessary for the rest of the proof) that
\ref{max f rev cond2}\,$\imp$\,\ref{max f rev cond3}\,$\imp$\,\ref{max f rev cond1} is obvious
by applying equality in \ref{max f rev cond2} to the quadratic function.

\ref{max f rev cond3}\,$\iff$\,\ref{max f rev cond6} is easy since
$\Phi(\rho)\Phi(\sigma)^{-1}\Phi(\rho)\le\Phi(\rho \sigma^{-1}\rho)$ (see, e.g.,
\cite[Proposition 2.7.3]{Bhatia2} and \cite[Lemma 3.5]{HMPB}).

\ref{max f rev cond6}\,$\iff$\,\ref{max f rev cond7} 
follows immediately from
\begin{align*}
\Phi_\sigma\bigl((\sigma^{-1/2}\rho \sigma^{-1/2})^2\bigr)&=
\Phi(\sigma)^{-1/2}\Phi(\rho\sigma^{-1}\rho)\Phi(\sigma)^{-1/2},\\
\bigl(\Phi_\sigma(\sigma^{-1/2}\rho\sigma^{-1/2})\bigr)^2&=
\Phi(\sigma)^{-1/2}\Phi(\rho)\Phi(\sigma)^{-1}\Phi(\rho)\Phi(\sigma)^{-1/2}.
\end{align*}

%

\ref{max f rev cond7}\,$\imp$\,\ref{max f rev cond9}.\enspace
By considering $\map_{\sigma}$ as a map from $\B(\sigma^0\hil)$ to $\B(\map(\sigma)^0\kil)$,
we can assume that  $\map_{\sigma}$ is unital.
Let $\wtilde\map_{\sigma}$ be the restriction of $\map_{\sigma}$ onto the commutative algebra generated by 
$\crn=\sigma^{-1/2}\rho \sigma^{-1/2}$. By \ref{max f rev cond7}, $\sigma^{-1/2}\rho \sigma^{-1/2}$ is in the multiplicative domain of $\wtilde\map_{\sigma}$ (see \eqref{multdom} and \eqref{multdom2}). Thus,
\begin{align*}
\vfi\bz\map_{\sigma}\bz\sigma^{-1/2}\rho \sigma^{-1/2}\jz\jz
&=
\vfi\bz\wtilde\map_{\sigma}\bz\sigma^{-1/2}\rho \sigma^{-1/2}\jz\jz\\
&=
\wtilde\map_{\sigma}\bz\vfi\bz\sigma^{-1/2}\rho \sigma^{-1/2}\jz\jz
=
\map_{\sigma}\bz\vfi\bz\sigma^{-1/2}\rho \sigma^{-1/2}\jz\jz,
\end{align*}
where the second equality is due to Lemma \ref{lemma:mult domain fcalculus}.
The equality of the first and the last terms above is exactly \ref{max f rev cond9}.


\ref{max f rev cond1}\,$\imp$\,\ref{max f rev cond3}.\enspace
For $s\in(0,+\infty)$ set
\begin{equation}\label{F-3.25}
f_s(x):=-{x\over x+s},\qquad x\in[0,+\infty),
\end{equation}
which is an operator convex function on $[0,+\infty)$. From the integral expression \eqref{F-2.4}
of $f$ one has
\begin{equation}\label{F-3.26}
\maxdiv{f}(\rho\|\sigma)=f(0)\Tr \sigma+a\Tr \rho+b\Tr \rho^2\sigma^{-1}
+\int_{(0,+\infty)}\biggl({\Tr \rho\over1+s}+\maxdiv{f_s}(\rho\|\sigma)\biggr)\,d\mu_f(s)
\end{equation}
and similarly
\begin{align}
\maxdiv{f}(\Phi(\rho)\|\Phi(\sigma))
&=f(0)\Tr\Phi(\sigma)+a\Tr\Phi(\rho)+b\Tr\Phi(\rho)^2\Phi(\sigma)^{-1} \nonumber\\
&\qquad+\int_{(0,+\infty)}\biggl({\Tr\Phi(\rho)\over1+s}
+\maxdiv{f_s}(\Phi(\rho)\|\Phi(\sigma))\biggr)\,d\mu_f(s) \nonumber\\
&=f(0)\Tr \sigma+a\Tr \rho+b\Tr\Phi(\rho)^2\Phi(\sigma)^{-1} \nonumber\\
&\qquad+\int_{(0,+\infty)}\biggl({\Tr \rho\over1+s}
+\maxdiv{f_s}(\Phi(\rho)\|\Phi(\sigma))\biggr)\,d\mu_f(s). \label{F-3.27}
\end{align}
By comparing \eqref{F-3.26} and \eqref{F-3.27} together with the monotonicity property of
Corollary \ref{P-3.20}, one must have
$$
\Tr\Phi(\rho)^2\Phi(\sigma)^{-1}=\Tr \rho^2\sigma^{-1}\quad\mbox{if $b>0$},
$$
$$
\maxdiv{f_s}(\Phi(\rho)\|\Phi(\sigma))=\maxdiv{f_s}(\rho\|\sigma)\quad
\mbox{for all $s\in\supp\mu$}.
$$
Since $f$ is non-linear, if follows that $b>0$ or $\supp\mu_f$ is not empty. So it suffices to
prove that (c) holds if $\maxdiv{f_s}(\Phi(\rho)\|\Phi(\sigma))=\maxdiv{f_s}(\rho\|\sigma)$ for
some $s\in(0,+\infty)$. Since $f_s(x)=-1+s(x+s)^{-1}$, the assumption implies that
\begin{align*}
&\Tr \Phi(\sigma)^{1/2}\bigl[\Phi(\sigma)^{-1/2}\Phi(\rho)\Phi(\sigma)^{-1/2}+sI\bigr]^{-1}
\Phi(\sigma)^{1/2}
=\Tr \sigma^{1/2}\bigl[\sigma^{-1/2}\rho \sigma^{-1/2}+sI\bigr]^{-1}\sigma^{1/2}.
\end{align*}
By noting that $\sigma^{-1/2}\rho \sigma^{-1/2}+sI
=\sigma^{-1/2}(\rho+s\sigma)\sigma^{-1/2}+s(I-\sigma^0)$, the above can be rephrased as
$$
\Tr\Phi(\sigma)^2\Phi(\rho+s\sigma)^{-1}=\Tr \sigma^2(\rho+s\sigma)^{-1}.
$$
As we have already proved
\ref{max f rev cond3}\,$\iff$\,\ref{max f rev cond7}\,$\imp$\,\ref{max f rev cond2}, we can
apply \ref{max f rev cond3}\,$\imp$\,\ref{max f rev cond2} to $\sigma$ and $\rho+s\sigma$
(in place of $\rho$, $\sigma$) and $f(x)=(x+\eps)^{-1}$ for any $\eps>0$. We then find that
\begin{align*}
&\Tr\Phi(\rho+s\sigma)^{1/2}\bigl[\Phi(\rho+s\sigma)^{-1/2}\Phi(\sigma)\Phi(\rho+s\sigma)^{-1/2}
+\eps I\bigr]^{-1}\Phi(\rho+s\sigma)^{1/2} \\
&\qquad=\Tr(\rho+s\sigma)^{1/2}\bigl[(\rho+s\sigma)^{-1/2}\sigma(\rho+s\sigma)^{-1/2}+\eps I\bigr]^{-1}(\rho+s\sigma)^{1/2}.
\end{align*}
Letting $\eps\searrow0$ yields
$$
\Tr\Phi(\rho+s\sigma)^2\Phi(\sigma)^{-1}=\Tr(\rho+s\sigma)^2\sigma^{-1}
$$
so that
$$
\Tr\Phi(\rho)^2\Phi(\sigma)^{-1}+2s\Tr\Phi(\rho)+s^2\Tr\Phi(\sigma)
=\Tr \rho^2\sigma^{-1}+2s\Tr \rho+s^2\Tr \sigma.
$$
Therefore, $\Tr\Phi(\rho)^2\Phi(\sigma)^{-1}=\Tr \rho^2\sigma^{-1}$.

\ref{max f rev cond4}\,$\imp$\,\ref{max f rev cond1}.\enspace
Assume \ref{max f rev cond4} for $\tau=\tau_h$ with a non-negative operator monotone function $h$.
From the integral expression \eqref{F-2.8} one writes
\begin{align}
\Phi(\sigma\,\tau_h\,\rho)&=a\Phi(\rho)+b\Phi(\sigma)
+\int_{(0,+\infty)}\Phi(\sigma\,\tau_{h_s}\,\rho)\,d\nu_h(s), \label{F-3.28}\\
\Phi(\sigma)\,\tau_h\,\Phi(\rho)
&=a\Phi(\rho)+b\Phi(\sigma)+\int_{(0,+\infty)}\Phi(\sigma)\,\tau_{h_s}\,\Phi(\rho)\,d\nu_h(s).
\label{F-3.29}
\end{align}
Comparing \eqref{F-3.28} and \eqref{F-3.29} implies by means of \eqref{Ando inequality} that
$$
\Phi(\sigma)\,\tau_{h_s}\,\Phi(\rho)=\Phi(\sigma\,\tau_{h_s}\,\rho)\quad
\mbox{for all $s\in\supp\nu_h$}.
$$
Since $h_s(x)=-(1+s)f_s(x)$ with $f_s$ given in \eqref{F-3.25}, one finds that
$$
\Tr\bigl(\sigma\,\tau_{h_s}\,\rho\bigr)=-(1+s)\maxdiv{f_s}(\rho\|\sigma).
$$
Therefore, the above equality means that
$$
\maxdiv{f_s}(\Phi(\rho)\|\Phi(\sigma))
=\maxdiv{f_s}(\rho\|\sigma)\quad\mbox{for $s\in\supp\nu_h$}.
$$
Hence \ref{max f rev cond1} follows since $\tau_h$ is non-linear so that $\supp\nu_h$ is not
empty.

\ref{max f rev cond8}\,$\iff$\,\ref{max f rev cond7}.
As before, by considering $\map_{\sigma}$ as a map from $\B(\sigma^0\hil)$ to $\B(\map(\sigma)^0\kil)$,
we can assume that $\map_{\sigma}$ is unital. If $\map$ is $2$-positive then so is $\map_{\sigma}$. Hence,
\ref{max f rev cond8} is equivalent to \ref{max f rev cond7}, according to 
\eqref{multdom} and \eqref{multdom2}.
\end{proof}


\begin{remark}
Note that \ref{max f rev cond3} gives a particularly easy-to-verify criterion for the rest of the points in Theorem \ref{T-3.21} to hold.
\end{remark}

\section{Comparison of different $f$-divergences}
\label{sec:comparison}

In this section we compare the quantum $f$-divergences  
$\maxdiv{f},\,S_f$, and $S_f^{\meas}$. In particular, 
in Section \ref{sec:maxdiv vs mod div}, we 
extend and strengthen Matsumoto's inequality 
$S_f(\rho\|\sigma)\le \maxdiv{f}(\rho\|\sigma)$,
that was proved in \cite{Ma} for the case where $f$ is operator convex on $[0,+\infty)$ and 
$\rho^0\le\sigma^0$. In Section \ref{sec:reversibility comparison}, we compare the preservation
of $S_f$ and $\maxdiv{f}$ in Theorems \ref{T-3.12} and \ref{T-3.21}. Finally, in Section
\ref{sec:meas fdiv}, we discuss the measured $f$-divergence.

\subsection{The relation of $S_f$ and $\maxdiv{f}$}
\label{sec:maxdiv vs mod div}

It is easy to verify that if $\rho,\sigma\in\BH_+$ are commuting, then
$S_f(\rho\|\sigma)=\maxdiv{f}(\rho\|\sigma)$ for every $f$. 
The main result of this section, given in Theorem \ref{P-4.1}, is that the converse is also true
in the sense that $S_f(\rho\|\sigma)=\maxdiv{f}(\rho\|\sigma)$ for some \ki{fixed} 
operator convex function $f$ implies the commutativity of $\rho$ and $\sigma$, provided that $f$
satisfies some technical condition.

In general, one has
\begin{align}\label{fdiv maxdiv ineq}
S_f(\rho\|\sigma)\le \maxdiv{f}(\rho\|\sigma)
\end{align}
for any operator convex function $f$ on $(0,+\infty)$. By Proposition \ref{P-3.8}, this is a
special case of a more general statement proved by Matsumoto \cite{Ma}, given in
\eqref{min-max properties} (although he only considered operator convex functions on
$[0,+\infty)$). The proof for the general case (i.e., without the assumption  $f(0^+)<+\infty$)
goes the same way, using Matsumoto's construction of the ``minimal reverse test"; we give it in
detail below as a preparation for the proof of  the stronger inequality given in Theorem
\ref{P-4.1}.

\begin{prop}\label{P-3.19}
For every $\rho,\sigma\in\BH_+$, and every operator convex function $f:\,(0,+\infty)\to\bR$,
$$
S_f(\rho\|\sigma)\le \maxdiv{f}(\rho\|\sigma).
$$
\end{prop}

\begin{proof}
By definitions \eqref{F-3.6} and \eqref{F-3.18} one may
assume that $\rho,\sigma>0$. Choose the spectral decomposition of
$\sigma^{-1/2}\rho \sigma^{-1/2}$ as
$$
\sigma^{-1/2}\rho \sigma^{-1/2}=\sum_{i=1}^k\lambda_iP_i,
$$
where the $P_i$ are orthogonal projections with $\sum_{i=1}^kP_{i}=I$. 
For every $i=1,\ldots,k$, let $\delta_i$ denote the indicator function of the singleton $\{i\}$
in the commutative algebra $\bC^k$, and 
define a trace-preserving positive linear map $\Phi$ from $\bC^k$ to $\BH$ by
$$
\Phi\Biggl(\sum_{i=1}^kx_i\delta_i\Biggr)
:=\sum_{i=1}^kx_i\,{\sigma^{1/2}P_i\sigma^{1/2}\over\Tr \sigma P_i},
$$
and $\a,\b\in\bC^k$ by
$$
\a:=\sum_{i=1}^k(\lambda_i\Tr \sigma P_i)\delta_i,\qquad\b:=\sum_{i=1}^k(\Tr \sigma P_i)\delta_i.
$$
Then $\Phi$ is CPTP, and
$$
\Phi(\a)=\sum_{i=1}^k\lambda_i\sigma^{1/2}P_i\sigma^{1/2}=\rho,\qquad
\Phi(\b)=\sum_{i=1}^k\sigma^{1/2}P_i\sigma^{1/2}=\sigma.
$$
Therefore, by the monotonicity property of $S_f$ (Proposition \ref{P-3.8}) one has
\begin{align*}
S_f(\rho\|\sigma)&\le S_f(\a\|\b)
=\sum_{i=1}^k(\Tr \sigma P_i)f\bigl((\lambda_i\Tr \sigma P_i)(\Tr \sigma P_i)^{-1}\bigr) \\
&=\sum_{i=1}^k(\Tr \sigma P_i)f(\lambda_i)=\Tr \sigma f(\sigma^{-1/2}\rho \sigma^{-1/2})
=\maxdiv{f}(\rho\|\sigma),
\end{align*}
which is the required inequality.
\end{proof}

It is easy to see that $S_f$ is actually equal to $\maxdiv{f}$ when $f$ is a polynomial of
degree two:

\begin{example}\label{E-3.18}\rm(Quadratic function)\quad
For the quadratic function $f_2(x):=x^2$ and for $\rho,\sigma>0$,
$$
S_{f_2}(\rho\|\sigma)=\Tr\rho^2\sigma^{-1}=\maxdiv{f_2}(\rho\|\sigma).
$$
Therefore, when $f$ is of the form $f(x)=ax^2+bx+c$ with $a\ge0$, we have
$S_f(\rho\|\sigma)=\maxdiv{f}(\rho\|\sigma)$ for all $\rho,\sigma\in\BH_+$.
\end{example}

\begin{thm}\label{P-4.1}
Let $\rho,\sigma\in\BH_+$ satisfy $\rho^0\le \sigma^0$ and $\rho \sigma\ne \sigma\rho$.
Then
\begin{equation}\label{F-4.1}
S_f(\rho\|\sigma)<\maxdiv{f}(\rho\|\sigma)
\end{equation}
for any operator convex function $f$ on $[0,+\infty)$ such that 
\begin{align}\label{strict order support cond}
|\supp\mu_f|\ge \abs{\spec(\sigma^{-1/2}\rho\sigma^{-1/2})\cup\spec\bz L_{\rho}R_{\sigma\inv}\jz}.
\end{align}
\end{thm}

\begin{proof}
The proof is based on the minimal reverse test \cite{Ma} as in the proof of Proposition
\ref{P-3.19}. 
Write the spectral decomposition of
$\sigma^{-1/2}\rho \sigma^{-1/2}$ as
$$
\sigma^{-1/2}\rho \sigma^{-1/2}=\sum_{i=1}^k\lambda_iP_i,\qquad
\lambda_1>\lambda_2>\dots>\lambda_k,
$$
and define the trace-preserving positive map $\Phi:\bC^k\to\BH$ and $\a,\b\in\bC^k$ as in the
proof of Proposition \ref{P-3.19}. Then $\Phi(\a)=\rho$, $\Phi(\b)=\sigma$ and
$S_f(\rho\|\sigma)\le S_f(\a\|\b)=\maxdiv{f}(\rho\|\sigma)$. Now, assume that
$S_f(\rho\|\sigma)=\maxdiv{f}(\rho\|\sigma)$ and prove that $\rho$ and $\sigma$ must be
commuting. Since $S_f(\rho\|\sigma)=S_f(\a\|\b)$ and \eqref{F-3.15} is satisfied, it
follows from Theorem \ref{T-3.12} that
\begin{equation}\label{F-4.2}
\a\b^{-1}\in\fix{\Phi^*\circ\Phi_\b}.
\end{equation}
Since
\begin{align}
\Phi_\b\Biggl(\sum_{i=1}^kx_i\delta_i\Biggr)
&=\sigma^{-1/2}\Phi\Biggl(\sum_{i=1}^k(\Tr \sigma P_i)x_i\delta_i\Biggr)\sigma^{-1/2} \nonumber\\
&=\sigma^{-1/2}\Biggl(\sum_{i=1}^kx_i\sigma^{1/2}P_i\sigma^{1/2}\Biggr)\sigma^{-1/2}
=\sum_{i=1}^kx_iP_i \label{F-4.3}
\end{align}
and $\a\b^{-1}=\sum_{i=1}^k\lambda_i\delta_i$, we have
$\Phi_\b(\a\b^{-1})=\sum_{i=1}^k\lambda_iP_i=\sigma^{-1/2}\rho \sigma^{-1/2}$. Moreover, since
$$
\Bigl\<X,\Phi\Biggl(\sum_ix_i\delta_i\Biggr)\Bigr\>_\HS
=\Tr\Biggl(X^*\sum_ix_i\,{\sigma^{1/2}P_i\sigma^{1/2}\over\Tr \sigma P_i}\Biggr)
=\sum_i\overline{\biggl({\Tr X\sigma^{1/2}P_i\sigma^{1/2}\over\Tr \sigma P_i}\biggr)}\,x_i,
$$
we have
$$
\Phi^*(X)=\sum_{i=1}^k{\Tr X\sigma^{1/2}P_i\sigma^{1/2}\over\Tr \sigma P_i}\,\delta_i
$$
with the convention $0/0=0$. Therefore,
$$
\Phi^*\circ\Phi_\b(\a\b^{-1})
=\sum_{i=1}^k{\Tr \sigma^{-1/2}\rho \sigma^{-1/2}\sigma^{1/2}P_i\sigma^{1/2}\over\Tr \sigma P_i}\,\delta_i
=\sum_{i=1}^k{\Tr \rho P_i\over\Tr \sigma P_i}\,\delta_i
$$
so that (4.4) yields
$$
\sum_{i=1}^k\lambda_i\delta_i=\sum_{i=1}^k{\Tr\rho P_i\over\Tr\sigma P_i}\,\delta_i.
$$
Since $\rho=\sum_{j=1}^k\lambda_j\sigma^{1/2}P_j\sigma^{1/2}$, this implies that
$$
\lambda_i={\Tr\rho P_i\over\Tr\sigma P_i}=\sum_{j=1}^k\lambda_j\pi_j^i,
\qquad1\le i\le k,
$$
where
$$
\pi_j^i:={\Tr\sigma^{1/2}P_j\sigma^{1/2}P_i\over\Tr\sigma P_i},
\qquad1\le i,j\le k.
$$
Note that $\pi_j^i\ge0$ and $\sum_{j=1}^k\pi_j^i=1$ for all $i$. Since
$\lambda_1>\lambda_2>\cdots$, by taking $i=1$ we obtain $\pi_j^1=0$ for all $j\ne1$, and
hence $\pi_1^i=0$ for all $i\ne1$, too, since
$\Tr\sigma^{1/2}P_j\sigma^{1/2}P_i=\Tr\sigma^{1/2}P_i\sigma^{1/2}P_j$. Using the same argument
for $i=2,3,\dots$, we obtain $\pi_j^i=0$, i.e., $\Tr\sigma^{1/2}P_j\sigma^{1/2}P_i=0$ for
$i\ne j$, so that $P_i\sigma^{1/2}P_j=0$, $i\ne j$. This yields that $\sigma^{1/2}$ commutes
with $\sigma^{-1/2}\rho\sigma^{-1/2}$, hence $\rho\sigma^{1/2}=\sigma^{1/2}\rho$, so that
$\rho\sigma=\sigma\rho$.
\end{proof}

\begin{example}\label{E-4.2}\rm(Log function)\quad
Consider $\eta(x):=x\log x$ as in Example \ref{E-3.4}. For every $\rho,\sigma\in\BH_+$ with
$\rho^0\le \sigma^0$ we have
$$
\maxdiv{\eta}(\rho\|\sigma)
=\Tr \sigma^{1/2}\rho \sigma^{-1/2}\log(\sigma^{-1/2}\rho \sigma^{-1/2})
=\Tr \rho\log(\rho^{1/2}\sigma^{-1}\rho^{1/2}),
$$
which is the {\it Belavkin-Staszewski relative entropy} $S_\BS(\rho\|\sigma)$ introduced in
\cite{BS}. Proposition \ref{P-3.19} gives the inequality
$S(\rho\|\sigma)\le S_\BS(\rho\|\sigma)$, which was first proved in \cite{HP}. By
Theorem \ref{P-4.1} we further have $S(\rho\|\sigma)<S_\BS(\rho\|\sigma)$ whenever
$\rho^0\le \sigma^0$ and $\rho \sigma\ne \sigma\rho$.
\end{example}

\begin{example}\label{E-4.3}\rm(Power functions)\quad
Consider $f_\alpha$, given in Example \ref{E-3.4}, for $\alpha\in(0,2]$. For
$\rho,\sigma>0$ we have
$$
\maxdiv{f_\alpha}(\rho\|\sigma)
=s(\alpha)\Tr \sigma^{1/2}(\sigma^{-1/2}\rho \sigma^{-1/2})^\alpha \sigma^{1/2}.
$$
When $0<\alpha\le1$, this is rewritten as
$$
\maxdiv{f_\alpha}(\rho\|\sigma)=-\Tr \sigma\,\#_\alpha\,\rho,
$$
where $\#_\alpha$ denotes the weighted geometric mean corresponding to $x^\alpha$.
Proposition \ref{P-3.19} gives the inequality
$\Tr \sigma\,\#_\alpha\,\rho\le\Tr \rho^\alpha \sigma^{1-\alpha}$ for
$\rho,\sigma\in\BH_+$, which is also a consequence of the well-known log-majorization \cite{AH}.
When $1\le\alpha\le2$ and $\rho,\sigma\in\BH_+$ with $\rho^0\le \sigma^0$, by Proposition
\ref{P-3.19} we also have
$$
\Tr \rho^\alpha \sigma^{1-\alpha}
\le\Tr \sigma^{1/2}(\sigma^{-1/2}\rho \sigma^{-1/2})^\alpha \sigma^{1/2},
$$
which seems a novel trace inequality in matrix theory. Furthermore, Theorem \ref{P-4.1}
implies that if $\rho^0\le \sigma^0$ and $\rho \sigma\ne \sigma\rho$, then
\begin{align}
\Tr \rho\,\#_\alpha\,\sigma&<\Tr \rho^{1-\alpha}\sigma^\alpha
\ \ \mbox{for $\alpha\in(0,1)$}, \label{F-4.8}\\
\Tr \rho^\alpha \sigma^{1-\alpha}&<\Tr \sigma^{1/2}(\sigma^{-1/2}\rho \sigma^{-1/2})^\alpha \sigma^{1/2}
\ \ \mbox{for $\alpha\in(1,2)$}. \label{geom-Renyi ineq}
\end{align}
Note that more refined results than \eqref{F-4.8} are found in \cite{Hi1}.
\end{example}

\begin{remark}
Further to the above example, it is worth mentioning that if $\rho^0\le\sigma^0$ and $\rho\sigma\ne\sigma\rho$,
then the strict inequality in (4.8) holds in the opposite direction for $\alpha\in(2,+\infty)$.
Indeed, by elaborating the method in [3], one can prove the following log-majorization results
(for the definition and basics of log-majorization, see [3]):
\begin{itemize}
\item[(a)] $\sigma^{1/2}(\sigma^{-1/2}\rho\sigma^{-1/2})^\alpha\sigma^{1/2}\prec_{\log}
(\sigma^{1-\alpha\over2z}\rho^{\alpha\over z}\sigma^{1-\alpha\over2z})^z$
if $0<\alpha\le1$ and $z>0$,
\item[(b)] $(\sigma^{1-\alpha\over2z}\rho^{\alpha\over z}\sigma^{1-\alpha\over2z})^z
\prec_{\log}\sigma^{1/2}(\sigma^{-1/2}\rho\sigma^{-1/2})^\alpha\sigma^{1/2}$
if $\alpha\ge1$ and $z\ge\max\{\alpha/2,\alpha-1\}$,
\item[(c)] $\sigma^{1/2}(\sigma^{-1/2}\rho\sigma^{-1/2})^\alpha\sigma^{1/2}\prec_{\log}
(\sigma^{1-\alpha\over2z}\rho^{\alpha\over z}\sigma^{1-\alpha\over2z})^z$
if $\alpha\ge1$ and $0<z\le\min\{\alpha/2,\alpha-1\}$.
\end{itemize}
In particular, when $z=1$, (a) and (b) imply the inequalities in (4.7) and (4.8), respectively,
and (c) implies the opposite inequality of (4.8), where the strict inequality when
$\alpha\in(0,+\infty)\setminus\{1,2\}$ can be shown by using [31, Theorem 2.1]. Note that
$\mathrm{Tr}(\sigma^{1-\alpha\over2z}\rho^{\alpha\over z}\sigma^{1-\alpha\over2z})^z$ is the
main component of the $\alpha$-$z$-R\'enyi divergence in (1.4).
\end{remark}
\medskip

It is natural to ask whether the support condition \eqref{strict order support cond} in Theorem
\ref{P-4.1} is necessary. The following proposition shows that, at least when $\dim\hil=2$, 
the condition \eqref{strict order support cond} is not needed, and $\supp\mu_f\ne\emptyset$
is sufficient to guarantee the strict inequality in \eqref{F-4.1} for any non-commuting pair 
$(\rho,\sigma)$. Note that this condition cannot be further weakened, as $\supp\mu_f=\emptyset$ 
means that $f$ is a polynomial of degree at most two, in which case $S_f=\maxdiv{f}$, according to Example \ref{E-3.18}.

\begin{prop}\label{P-4.5}
Let $f$ be an operator convex function on $[0,+\infty)$ that is not a
polynomial,
i.e., $\supp\mu_f\ne\emptyset$. Then for any non-commuting $\rho,\gamma\in\cB(\bC^2)_+$ with
$\rho^0\le\gamma^0$,
$$
S_f(\rho\|\gamma)<\maxdiv{f}(\rho\|\gamma).
$$
\end{prop}

\begin{proof}
We sketch the proof here. We may restrict to $2\times2$ density matrices. In the
well-known Bloch sphere description, a qubit density matrix is written as
${1\over2}(I+\bw\cdot\sigma)$, where $\bw\cdot\sigma:=w_1\sigma_1+w_2\sigma_2+w_3\sigma_3$ for
$\bw=(w_1,w_2,w_3)\in\bR^3$ with $|\bw|:=\bigl(w_1^2+w_2^2+w_3^2\bigr)^{1/2}\le1$ and Pauli
matrices $\sigma_1=\begin{bmatrix}0&1\\1&0\end{bmatrix}$,
$\sigma_2=\begin{bmatrix}0&-i\\i&0\end{bmatrix}$ and
$\sigma_3=\begin{bmatrix}1&0\\0&-1\end{bmatrix}$. Let $\rho={1\over2}(I+\bw\cdot\sigma)$ and
$\gamma={1\over2}(I+\bx\cdot\sigma)$ with $\bw,\bx\in\bR^3$, and assume that $\rho,\gamma>0$,
equivalently $|\bw|,|\bx|<1$. Set $\by:=\bw-\bx$, $\bu:=\bw-s\bx$ and $\bv:=\bw+s\bx$. Thanks
to the integral expression \eqref{F-2.3}, to prove Proposition \ref{P-4.5}, it is enough to show
that if $\rho\gamma\ne\gamma\rho$ then
\begin{equation}\label{F-4.17}
S_{g_s}(\rho\|\gamma)<\maxdiv{g_s}(\rho\|\gamma),
\end{equation}
where $g_s(x):=(x-1)^2/(x+s)$ with $s\in(0,+\infty)$; here note that $s=0$ is excluded due
to $f(0^+)<+\infty$. Since
$$
S_{g_s}(\rho\|\gamma)=\Tr(\rho-\gamma)\,{1\over L_\rho+sR_\gamma}(\rho-\gamma),
$$
we have by \cite[Lemma B.5]{HR}
\begin{equation}\label{F-4.18}
S_{g_s}(\rho\|\gamma)=(1+s)\Bigl\<\by,\bigl[\bigl\{(1+s)^2-|\bu|^2\bigr\}I
+|\bu\>\<\bu|-|\bv\>\<\bv|\bigr]^{-1}\by\Bigr\>.
\end{equation}
On the other hand, by using \cite[(B2)]{HR}, we have
\begin{align*}
\maxdiv{g_s}(\rho\|\gamma)
&=\Tr\gamma^{-1}(\rho-\gamma)\gamma(\rho+s\gamma)^{-1}(\rho-\gamma) \\
&={1\over2(1-|\bx|^2)\bigl[(1+s)^2-|\bv|^2\bigr]}
\,\Tr(I-\bx\cdot\sigma)(\by\cdot\sigma)(I+\bx\cdot\sigma)
[(1+s)I-\bv\cdot\sigma](\by\cdot\sigma).
\end{align*}
A bit tedious computation using \cite[(B1)]{HR} gives
$$
\Tr(I-\bx\cdot\sigma)(\by\cdot\sigma)(I+\bx\cdot\sigma)
[(1+s)I-\bv\cdot\sigma](\by\cdot\sigma)
=2(1+s)|\by|^2(1-|\bx|^2)
$$
and hence
\begin{equation}\label{F-4.19}
\maxdiv{g_s}(\rho\|\gamma)={1+s\over(1+s)^2-|\bv|^2}\,|\by|^2.
\end{equation}
Here, note that $\rho\gamma\ne\gamma\rho$ if and only if $\bw,\bx$ are linearly independent,
equivalently so are $\bu,\bv$. When this holds, an elementary but again tedious computation with
\eqref{F-4.18} and \eqref{F-4.19} shows that \eqref{F-4.17} is equivalent to
$$
|\bu|^2+|\bv|^2-2\biggl({1-s\over1+s}\biggr)\bu\cdot\bv<4s.
$$
Since $\bu-\bv=-2s\bx$ and $\bu\cdot\bv=|\bw|^2-s^2|\bx|^2$, the above left-hand side is
${4s\over1+s}(|\bw|^2+s|\bx|^2)<4s$ since $|\bw|,|\bx|<1$. When $\rho\not>0$ and
$\gamma>0$, the computation is similar with $|\bw|=1$ and $|\bx|<1$.
\end{proof}

\subsection{The relation of the preservation conditions}
\label{sec:reversibility comparison}

In this section we compare the implications of the preservation of the two $f$-divergences, 
$S_f$ and $\maxdiv{f}$, by a quantum operation; 
that is, we compare Theorem \ref{T-3.12} and Theorem \ref{T-3.21}.
As it turns out, the preservation of $S_f$ is in general strictly stronger than the preservation 
of $\maxdiv{f}$, i.e., in general the preservation of $\maxdiv{f}$ does not imply the reversibility 
of the quantum operation as in Definition \ref{def:rev}.

This can be seen in various ways. In \cite[Remark 5.4]{HMPB}, an example from \cite{JPP} was used to show states 
$\rho,\sigma$ and a CPTP map $\map$ such that $\map$ is not reversible on $\{\rho,\sigma\}$,
but $S_{f_2}(\map(\rho)\|\map(\sigma))=S_{f_2}(\rho\|\sigma)$ holds for $f_2(x)=x^2$.
By Example \ref{E-3.18} and \ref{max f rev cond3} of Theorem \ref{T-3.21}, this latter condition implies that  
$\maxdiv{f}(\map(\rho)\|\map(\sigma))=\maxdiv{f}(\rho\|\sigma)$ for every operator convex function 
$f$ on $(0,+\infty)$; yet reversibility does not hold. The example from \cite{JPP} is rather involved; below we give a much simpler one, 
in Example \ref{E-4.4}.

Another way to see the above statement is to consider Matsumoto's minimal reverse test 
$(\map,\a,\b)$ as in the proof of Proposition \ref{P-3.19}. Then 
\begin{align*}
\maxdiv{f}(\map(\a)\|\map(\b))&=
\maxdiv{f}(\rho\|\sigma)=
S_f(\a\|\b)=
\maxdiv{f}(\a\|\b)
\end{align*}
for any operator convex function $f$ on $(0,+\infty)$, and thus all of \ref{max f rev cond1}--\ref{max f rev cond7} in Theorem 
\ref{T-3.21} hold. 
However, if $f$ satisfies the support condition \eqref{strict order support cond} and
$\rho^0\le\sigma^0$ and $\rho\sigma\ne\sigma\rho$, then
by Theorem 4.3 we have
\begin{align*}
S_f(\map(\a)\|\map(\b))&=S_f(\rho\|\sigma)<\maxdiv{f}(\rho\|\sigma)=S_f(\a\|\b),
\end{align*}
and hence none of \ref{rev cond1}--\ref{rev cond9} in Theorem \ref{T-3.12} hold.
Note that while the argument in the previous paragraph was based on a very specific example, using the
function $f_2$, the argument in this paragraph shows that, in general, preservation of 
$\maxdiv{f}$ does not imply reversibility for any function that satisfies the support condition 
\eqref{strict order support cond}.

Yet another approach is given in Example \ref{E-4.4} below, where we directly compare 
\ref{rev cond7} of Theorem \ref{T-3.12} and \ref{max f rev cond6} of Theorem \ref{T-3.21}.
Note that the map used in \cite[Remark 5.4]{HMPB} is not unital, and neither is the map $\map$ in the minimal reverse test 
unless $\rho,\sigma$ are commuting and $k=\dim\cH$ (i.e., all
$P_i$ are rank one). Hence, Example \ref{E-4.4} with a unital qutrit channel
gives a further non-trivial insight into the difference of the preservation of the two $f$-divergences.

On the other hand, the points of Theorems \ref{T-3.12} and \ref{T-3.21} become equivalent when some further 
conditions are imposed on $(\map,\rho,\sigma)$. This happens, for instance, in the qubit case when $\map$ is unital, as shown in Proposition \ref{P-4.6}, or in the case where $\map(\rho)$ and $\map(\sigma)$ commute, given in Proposition \ref{C-3.22} below.

\begin{example}\label{E-4.4}\rm
Let $\mathcal{H}=\bC^3$ and $P$ be the orthogonal projection onto $\bC^2\oplus0$. Let
$\Phi:\BH\to\BH$ be the pinching
$$
\Phi(X):=PXP+(I-P)X(I-P),
$$
which is a unital qutrit channel. Let $\rho:=|\psi\>\<\psi|$ with $\psi\in\ran P$, and
$\sigma:=b_1|x_1\>\<x_1|+b_2|x_2\>\<x_2|+|x_3\>\<x_3|$ with $b_1,b_2>0$,
where $\{x_1,x_2,x_3\}$ is an orthonormal basis in $\bC^3$.
It is easy to verify that
\begin{align}
&\Phi^*(\Phi(\sigma)^{-1/2}\Phi(\rho)\Phi(\sigma)^{-1/2})=\sigma^{-1/2}\rho \sigma^{-1/2} \label{counterex1}\\
&\ds\ds\ds\iff
\ds |(P\sigma P)^{-1/2}\psi\>\<(P\sigma P)^{-1/2}\psi|=|\sigma^{-1/2}\psi\>\<\sigma^{-1/2}\psi|, \ds\text{and}\label{F-4.10}\\
&\Phi(\rho \sigma^{-1}\rho)=\Phi(\rho)\Phi(\sigma)^{-1}\Phi(\rho)\label{F-4.11}\\
&\ds\ds\ds \iff
\ds \<\psi,\sigma^{-1}\psi\>=\<\psi,(P\sigma P)^{-1}\psi\>,\label{counterex2}
\end{align}
where \eqref{counterex1} is \ref{rev cond7} of Theorem \ref{T-3.12}, and \eqref{F-4.11} is
\ref{max f rev cond6} of Theorem \ref{T-3.21}. Hence, in order to find an example where the
equivalent points of of Theorem \ref{T-3.21} hold, but those of Theorem \ref{T-3.12} do not, we
have to set the parameters above so that 
\begin{align}
&\text{\textbullet}\ds\sigma^{-1/2}\psi\ds\text{and}\ds(P\sigma P)^{-1/2}\psi \ds \text{are linearly independent (i.e., \eqref{F-4.10} fails), and}\label{counterex3}\\
&\text{\textbullet}\ds\|\sigma^{-1/2}\psi\|=\|(P\sigma P)^{-1/2}\psi\|\ds \text{(i.e., \eqref{counterex2} holds).}\label{counterex4}
\end{align}
In order to achieve this, let us choose
$$
\psi:=\begin{bmatrix}1\\1\\0\end{bmatrix},\quad
x_1:={1\over\sqrt3}\begin{bmatrix}1\\1\\1\end{bmatrix},\quad
x_2:={1\over\sqrt2}\begin{bmatrix}1\\0\\-1\end{bmatrix},\quad
x_3:={1\over\sqrt6}\begin{bmatrix}1\\-2\\1\end{bmatrix}.
$$
It is straightforward to compute
$$
\sigma^{-1/2}\psi=b_1^{-1/2}\<x_1,\psi\>x_1+b_2^{-1/2}\<x_2,\psi\>x_2+\<x_3,\psi\>x_3
=\begin{bmatrix}{2\over3}b_1^{-1/2}+{1\over2}b_2^{-1/2}-{1\over6}\\
{2\over3}b_1^{-1/2}+{1\over3}\\
{2\over3}b_1^{-1/2}-{1\over2}b_2^{-1/2}-{1\over6}\end{bmatrix},
$$
and
\begin{equation}\label{F-4.15}
\|\sigma^{-1/2}\psi\|^2={4\over3}b_1^{-1}+{1\over2}b_2^{-1}+{1\over6}.
\end{equation}
On the other hand, we have
\begin{align*}
P\sigma P=b_1|Px_1\>\<Px_1|+b_2|Px_2\>\<Px_2|+|Px_3\>\<Px_3|
=\begin{bmatrix}{2b_1+3b_2+1\over6} & {b_1-1\over3} & 0\\{b_1-1\over3} & {b_1+2\over3} \\ 0 & 0 & 0\end{bmatrix},
\end{align*}
and
$$
(P\sigma P)^{-1}={6\over b_1b_2+3b_1+2b_2}
\begin{bmatrix}{b_1+2\over3} & -{b_1-1\over3} & 0\\-{b_1-1\over3} & {2b_1+3b_2+1\over6} & 0\\ 0 & 0 & 0\end{bmatrix},
$$
so that
\begin{equation}\label{F-4.16}
\|(P\sigma P)^{-1/2}\psi\|^2={3(b_2+3)\over b_1b_2+3b_1+2b_2}.
\end{equation}
We find that \eqref{F-4.15} and \eqref{F-4.16} are equal, for instance, when $b_1=1/3$ and
$b_2=3/11$, in which case the third coordinate of $\sigma^{-1/2}\psi$ is non-zero. Therefore, when
$\sigma={1\over3}|x_1\>\<x_1|+{3\over11}|x_2\>\<x_2|+|x_3\>\<x_3|$, we see that both
\eqref{counterex3} and \eqref{counterex4} are satisfied, as required.

This shows that $\what S_f(\map(\rho)\|\map(\sigma))=\what S_f(\rho\|\sigma)$ for any operator convex $f$
on $[0,+\infty)$, while
$S_f(\map(\rho)\|\map(\sigma))<S_f(\rho\|\sigma)$ for any operator convex $f$ such that $|\supp\mu_f|\ge 7$.
In particular,
$\Phi$ is not reversible on $\{\rho,\sigma\}$, while \ref{max f rev cond1}--\ref{max f rev cond7} of Theorem \ref{T-3.21} hold.
\end{example}

\begin{remark}
 Since in the above example
$\Phi^*\circ\Phi_\sigma=\Phi_\sigma$, 
comparing this with \ref{rev cond9} of Theorem \ref{T-3.12} and \ref{max f rev cond7} of Theorem \ref{T-3.21}
shows that 
$\Phi_\sigma$ is an example of a unital channel
$\Psi:\cB(\bC^3)\to\cB(\bC^3)$ such that
$$
\bC I\subsetneqq\fix{\Psi}\subsetneqq\cM_\Psi\subsetneqq\cB(\bC^3);
$$
cf.~also Appendix \ref{sec:examples}.
\end{remark}
\bigskip

The next proposition shows that Example \ref{E-4.4} has minimal dimension among {\it unital}
channels for which the points of Theorems \ref{T-3.12} and \ref{T-3.21} are inequivalent,
though we have a {\it non-unital} qubit channel showing the difference (see the discussion
before Example \ref{E-4.4}).

\begin{prop}\label{P-4.6}\rm
All the points of Theorems \ref{T-3.12} and \ref{T-3.21} are equivalent to each other for any unital qubit
channel $\Phi:\cB(\bC^2)\to\cB(\bC^2)$.
\end{prop}

\begin{proof}
Let $\Phi$ be a unital qubit channel, and let $\rho,\gamma\in\mathcal{B}(\bC^2)_+$ with
$\gamma>0$. Assume that the equivalent statements of Theorem \ref{T-3.21} hold for $\Phi$ and
$\rho,\gamma$, and we prove that $\Phi$ is reversible on $\{\rho,\gamma\}$. By considering
$(\rho+\gamma)/\Tr(\rho+\gamma)$ and $\gamma/\Tr \gamma$ in place of $\rho$ and $\gamma$,
respectively, it suffices to assume that $\rho$ and $\gamma$ are invertible density
matrices, so we write $\rho={1\over2}(I+\bw\cdot\sigma)$ and
$\gamma={1\over2}(I+\bx\cdot\sigma)$ with $|\bw|,|\bx|<1$, and let $\by$ and $\bv$ be as in
the proof of Proposition \ref{P-4.5}. We may also assume that $\rho\ne\gamma$, i.e.,
$\bw\ne\bx$. In the Bloch sphere description, recall that $\Phi$ acts on density matrices
as follows:
$$
\Phi:{1\over2}(I+\bbz\cdot\sigma)\longmapsto{1\over2}(I+T\bbz\cdot\sigma),\qquad
\bbz\in\bR^3,\ |\bbz|\le1,
$$
where $T$ is a $3\times3$ real matrix with the operator norm $\|T\|_\infty\le1$. Consider
$g_s(x):=(x-1)^2/(x+s)$, $s\in[0,+\infty)$, as given in the proof of Proposition
\ref{P-4.5}. By assumption we have the equality
$\maxdiv{g_s}(\rho\|\gamma)=\maxdiv{g_s}(\Phi(\rho)\|\Phi(\sigma))$, which means by
\eqref{F-4.19} that
$$
{1+s\over(1+s)^2-|\bv|^2}\,|\by|^2={1+s\over(1+s)^2-|T\bv|^2}\,|T\by|^2.
$$
Since $\|T\|_\infty\le1$, this forces $|T\bv|=|\bv|$ and $|T\by|=|\by|$, which are
equivalent to $T^*T\bv=\bv$ and $T^*T\by=\by$. Hence $T^*T\bw=\bw$ and $T^*T\bx=\bx$. Now,
recall \cite[(17) and (22)]{HR} that the so-called Bogoliubov-Kubo-Mori monotone Riemannian
metric on invertible density matrices is
$\<X,\Omega_\gamma^\BKM(Y)\>_\HS=\Tr X\Omega_\gamma^\BKM(Y)$ for $X,Y\in\cB(\bC^2)_{\sa}^0$,
where $\Omega_\gamma^\BKM=\Omega_\gamma^\kappa$ with
$\kappa(x):=(\log x)/(x-1)$ in \eqref{F-2.11} is given as
\begin{equation}\label{F-4.20}
\Omega_\gamma^\BKM(Y):={1\over2}\int_0^\infty{1\over2\gamma+sI}\,Y\,{1\over2\gamma+sI}\,ds.
\end{equation}
Thanks to \cite[(B21)]{HR} we have
\begin{align*}
&{1\over2}\,\Tr(\by\cdot\sigma){1\over(1+s)I+\bx\cdot\sigma}(\by\cdot\sigma)
{1\over(1+s)I+\bx\cdot\sigma} \\
&\qquad={|\by|^2\over\bigl[(1+s)^2-|\bx|^2\bigr]^2}\bigl[(1+s)^2+|\bx|^2\cos2\theta\bigr],
\end{align*}
where $\theta$ is the angle between $\bx$ and $\by$. Since $|T\bx|=|\bx|$, $|T\by|=|\by|$ and
$$
(T\bx)\cdot(T\by)=\bx\cdot(T^*T\by)=\bx\cdot\by,
$$
it follows that the angle between $T\bx$ and $T\by$ coincides with $\theta$. Therefore,
\begin{align*}
&{1\over2}\,\Tr((T\by)\cdot\sigma){1\over(1+s)I+(T\bx)\cdot\sigma}((T\by)\cdot\sigma)
{1\over(1+s)I+(T\bx)\cdot\sigma} \\
&\qquad={1\over2}\,\Tr(\by\cdot\sigma){1\over(1+s)I+\bx\cdot\sigma}(\by\cdot\sigma)
{1\over(1+s)I+\bx\cdot\sigma}.
\end{align*}
Integrate the above for $s\in(0,+\infty)$, and apply \eqref{F-4.20} to obtain
$$
\bigl\<\Phi(\by\cdot\sigma),\Omega_{\Phi(\gamma)}^\BKM(\Phi(\by\cdot\sigma))\bigr\>_\HS
=\bigl\<\by\cdot\sigma,\Omega_\gamma^\BKM(\by\cdot\sigma)\bigr\>_\HS.
$$
Since
$$
{\log x\over x-1}=\int_{(0,+\infty)}{1\over(x+s)(1+s)}\,ds,
$$
it follows that \ref{rev cond10} of Theorem \ref{T-3.12} holds with
$\supp\nu_\kappa=(0,+\infty)$, so that $\Phi$ is reversible on $\{\rho,\gamma\}$.
\end{proof}

\begin{prop}\label{C-3.22}
Let $\rho,\sigma$, and $\Phi$ be as in Theorem \ref{T-3.21}.
If $\Phi(\rho)$ commutes with $\Phi(\sigma)$, then the points of Theorem \ref{T-3.21}
imply those of Theorem \ref{T-3.12}.
\end{prop}

\begin{proof}
By Propositions \ref{P-3.8} and \ref{P-3.19}, we have 
\begin{align}\label{general inequalities}
S_f(\map(\rho)\|\map(\sigma))\le S_f(\rho\|\sigma)\le\maxdiv{f}(\rho\|\sigma)
\end{align}
for all operator convex functions $f$ on $[0,+\infty)$. Assume now \ref{max f rev cond2} of 
Theorem \ref{T-3.21}. Since $\map(\rho)$ and $\map(\sigma)$ commute, we then have
\begin{align*}
S_f(\map(\rho)\|\map(\sigma))=\maxdiv{f}(\map(\rho)\|\map(\sigma))=
\maxdiv{f}(\rho\|\sigma)
\end{align*}
for all operator convex functions $f$ on $[0,+\infty)$, from which, when combined with 
\eqref{general inequalities}, we get that  $S_f(\map(\rho)\|\map(\sigma))=S_f(\rho\|\sigma)$ 
for all operator convex functions $f$ on $[0,+\infty)$, i.e., all the points of Theorem
\ref{T-3.21} hold.
\end{proof}

By Theorems \ref{T-3.12}, \ref{T-3.21}, and Proposition \ref{C-3.22} we have

\begin{cor}\label{C-3.23}
Let $\map:\,\B(\hil)\to\B(\kil)$ be a $2$-positive trace-preserving map. If
$\Phi(\rho)$ commutes with $\Phi(\sigma)$ for all $\rho\in\BH$ (in particular, if $\Phi$ is
a quantum-classical channel, i.e., the range of $\Phi$ is commutative), then
$\fix{\Phi^*\circ\Phi_\sigma}=\cM_{\Phi_\sigma}$.
\end{cor}

In particular, if $\Phi$ is a unital channel (trace-preserving) and $\sigma=I$ (so
$\Phi(\sigma)=I$), then $\fix{\Phi^*\circ\Phi}=\cM_\Phi$ holds. This is contained in
\cite[Theorem 11]{CJK}, where the fixed point algebra $\fix{\Phi^*\circ\Phi}$ was denoted
by $UCC(\Phi)$ and called the UCC algebra ({\it unitarily correctable} codes). The unitality
of the channel $\Phi$ seems essential in \cite{CJK}.

Another special case is when $\Phi$ is a (trace-preserving) conditional expectation onto a
subalgebra of $\BH$ and $\Phi(\sigma)=\sigma$. 
In this case, $\Phi_\sigma=\Phi$ and $\Phi^*$ is the inclusion map
of the subalgebra into $\mathcal{B}(\mathcal{H})$, hence
$\Phi^*\circ\Phi_\sigma=\Phi_\sigma=\Phi$. Moreover, we have
$\mathcal{F}_\Phi=\mathcal{M}_\Phi$ by the above corollary since
$\Phi(I)=I$, which is also easily verified directly, and therefore
the points of Theorem 3.1 imply those of Theorem 3.16.

\subsection{Measured $f$-divergence}
\label{sec:meas fdiv}

A measurement $\M$ on $\hil$ is given by $(M_x)_{x\in\X}$, where $\X$ is a finite set (the set of possible outcomes), $M_x\in\B(\hil)_+$ for all $x\in\X$, and $\sum_{x\in\X}M_x=I$. The measurement 
$\M$ is then a CPTP map from $\B(\hil)$ to $\bC^{\X}$, given by $\M(A):=\sum_{x\in\X}(\Tr AM_x)\delta_x$, where $\delta_x$ is the indicator function of the singleton $\{x\}$. 
We will use the same notation for this CPTP map and the collection of operators $(M_x)_{x\in\X}$. We will denote
the set of all measurements on $\hil$ with outcomes in $\X$ by $\povm(\hil|\X)$.

We say that the measurement is projective if all the $M_x$ are projections, and it is a von Neumann measurement if all the $M_x$ are rank $1$ projections.
We will use the notation 
\begin{align*}
|\M|:=|\X|.
\end{align*}
It is easy to see that with definition \eqref{meas fdiv def}, we have 
\begin{align}\label{meas fdiv2}
S_f^{\meas}(\rho\|\sigma)=\sup\{&S_f(\M(\rho)\|\M(\sigma)):\,\M\text{ measurement on }\hil\}.
\end{align}
Here we use the classical $f$-divergence 
\begin{align*}
S_f(p,q):=\sum_{x\in\X}P_f(p(x),q(x)),\ds\ds\ds p,q\in[0,+\infty)^{\X},
\end{align*}
that reduces to \eqref{cl fdiv def1} when both $p$ and $q$ are strictly positive.
We can introduce two variants of the measured $f$-divergences, by restricting the measurements to projective and von Neumann measurements, respectively:
\begin{align}
S_f^{\pro}(\rho\|\sigma):=\sup\{&S_f(\M(\rho)\|\M(\sigma)):\,\M\text{ projective measurement on }\hil\},\label{pr fdiv def}\\
S_f^{\vN}(\rho\|\sigma):=\sup\{&S_f(\M(\rho)\|\M(\sigma)):\,\M\text{ rank $1$ projective measurement on }\hil\}.\label{vN fdiv def}
\end{align}
Obviously,
\begin{align}\label{measured fdiv inequalities}
S_f^{\vN}(\rho\|\sigma)\le S_f^{\pro}(\rho\|\sigma)\le S_f^{\meas}(\rho\|\sigma)
\end{align}
for any $\rho,\sigma$ and any $f$. 

\begin{lemma}\label{lemma:vN=pr}
For any $\rho,\sigma\in\B(\hil)_+$, and any convex function $f:\,(0,+\infty)\to\bR$, 
\begin{align*}
S_f^{\vN}(\rho\|\sigma)=S_f^{\pro}(\rho\|\sigma).
\end{align*}
\end{lemma}
\begin{proof}
By \eqref{measured fdiv inequalities}, we only have to prove $S_f^{\vN}(\rho\|\sigma)\ge S_f^{\pro}(\rho\|\sigma)$. To this end, let 
$\M$ be a projective measurement, given by the measurement operators $P_x,\,x\in\X$. Each $P_x$ can be decomposed as
$P_x=\sum_{i=1}^{k_x}\pr{e_{x,i}}$, where $\{e_{x,i}\}_{i=1}^{k_x}$ is an ONB in $\ran P_x$. By the generalized log-sum inequality (i.e., the joint convexity of the classical perspective function), for every $\ep>0$ we have
\begin{align*}
\sum_{i=1}^{k_x}(\Tr\pr{e_{x,i}}\sigma+\ep)f\bz\frac{\Tr\pr{e_{x,i}}\rho+\ep}{\Tr\pr{e_{x,i}}\sigma+\ep}\jz
\ge
(\Tr P_x\sigma+k_x\ep)f\bz\frac{\Tr P_x\rho+k_x\ep}{\Tr P_x\sigma+k_x\ep}\jz.
\end{align*}
Summing over $x$, and taking the limit $\ep\searrow 0$ yields
$$
S_f\bz\{\Tr\pr{e_{x,i}}\rho\}_{x,i}\|\{\Tr\pr{e_{x,i}}\sigma\}_{x,i}\jz\ge S_f\bz\{\Tr P_x\rho\}_x\|\{\Tr P_x\sigma\}_x\jz,
$$
from which the assertion follows immediately.
\end{proof}

Due to Lemma \ref{lemma:vN=pr}, we will only use the notation $S_f^{\pro}$ for the rest, with the understanding that the supremum 
in \eqref{pr fdiv def} is achieved at a von Neumann measurement (see Proposition \ref{prop:meas fdiv achievability} below).

When $f$ is operator convex, the inequalities in \eqref{measured fdiv inequalities} can be continued as
\begin{align}\label{measured fdiv inequalities2}
S_f^{\pro}(\rho\|\sigma)\le S_f^{\meas}(\rho\|\sigma)\le S_f(\rho\|\sigma)
\end{align}
for any $\rho,\sigma$, according to Proposition \ref{P-3.8}.
It is an interesting open question whether the first inequality holds as an equality 
for a general operator convex function $f$ and every $\rho,\sigma$. This has been shown very recently in \cite{BFT_variational} to be true for 
\begin{align}\label{Renyi functions}
f(x)=f_{\alpha}(x)=s(\alpha)x^{\alpha}\ds\text{for}\ds\alpha\in(0,+\infty),\ds\ds\ds\text{and}
\ds\ds\ds f(x)=\eta(x)=x\log x,
\end{align}
(cf.~Example \ref{E-3.4}); we will give some further insight into this result after Theorem \ref{thm:measured}.

On the other hand, equality in the second inequality in \eqref{measured fdiv inequalities2} turns out to be very restrictive; indeed, under some mild technical 
conditions on $f$, $S_f^{\meas}(\rho\|\sigma)= S_f(\rho\|\sigma)$ implies that $\rho$ and $\sigma$ commute, in which case all the 
inequalities in \eqref{measured fdiv inequalities2} hold trivially as equalities. We will show this in Theorem 
\ref{thm:measured}, by combining a result by Petz \cite[Lemma 4.1]{mon_revisited} with Theorem \ref{T-3.12}. For this, we will
show that all 
the suprema in \eqref{meas fdiv2}--\eqref{vN fdiv def} are attained, an interesting fact in itself.
These will follow by simple compactness and continuity arguments. For \eqref{meas fdiv2}, we need some preparation first; namely, we show that it is sufficient to consider measurements with at most 
$(\dim\hil)^2$ outcomes.

\begin{lemma}\label{lemma:d square enough}
Let $\rho,\sigma\in\B(\hil)_+$ and $f:\,(0,+\infty)\to\bR$ be an operator convex function. 
For any measurement $\M=(M_x)_{x\in\X}$ on $\hil$, there exists a measurement $\wtilde M=(\wtilde M_k)_{k\in\{1,\ldots,(\dim\hil)^2\}}$ 
such that $S_f\bz\M(\rho)\|\M(\sigma)\jz\le S_f\bz\wtilde M(\rho)\|\wtilde \M(\sigma)\jz$.
As a consequence,
\begin{align}\label{meas fdiv3}
S_f^{\meas}(\rho\|\sigma)=\sup\{&S_f(\M(\rho)\|\M(\sigma)):\,\M\in\povm(\hil|[d^2])\},
\end{align}
where $[d^2]:=\{1,\ldots,(\dim\hil)^2\}$.
\end{lemma}
\begin{proof}
$\povm(\hil|\X)$ is a compact convex set of the finite-dimensional complex vector space 
$\B(\hil)^{\X}:=\{A:\,\X\to\B(\hil)\}$ (equipped with any norm).
Thus, any $\M\in\povm(\hil|\X)$ can be decomposed as $\M=\sum_{i\in\I}p_i\M^{(i)}$, where 
$\I$ is a finite set,
$p$ is a probability distribution on $\I$, and all the 
$\M^{(i)}$ are extremal points of $\povm(\hil|\X)$. Using the convexity of $S_f$ (Proposition \ref{P-3.6}), we get that 
\begin{align*}
S_f(\M(\rho)\|\M(\sigma))
\le
\sum_{i\in\I}p_i S_f(\M^{(i)}(\rho)\|\M^{(i)}(\sigma))
\le
S_f(\M^{(i)}(\rho)\|\M^{(i)}(\sigma))
\end{align*}
for some $i\in\I$.
Various characterizations of the extremal points of $\povm(\hil|\X)$ were given, e.g., in
\cite{Arveson,D'ArianoPrestiPerinotti,Parthasarathy}; in particular, it is known 
that if $\M^{(i)}$ is an extremal point of $\povm(\hil|\X)$ then $|\{x\in\X:\,M^{(i)}_x\ne 0\}|\le(\dim\hil)^2$.
Since $S_f(\M^{(i)}(\rho)\|\M^{(i)}(\sigma))$ only depends on the outcome probabilities 
$(\Tr M^{(i)}_x\rho)_{x\in\X}$ and $(\Tr M^{(i)}_x\sigma)_{x\in\X}$, we can assume without loss of generality that 
$\M^{(i)}$ has outcomes in $[d^2]$. From this, the assertion follows.
\end{proof}

\begin{remark}
Note that Lemma \ref{lemma:d square enough} holds for any convex function $f$ for 
which $S_f$ is jointly convex. According to Proposition \ref{P-3.6}, operator convexity of $f$ is sufficient for this, and
Remark \ref{rem:op conv necessary} shows that it is also likely to be necessary.
\end{remark}

Next, we want to show that the supremum in \eqref{meas fdiv3} is attained. 
Since $\povm(\hil|\X)$ is compact, the assertion would follow if the map $\M\mapsto S_f(\M(\rho)\|\M(\sigma))$ was continuous. 
This is not possible in general, since $S_f(\M(\rho)\|\M(\sigma))$ can  be $+\infty$, but with some care, these pathological cases can be treated as well.
The following observation about classical $f$-divergences will be useful in this direction:

\begin{remark}\label{rem:cl fdiv cont}
It is easy to see from the definition \eqref{persp ext} that 
$P_f$ is continuous on 
\begin{align*}
A_{\gamma_0,\gamma_1}:=\left\{(r\cos\gamma,r\sin\gamma):\,r\ge 0,\, \gamma_0\le\gamma\le\gamma_1\right\}
\end{align*}
for any 
$0<\gamma_0<\gamma_1<\pi/2$. If $f(0^+)<+\infty$ then $P_f$ is continuous also on $A_{\gamma_0,\pi/2}$ for any 
$0<\gamma_0<\pi/2$, and if $f'(+\infty)<+\infty$ then $P_f$ is continuous on $A_{0,\gamma_1}$ for any 
$0<\gamma_1<\pi/2$.
In particular, if $f(0^+)$ and $f'(+\infty)$ are both finite then 
$P_f$ is continuous on 
$\bR_+\times\bR_+$, where $\bR_+:=[0,+\infty)$.

Let $S_{f,\X}$ denote the classical $f$-divergence on $\bR_+^{\X}\times \bR_+^{\X}\equiv\bz\bR_+\times\bR_+\jz^{\X}$. By the above, we have 
\begin{align*}
S_{f,\X}\ds\text{is continuous on}\ds
\begin{cases}
A_{\gamma_0,\gamma_1}^{\X},&\text{for any}\ds 0<\gamma_0<\gamma_1<\pi/2,\\
A_{\gamma_0,\pi/2}^{\X},& \text{for any}\ds 0<\gamma_0<\pi/2,\ds\text{if}\ds f(0^+)<+\infty,\\
A_{0,\gamma_1}^{\X},& \text{for any}\ds 0<\gamma_1<\pi/2,\ds\text{if}\ds f'(+\infty)<+\infty.
\end{cases}
\end{align*}
\end{remark}


\begin{prop}\label{prop:meas fdiv achievability}
Let $\rho,\sigma\in\B(\hil)_+$ and $f:\,(0,+\infty)\to\bR$ be a convex function. 
Then the suprema in \eqref{pr fdiv def} and \eqref{vN fdiv def} are attained. 
If $f$ is also operator convex then the suprema in \eqref{meas fdiv3} and \eqref{meas fdiv2} are attained as well.
\end{prop}
\begin{proof}
It is enough to prove the assertions about \eqref{vN fdiv def} (due to Lemma \ref{lemma:vN=pr}) and about 
\eqref{meas fdiv3}. We start by proving the latter.

Note that for any $\M\in \povm(\hil|[d^2])$, 
\begin{align*}
S_f(\M(\rho)\|\M(\sigma))=\sum_{k=1}^{d^2}P_f(\Tr\rho M_k,\Tr\sigma M_k)
=
S_{f,[d^2]}\bz(\Tr M_k\rho)_{k=1}^{d^2},(\Tr M_k\sigma)_{k=1}^{d^2}\jz.
\end{align*}
For fixed $\rho$ and $\sigma$, the map 
\begin{align}\label{meas map}
\povm(\hil|[d^2])\ni M\mapsto \bz(\Tr\rho M_k)_{k=1}^{d^2},(\Tr\sigma M_k)_{k=1}^{d^2}\jz\in\bR_+^{d^2}\times\bR_+^{d^2}\equiv(\bR_+\times\bR_+)^{d^2}
\end{align}
is continuous.
Thus, by Remark \ref{rem:cl fdiv cont}, if $f(0^+)$ and $f'(+\infty)$ are both finite then 
the map $\M\mapsto S_f(\M(\rho)\|\M(\sigma))$ is continuous on the compact set 
$\povm(\hil|[d^2])$, and therefore the supremum in \eqref{meas fdiv3} is attained.
Hence, the only thing left is to prove the assertion when $f(0^+)$ and $f'(+\infty)$ are not both finite.

If $f(0^+)=+\infty$ and $\sigma^0\nleq\rho^0$ then for the two-outcome measurement
$\M=(\rho^0,I-\rho^0)$ we have 
$S_f(\rho\|\sigma)=+\infty=S_f(\M(\rho)\|\M(\sigma))$ (see Corollary \ref{cor:fdiv infty}), from which 
it is trivial that the supremum in \eqref{meas fdiv3} is attained. 
Similarly, if $f'(+\infty)=+\infty$ and $\rho^0\nleq\sigma^0$ then we can choose
$\M=(\sigma^0,I-\sigma^0)$ to arrive at the same conclusion. 

Hence, for the rest we assume that 
$\sigma^0\le\rho^0$ when $f(0^+)=+\infty$, and $\rho^0\le\sigma^0$ when $f'(+\infty)=+\infty$.
Note that if $\sigma^0\le\rho^0$ then there exists a positive constant $c_1>0$ such that 
$\sigma\le c_1\rho$, and hence for any measurement operator $M$, $\Tr\sigma M\le c_1\Tr\rho M$.
This means that the map in \eqref{meas map} maps $\povm(\hil|[d^2])$ into $A_{0,\gamma_1}^{[d^2]}$, where 
$\gamma_1:=\arctan c_1$. Similarly, if $\rho^0\le\sigma^0$ then there exists a $c_0>0$ such that 
the map in \eqref{meas map} maps $\povm(\hil|[d^2])$ into $A_{\gamma_0,\pi/2}^{[d^2]}$, where 
$\gamma_0:=\arctan c_0$. Hence we see that in the remaining cases, the map in \eqref{meas map} maps $\povm(\hil|[d^2])$ into
a domain on which $P_{f,[d^2]}$ is continuous, and thus we can use continuity and compactness again to conclude that the supremum in \eqref{meas fdiv3} is attained.

The proof of the assertion about \eqref{vN fdiv def} goes almost the same way.
Let $d:=\dim\hil$, and equip $\hil^d:=\times_{i=1}^d\hil$ with the product topology. Let
$\ONB(\hil)$ be the set of all ONB's $(e_i)_{i=1}^d$ of $\hil$. Then $\ONB(\hil)$ is a compact subset of $\hil^d$, and 
\begin{align*}
\ONB(\hil)\ni (e_i)_{i=1}^d\longmapsto 
\bz(\inner{e_i}{\rho e_i})_{i=1}^d,\inner{e_i}{\sigma e_i})_{i=1}^d\jz
\end{align*}
is continuous. Repeating the above argument with this map in place of the one in \eqref{meas map}, and 
$S_{f,[d]}$ in place of $S_{f,[d^2]}$, yields the assertion.
\end{proof}

\medskip

Now we are ready to prove the following:

\begin{thm}\label{thm:measured}
Let $\rho,\sigma\in\B(\hil)_+$ be such that $\rho^0\le\sigma^0$. The following are equivalent:
\begin{enumerate}
\item\label{meas1}
$S_f(\rho\|\sigma)= S_f^{\meas}(\rho\|\sigma)$ for some operator convex function
$f$ on $[0,+\infty)$ such that 
\begin{align*}
|\supp\mu_f|\ge \abs{\spec\bz L_{\rho}R_{\sigma\inv}\jz}+(\dim\hil)^2.
\end{align*}

\item\label{meas2}
$\rho\sigma=\sigma\rho$.
\item\label{meas3}
$S_f(\rho\|\sigma)= S_f^{\pro}(\rho\|\sigma)$ for all convex functions $f:\,(0,+\infty)\to\bR$.

\item\label{meas4}
$S_f(\rho\|\sigma)= S_f^{\pro}(\rho\|\sigma)$ for 
a continuous operator convex function
$f$ on $[0,+\infty)$ such that 
\begin{align*}
|\supp\mu_f|\ge \abs{\spec\bz L_{\rho}R_{\sigma\inv}\jz}+\dim\hil.
\end{align*}
\end{enumerate}
\end{thm}
\begin{proof}
The implications \ref{meas2}\,$\imp$\,\ref{meas3}\,$\imp$\,\ref{meas4}, and 
\ref{meas2}\,$\imp$\,\ref{meas1} are obvious. Assume that \ref{meas1} or \ref{meas4} holds; then, by Proposition \ref{prop:meas fdiv achievability}, there exists a 
measurement $\M$ such that $S_f(\rho\|\sigma)=S_f(\M(\rho)\|\M(\sigma))$.
Then, by Theorem \ref{T-3.12}, 
$S_f(\rho\|\sigma)= S_f(\M(\rho)\|\M(\sigma))$ for $f(x):=-x^{1/2}$. A straightforward modification of the argument 
by Petz in \cite[Lemma 4.1]{mon_revisited} (to avoid the assumption
$\sigma>0$) then shows \ref{meas2}.
%
\end{proof}

\medskip

It is a very natural requirement for a quantum divergence to be invariant under isometric
embeddings of a system into a larger system. It is easy to see that both quantum $f$-divergences
$S_f$ and $\maxdiv{f}$ have this invariance property, i.e.,
\begin{align*}
S_f(V\rho V^*\|V\sigma V^*)=S_f(\rho\|\sigma)
\end{align*}
for any $\rho,\sigma\in\B(\hil)_+$ and any isometry $V:\,\hil\to\kil$, and the same holds true
for $\maxdiv{f}$. 
It is easy to see that the same holds also for the measured $f$-divergence
$S_f^{\meas}$. However, it is not clear whether $S_f^{\pro}$ has the same invariance property.
In fact, the next proposition says that this is equivalent to the equality $S_f^\meas=S_f^\pro$.

\begin{prop}\label{prop:isometric invariance}
For every $\rho,\sigma\in\BH_+$ and any convex function $f$ on $(0,+\infty)$, we have:
\begin{itemize}
\item[(1)] $S_f^\pro(V\rho V^*\|V\sigma V^*)\ge S_f^\pro(\rho\|\sigma)$ for any isometry
$V:\,\hil\to\kil$.
\item[(2)] $S_f^{\meas}(\rho\|\sigma)=
\sup\left\{S_f^{\pro}(V\rho V^*\|V\sigma V^*):\,V\text{ isometry}\right\}$.
\item[(3)] The following (i) and (ii) are equivalent:
\begin{itemize}
\item[(i)] $S_f^{\pro}(V\rho V^*\|V\sigma V^*)=S_f^\pro(\rho\|\sigma)$ for any isometry $V$;
\item[(ii)] $S_f^{\meas}(\rho\|\sigma)=S_f^\pro(\rho\|\sigma)$.
\end{itemize}
\end{itemize}
\end{prop}

\begin{proof}
(1)\enspace
For every projective measurement $\M=(P_x)_{x\in\X}$ on $\hil$ one can define a projective
measurement $\M_V=(Q_x)_{x\in\X\cup\{x_0\}}$, $x_0\notin\X$, on $\kil$, by
$$
Q_x:=VP_xV^*\ds\text{for}\ds x\in\X,\qquad Q_{x_0}:=I_\kil-VV^*.
$$
From $\Tr Q_xV\rho V^*=\Tr P_x\rho$ and $\Tr Q_{x_0}V\rho V^*=0$ as well as the same for
$V\sigma V^*$, it follows that $S_f(\M_V(V\rho V^*)\|\M_V(V\sigma V^*))
=S_f(\M(\rho)\|\M(\sigma))$, implying (1).

(2)\enspace
The inequality $\ge$ is obvious since
$$
S_f^\meas(\rho\|\sigma)=S_f^\meas(V\rho V^*\|V\sigma V^*)
\ge S_f^\pro(V\rho V^*\|V\sigma V^*)
$$
for any isometry $V$. For the converse, for any measurement $\M=(M_x)_{x\in\X}$ on $\hil$, by
Naimark's dilation theorem, we get an isometry $V:\,\hil\to\kil$ and a projective measurement
$\overline\M=(P_x)_{x\in\X}$ on $\kil$ such that $M_x=V^*P_xV$ for all $x\in\X$. Since
$\Tr M_x\rho=\Tr P_xV\rho V^*$, we have
$$
S_f(\M(\rho)\|\M(\sigma))=S_f(\overline\M(V\rho V^*)\|\overline\M(V\sigma V^*))
\le S_f^\pro(V\rho V^*\|V\sigma V^*).
$$

(3) is immediate from (1) and (2).
\end{proof}

It is easy to see that monotonicity implies invariance under isometries, but not the other way around; an example for the latter is $S_{f_{\alpha}}$ with $\alpha>2$, that is invariant under isometries but not monotone \cite[Page 5]{Renyi_new}.
We say that a quantum $f$-divergence $S_f^q$ is \ki{invariant under partial isometries} if 
\begin{align*}
S_f^q(V\rho V^*\|V\sigma V^*)=S_f^q(\rho\|\sigma)
\end{align*}
for any $\rho,\sigma\in\B(\hil)_+$ and any partial isometry $V:\,\hil\to\kil$ such that 
$\rho^0,\sigma^0\le V^*V$. It is easy to see that any $f$-divergence $S_f$ is invariant under partial isometries.
Proposition \ref{prop:isometric invariance} yields the following:
\begin{cor}
For a convex function $f$ on $(0,+\infty)$, the following are equivalent:
\begin{enumerate}
\item\label{inv-mon1}
$S_f^{\pro}$ is invariant under partial isometries;
\item\label{inv-mon2}
$S_f^{\pro}$ is invariant under isometries;
\item\label{inv-mon3}
$S_f^{\pro}=S_f^{\meas}$;
\item\label{inv-mon4}
$S_f^{\pro}$ is monotone under positive trace-preserving maps;
\item\label{inv-mon5}
$S_f^{\pro}$ is monotone under CPTP maps.
\end{enumerate}
\end{cor}
\begin{proof}
\ref{inv-mon1}\,$\imp$\,\ref{inv-mon2} is trivial, \ref{inv-mon2}\,$\imp$\,\ref{inv-mon3}
follows from Proposition \ref{prop:isometric invariance}, 
\ref{inv-mon3}\,$\imp$\,\ref{inv-mon4} is trivial as $S_f^{\meas}$ is 
monotone under positive trace-preserving maps, and 
\ref{inv-mon4}\,$\imp$\,\ref{inv-mon5} is again trivial.
Assume now that \ref{inv-mon5} holds, and let $\rho,\sigma\in\B(\hil)_+$ and $V:\,\hil\to\kil$ be a partial isometry such that $\rho^0,\sigma^0\le V^*V$. We use a construction from \cite[Section 4.6.3]{Wilde_book}
to prove \ref{inv-mon1}.
For fixed states 
$\tau_{\hil}\in\B(\hil),\,\tau_{\kil}\in\B(\kil)$, define 
$\map_{\hil\to\kil}(\cdot):=V(\cdot)V^*+\tau_{\kil}\Tr (\cdot)(I-V^*V)$ and 
$\map_{\kil\to\hil}(\cdot):=V^*(\cdot)V+\tau_{\hil}\Tr (\cdot)(I-VV^*)$. 
Then $\map_{\hil\to\kil}$ and $\map_{\kil\to\hil}$ are 
CPTP maps such that 
$\map_{\hil\to\kil}(\rho)=V\rho V^*$, 
$\map_{\kil\to\hil}(V\rho V^*)=\rho$, and similarly for $\sigma$. The assumed monotonicity of $S_f^{\pro}$ then yields
$S_f^{\pro}(\rho\|\sigma)\le S_f^{\pro}(V\rho V^*\|V\sigma V^*)\le S_f^{\pro}(\rho\|\sigma)$, proving \ref{inv-mon1}.
\end{proof}
\bigskip

Analogously to the corresponding definitions for $f$-divergences, one can define the measured versions of the R\'enyi divergences as 
\begin{align}
D_{\alpha}^{\meas}(\rho\|\sigma):=\sup\{&D_{\alpha}(\M(\rho)\|\M(\sigma)):\,\M\text{ measurement on }\hil\},\label{measured Renyi def}\\
D_{\alpha}^{\pro}(\rho\|\sigma):=\sup\{&D_{\alpha}(\M(\rho)\|\M(\sigma)):\,\M\text{ projective measurement on }\hil\}\label{pr measured Renyi def}
\end{align}
for every $\alpha\in(0,+\infty)$, where $D_1(\rho\|\sigma):=\frac{1}{\Tr\rho}S(\rho\|\sigma)$, according to \eqref{Renyi limit}.
For $\alpha\ne 1$, these are simply functions of $S_{f_{\alpha}}^{\meas}$ and $S_{f_{\alpha}}^{\pro}$, respectively.
Note that $D_{\alpha}$ is monotone non-increasing under measurements for $\alpha\in(0,2]$ according to \eqref{Renyi div as fdiv} and Proposition \ref{P-3.8}.
While for $\alpha>2$, $D_{\alpha}$ is not monotone under CPTP maps, it is still monotone under measurements, as it has been shown in 
\cite[Section 3.7]{Hayashibook}. Thus, it is meaningful to take the suprema in the definitions \eqref{measured Renyi def} and \eqref{pr measured Renyi def}.

Now we review the results of \cite{BFT_variational} on the equality $S_f^{\pro}=S_f^{\meas}$ for the functions
$f_{\alpha}$ and $\eta$ in \eqref{Renyi functions}. The key ingredients are the following
variational expressions, given in \cite[Lemma 3]{BFT_variational}:
\begin{align}\label{f alpha variational}
S_{f_{\alpha}}^{\pro}(\rho\|\sigma)=\sup_{\omega\in\B(\hil)_{++}}
\begin{cases}
s(\alpha)\alpha\Tr\rho\omega+s(\alpha)(1-\alpha)\Tr\sigma \omega^{\frac{\alpha}{\alpha-1}},&\alpha\in(0,1/2),\\
s(\alpha)\alpha\Tr\rho\omega^{\frac{\alpha-1}{\alpha}}+s(\alpha)(1-\alpha)\Tr\sigma \omega,&\alpha
\in [1/2,+\infty).
\end{cases}
\end{align}
Here, note that the above expressions hold for general $\rho,\sigma\in\BH_+$, though proved
in \cite{BFT_variational} under the assumption $\rho^0\le\sigma^0$. In fact, if
$\rho^0\not\le\sigma^0$, then both sides of \eqref{f alpha variational} are $+\infty$ when
$\alpha>1$, and \eqref{f alpha variational} when $\alpha\in(0,1)$ follows by taking
$\sup_{\ep>0}$ of both sides of the expression for $\rho$ and $\sigma+\eps I$, noting that
$S_{f_\alpha}^\pro(\rho\|\sigma)=\sup_{\ep>0}S_{f_\alpha}^\pro(\rho\|\sigma+\eps I)$.

The following Proposition \ref{prop:pr=meas} is the same as Theorems 2 and 4 in
\cite{BFT_variational}; here we provide a proof based on \eqref{f alpha variational} and
Proposition \ref{prop:isometric invariance}\,(2), different from the one in
\cite{BFT_variational}. 
\begin{prop}\label{prop:pr=meas}
Let $f=f_{\alpha}$ for $\alpha\in(0,+\infty)$ or $f(x)=\eta(x)=x\log x$. Then
$S_f^{\meas}=S_f^{\pro}$, and hence
$D_\alpha^{\meas}=D_\alpha^{\pro}$ for every $\alpha\in(0,+\infty)$.
\end{prop}
\begin{proof}
Assume that $\alpha\in(0,1/2)$. By \eqref{f alpha variational}, for any $\rho,\sigma\in\BH_+$
and any isometry $V:\,\hil\to\kil$, we have 
\begin{align*}
S_{f_{\alpha}}^{\pro}(V\rho V^*\|V\sigma V^*)&=
\sup_{\omega\in\B(\hil)_{++}}
s(\alpha)\alpha\Tr V\rho V^*\omega+s(\alpha)(1-\alpha)\Tr V\sigma V^*\omega^{\frac{\alpha}{\alpha-1}}.
\end{align*}
Since $x\mapsto s(\alpha)(1-\alpha)x^{\frac{\alpha}{\alpha-1}}$ is operator concave, we have 
$s(\alpha)(1-\alpha) V^*\omega^{\frac{\alpha}{\alpha-1}}V\le s(\alpha)(1-\alpha) \bz V^*\omega V\jz^{\frac{\alpha}{\alpha-1}}$; see, e.g., 
\cite[Theorem 2.5.7]{Hiai_book}. Moreover, $\omega>0$ implies $V^*\omega V>0$, and hence
\begin{align*}
S_{f_{\alpha}}^{\pro}(V\rho V^*\|V\sigma V^*)&\le
\sup_{\omega\in\B(\hil)_{++}}
s(\alpha)\alpha\Tr \rho \bz V^*\omega V\jz+s(\alpha)(1-\alpha)\Tr \sigma \bz V^*\omega V\jz^{\frac{\alpha}{\alpha-1}}\\
&=
\sup_{\tilde\omega\in\B(\kil)_{++}}
s(\alpha)\alpha\Tr \rho \tilde \omega+s(\alpha)(1-\alpha)\Tr \sigma \tilde\omega^{\frac{\alpha}{\alpha-1}}\\
&=
S_{f_{\alpha}}^{\pro}(\rho\|\sigma).
\end{align*}
Taking now the supremum over all isometries $V$ and using
Proposition \ref{prop:isometric invariance}\,(2), we get the assertion. The proof for
$\alpha\ge 1/2$ goes the same way.

Consider now the R\'enyi divergences $D_{\alpha}$ defined in Example \ref{E-3.4}. By the above,
$S_{f_{\alpha}}^{\pro}=S_{f_{\alpha}}^{\meas}$ for all $\alpha\in(0,+\infty)$, and hence
$D_{\alpha}^{\pro}=D_{\alpha}^{\meas}$ for all $\alpha\in(0,+\infty)\setminus\{1\}$,
with the obvious definitions of the latter quantities; see \eqref{measured Renyi def} and
\eqref{pr measured Renyi def}. Moreover, \eqref{Renyi limit} implies
$$
{1\over\Tr\rho}\,S^{\pro}(\rho\|\sigma)
=\sup_{\alpha\in(0,1)}D_{\alpha}^{\pro}(\rho\|\sigma),\qquad
{1\over\Tr\rho}\,S^{\meas}(\rho\|\sigma)
=\sup_{\alpha\in(0,1)}D_{\alpha}^{\meas}(\rho\|\sigma).
$$
Combining these yields the assertion for $f=\eta$.
\end{proof}

\begin{remark}
From Propositions \ref{prop:isometric invariance}\,(3) and \ref{prop:pr=meas} we also see
that $S_{f_{\alpha}}^\pro$ and $S^\pro=S_{\eta}^\pro$ are invariant under isometries.
If one could prove these invariances directly, that would immediately imply Proposition \ref{prop:pr=meas},
again due to Proposition \ref{prop:isometric invariance}\,(3).
\end{remark}

\begin{remark}
In \cite{BFT_variational}, $S^{\pro}=S^{\meas}$ was proved using a separate variational expression for the projectively measured relative entropy
$S^{\pro}$. The same argument as above, using operator concavity and Proposition \ref{prop:isometric invariance}\,(2), could be applied to that variational formula to obtain
$S^{\pro}=S^{\meas}$; however, in the above proof we could proceed in a simpler way, without using the variational formula for $S^{\pro}$.
\end{remark}
\medskip

Consider also the sandwiched R\'enyi divergences, defined in \eqref{F-1.1}.
These quantities have been shown to be monotone non-increasing under CPTP maps for $\alpha\ge 1/2$ in 
\cite{Beigi,FL13,MO13,Renyi_new,WWY13}; in fact, for $\alpha\ge 1$, they are also monotone under positive trace-preserving maps
\cite{Beigi,M-HR}. By the Araki-Lieb-Thirring inequality \cite{Araki,Lieb-Thirring}, we have
\begin{align}\label{ALT}
D_{\alpha}^*(\rho\|\sigma)\le D_{\alpha}(\rho\|\sigma)
\end{align}
for any $\rho,\sigma$ and $\alpha\in(0,+\infty)$ \cite{WWY13}, with equality if and only if
$\rho$ commutes with $\sigma$ or $\alpha=1$ \cite{Hi1}. In particular, monotonicity of
$D_{\alpha}^*$ for $\alpha>2$ and \eqref{ALT} give an alternative proof for the non-increasing
property of $D_{\alpha}$ under measurements for $\alpha>2$. More generally, if
$\map:\,\B(\hil)\to\B(\kil)$ is a  positive trace-preserving map, and $\rho,\sigma\in\B(\hil)_+$
are such that $\map(\rho)$ and $\map(\sigma)$ commute, then 
$D_{\alpha}(\map(\rho)\|\map(\sigma))\le D_{\alpha}^*(\rho\|\sigma)\le D_{\alpha}(\rho\|\sigma)$
for any $\alpha\in[1/2,+\infty)$; in particular, 
$D_{\alpha}(\map(\rho)\|\map(\sigma))\le D_{\alpha}(\rho\|\sigma)$ also for $\alpha>2$.

It is straightforward from the definition of the measured R\'enyi divergence that for any fixed 
$\rho,\sigma$, $n\mapsto D_{\alpha}^{\meas}(\rho^{\otimes n}\|\sigma^{\otimes n})$ is superadditive, and hence 
\begin{align}\label{measured Renyi limit}
\ol D_{\alpha}^{\meas}(\rho\|\sigma):=
\sup_{n\in\bN}\frac{1}{n}D_{\alpha}^{\meas}\bz\rho^{\otimes n}\|\sigma^{\otimes n}\jz
=
\lim_{n\to\infty}\frac{1}{n}D_{\alpha}^{\meas}\bz\rho^{\otimes n}\|\sigma^{\otimes n}\jz.
\end{align}
We call $\ol D_{\alpha}^{\meas}$ the \ki{regularized measured R\'enyi divergence}.
Moreover, for $\alpha\ge 1/2$ we have 
\begin{align}\label{measured Renyi limit2}
\ol D_{\alpha}^{\meas}(\rho\|\sigma)=D_{\alpha}^*(\rho\|\sigma);
\end{align}
see \cite{HP} for $\alpha=1$, \cite{MO13} for $\alpha>1$, and \cite{HT14} for $\alpha\in[1/2,1)$. 
For $\alpha\in(0,1/2)$ this is no longer true, and instead we have 
$D_{\alpha}^*(\rho\|\sigma)\le D_{\alpha}^{\meas}(\rho\|\sigma)$, with strict inequality for non-commuting 
$\rho,\sigma$, as it has been shown very recently in \cite[Theorem 7]{BFT_variational}.
However, it is true for any $\alpha\in(0,+\infty)$ and any $\rho,\sigma\in\B(\hil)_+$ 
that there exists a sequence of measurements $\M_n$ on $\hil^{\otimes n}$, $n\in\bN$, such that 
\begin{align*}
D_{\alpha}^*(\rho\|\sigma)=\lim_{n\to\infty}\frac{1}{n}D_{\alpha}\bz\M_n(\rho^{\otimes n})\|\M_n(\sigma^{\otimes n})\jz.
\end{align*}
Such a measurement can be chosen as a von Neumann measurement in a common eigenbasis of
$\sigma^{\otimes n}$ and $\P_{\sigma^{\otimes n}}(\rho^{\otimes n})$, where 
$\P_{\sigma^{\otimes n}}$ is the pinching by the spectral projections of $\sigma^{\otimes n}$;
see \cite{HP} for $\alpha=1$,
\cite[Theorem 3.7]{MO13} for $\alpha>1$, and Lemma 3 and Corollary 4 in \cite{HT14}
for $\alpha\in(0,1)$.

The relations of the various quantum R\'enyi divergences mentioned above can be summarized as follows:

\begin{prop}\label{prop:Renyi relations}
For any $\rho,\sigma\in\B(\hil)_+$, we have 
\begin{align}
&D_{\alpha}^*(\rho\|\sigma)\le D_{\alpha}^{\pro}(\rho\|\sigma)=D_{\alpha}^{\meas}(\rho\|\sigma)\le
\ol D_{\alpha}^{\meas}(\rho\|\sigma)\le D_{\alpha}(\rho\|\sigma),\ds\ds\ds \alpha\in(0,1/2),
\label{Renyi chain1}\\
&D_{\alpha}^{\pro}(\rho\|\sigma)=D_{\alpha}^{\meas}(\rho\|\sigma)\le
\ol D_{\alpha}^{\meas}(\rho\|\sigma)=
D_{\alpha}^*(\rho\|\sigma)\le D_{\alpha}(\rho\|\sigma),\ds\ds\ds \alpha\in[1/2,+\infty).
\label{Renyi chain2}
\end{align}
If $\rho$ and $\sigma$ commute or $D_{\alpha}(\rho\|\sigma)=+\infty$ then all the inequalities above are equalities;
otherwise all the inequalities are strict, except the first inequality in \eqref{Renyi chain2}
for $\alpha=1/2$, the last inequality in \eqref{Renyi chain2} for $\alpha=1$, 
and 
possibly the last two inequalities in 
\eqref{Renyi chain1}, of which at least one is strict.
%
\end{prop}
\begin{proof}
When $\rho$ and $\sigma$ commute or $D_{\alpha}(\rho\|\sigma)=+\infty$ then it is straightforward that all the above quantities are equal to each other, and hence for the rest we assume the contrary.

The relations $D_{\alpha}^*(\rho\|\sigma)\le D_{\alpha}^{\pro}(\rho\|\sigma)=D_{\alpha}^{\meas}(\rho\|\sigma)$ for $\alpha\in(0,1/2)$, with strict inequality for non-commuting $\rho,\sigma$, as well as 
$D_{\alpha}^{\pro}(\rho\|\sigma)=D_{\alpha}^{\meas}(\rho\|\sigma)$ for $\alpha\in[1/2,+\infty)$,
and the strict inequality 
$D_{\alpha}^{\meas}(\rho\|\sigma)<D_{\alpha}^*(\rho\|\sigma)$ for non-commuting $\rho,\sigma$ and 
$\alpha>1/2$, 
were proved in \cite{BFT_variational}.

If the last two inequalities in \eqref{Renyi chain1} are both equalities then we also have 
$D_{\alpha}^{\pro}(\rho\|\sigma)=D_{\alpha}(\rho\|\sigma)$, and 
$\rho\sigma=\sigma\rho$ follows by Theorem \ref{thm:measured} applied to $f_{\alpha}(t)=-t^{\alpha}$.
Finally, the last inequality in \eqref{Renyi chain2} and its equality case follow from the Araki-Lieb-Thirring inequality and its equality case, as discussed above.
%
\end{proof}

\begin{remark}
The case $\alpha=1/2$ is special 
in the sense that
$D_{1/2}^*=-2\log F$, where $F$ is the fidelity,
so that $D_{1/2}^*=D_{1/2}^{\meas}$; see, e.g., \cite[Chapter 9]{NC}.
\end{remark}

\begin{remark}
It is an interesting open problem to find a closed expression for
$\ol D_{\alpha}^{\meas}(\rho\|\sigma)$
for $\alpha\in(0,1/2)$; one possible candidate is $D_{\alpha}(\rho\|\sigma)$, based on \eqref{Renyi chain1}.
This is related to another question left open in the above proposition, namely whether both of the last two inequalities in \eqref{Renyi chain1} are strict for non-commuting $\rho$ and $\sigma$.
\end{remark}

\begin{remark}
Note that both the standard and the sandwiched R\'enyi divergences are additive, i.e., 
$D_{\alpha}\bz\rho^{\otimes n}\|\sigma^{\otimes n}\jz=nD_{\alpha}(\rho\|\sigma)$, 
$D_{\alpha}^*\bz\rho^{\otimes n}\|\sigma^{\otimes n}\jz=nD_{\alpha}^*(\rho\|\sigma)$
for all $\rho,\sigma$, all $n\in\bN$, and all $\alpha\in(0,+\infty)$. 
By Proposition \ref{prop:Renyi relations} and \eqref{measured Renyi limit}--\eqref{measured Renyi limit2}, we see that the
measured R\'enyi divergences are not additive for $\alpha>1/2$; more precisely, if
$\rho\sigma\ne\sigma\rho$ then for every $\alpha>1/2$ there exists an $n\in\bN$ such that 
$D_{\alpha}^{\meas}(\rho^{\otimes n}\|\sigma^{\otimes n})>nD_{\alpha}^{\meas}(\rho\|\sigma)$.
It is an open question whether the same holds for $\alpha\in(0,1/2)$.
\end{remark}

\bigskip
We close this section by proving the strict positivity of $f$-divergences (when properly normalized as $f(1)=0$) on pairs of quantum states. More precisely, we prove a Pinsker-type inequality for the projectively measured $f$-divergences. While we don't use it in the rest of the paper, it is interesting in its own right.

The quantum version of the \ki{Pinsker} (or \ki{Pinsker-Csisz\'ar}) \ki{inequality}
\begin{align}\label{Cs-P}
{1\over2}\,\|\rho-\sigma\|_1^2\le S(\rho\|\sigma)
\end{align}
for quantum states $\rho,\sigma$ was first shown in \cite{HOT}, where $\|\cdot\|_1$ denotes
the trace-norm. The following proposition is not only a generalization to general
$f$-divergences, but it also strengthens \eqref{Cs-P} even in the case of the relative entropy,
according to Theorem \ref{thm:measured}.

\begin{prop}\label{Pinsker}
Let $f$ be an operator convex function on $(0,+\infty)$ with $f(1)=0$. Then for every
density operators $\rho,\sigma$ on $\hil$,
$$
{f''(1)\over2}\,\|\rho-\sigma\|_1^2\le S_f^{\pro}(\rho\|\sigma),
$$
where $S_f^{\pro}(\rho\|\sigma)$ is given in \eqref{pr fdiv def}. Hence,
${f''(1)\over2}\,\|\rho-\sigma\|_1^2\le S_f^q(\rho\|\sigma)$ holds for every quantum
$f$-divergence in the sense stated in Section 3.1. Here, $f''(1)>0$ if and only if $f$ is
non-linear.
\end{prop}

\begin{proof}
Let $(e_i)_{i=1}^d$ be an orthonormal basis consisting of eigenvectors of $\rho-\sigma$, and define 
$\E(X):=\sum_{i=1}^d\inner{e_i}{Xe_i}\pr{e_i}$, $X\in\B(\hil)$.
Set $p:=\cE(\rho)=\sum_{i=1}^d p_i\pr{e_i}$ and
$q:=\cE(\sigma)=\sum_{i=1}^d q_i\pr{e_i}$.
Since $\rho-\sigma=\cE(\rho-\sigma)=p-q$ and
$S_f(p\|q)\le S_f^{\pro}(\rho\|\sigma)$, it suffices to show that
\begin{equation}\label{pinsk-0}
{f''(1)\over2}\,\|p-q\|_1^2\le S_f(p\|q).
\end{equation}
Although this is known \cite[Theorem 3]{Gil} for a more general class of convex functions $f$,
we have, for operator convex $f$, the following simple proof based on the integral expression
in \eqref{F-2.3}. As easily verified, note that
\begin{equation}\label{pinsk-1}
f''(1)=2\biggl(c+\int_{[0,+\infty)}{1\over1+s}\,d\lambda(s)\biggr),
\end{equation}
which shows that $f''(1)>0$ if and only if $f$ is non-linear. We may assume by continuity
that $p,q>0$, and we have the expression
\begin{equation}\label{pinsk-2}
S_f(p\|q)=c\sum_{i=1}^d{(p_i-q_i)^2\over q_i}
+\int_{[0,+\infty)}\sum_{i=1}^d{(p_i-q_i)^2\over p_i+sq_i}\,d\lambda(s).
\end{equation}
We estimate
\begin{align}
\sum_{i=1}^d|p_i-q_i|&=\sum_{i=1}^d{|p_i-q_i|\over\sqrt{q_i}}\,\sqrt{q_i}
\le\Biggl(\sum_{i=1}^d{(p_i-q_i)^2\over q_i}\Biggr)^{1/2}
\Biggl(\sum_{i=1}^dq_i\Biggr)^{1/2} \nonumber\\
&=\Biggl(\sum_{i=1}^d{(p_i-q_i)^2\over q_i}\Biggr)^{1/2}, \label{pinsk-3}
\end{align}
and for every $s\in[0,+\infty)$,
\begin{align}
\sum_{i=1}^d|p_i-q_i|&=\sum_{i=1}^d{|p_i-q_i|\over\sqrt{p_i+sq_i}}\,\sqrt{p_i+sq_i}
\le\Biggl(\sum_{i=1}^d{(p_i-q_i)^2\over p_i+sq_i}\Biggr)^{1/2}
\Biggl(\sum_{i=1}^d(p_i+sq_i)\Biggr)^{1/2} \nonumber\\
&=\Biggl(\sum_{i=1}^d{(p_i-q_i)^2\over p_i+sq_i}\Biggr)^{1/2}(1+s)^{1/2}. \label{pinsk-4}
\end{align}
Combining \eqref{pinsk-1}--\eqref{pinsk-4} yields \eqref{pinsk-0}.
\end{proof}

\section{Reversibility via R\'enyi divergences}
\label{sec:Renyi rev}

The notion of the {\it $\alpha$-$z$-R\'enyi relative entropy} was first introduced in
\cite[Section 3.3]{JOPP}, and further studied in \cite{AD}. It is defined for two positive
operators $\rho,\sigma\in\BH_+$ with $\rho^0\le\sigma^0$ as
\begin{align*}
D_{\alpha,z}(\rho\|\sigma)
:=
{1\over\alpha-1}\log\Tr\Bigl(\sigma^{1-\alpha\over2z}
\rho^{\alpha\over z}\sigma^{1-\alpha\over2z}\Bigr)^z
=
{1\over\alpha-1}\log\Tr\Bigl(\rho^{\alpha\over 2z}\sigma^{1-\alpha\over z}\rho^{\alpha\over 2z}\Bigr)^z,
\end{align*}
for any $\alpha\in\bR\setminus\{1\}$ and $z>0$. Below we restrict to the case $\alpha,z>0$ with $\alpha\ne1$. The above definition can be extended to general $\rho,\sigma\in\BH_+$ as 
\begin{align}\label{az def2}
D_{\alpha,z}(\rho\|\sigma)
:=
\lim_{\ep\searrow 0}
{1\over\alpha-1}\log\Tr\Bigl(\rho^{\alpha\over 2z}(\sigma+\ep I)^{1-\alpha\over z}\rho^{\alpha\over 2z}\Bigr)^z.
\end{align}
\begin{lemma}
The limit in \eqref{az def2} exists, and is equal to 
\begin{align*}
\begin{cases}
{1\over\alpha-1}\log\Tr\Bigl(\rho^{\alpha\over 2z}\sigma^{1-\alpha\over z}\rho^{\alpha\over 2z}\Bigr)^z,& \alpha\in(0,1)\text{ or } \rho^0\le \sigma^0,\\
+\infty,&\text{otherwise}.
\end{cases}
\end{align*}
\end{lemma}
\begin{proof}
The only slightly non-trivial part of the claim is when $\alpha>1$ and $\rho^0\nleq\sigma^0$. In this case, there exists a unit vector $\psi\perp\supp\sigma$ such that 
$\inner{\psi}{\rho^0\psi}>0$. Note that $(\sigma+\ep I)^{1-\alpha\over z}\ge \ep^{1-\alpha\over z}\pr{\psi}$, and thus 
\begin{align*}
\Tr\Bigl(\rho^{\alpha\over 2z}(\sigma+\ep I)^{1-\alpha\over z}\rho^{\alpha\over 2z}\Bigr)^z
\ge
\ep^{1-\alpha}\Tr\Bigl(\rho^{\alpha\over 2z}\pr{\psi}\rho^{\alpha\over 2z}\Bigr)^z
&=
\ep^{1-\alpha}\Tr\Bigl(\pr{\psi}\rho^{\alpha\over z}\pr{\psi}\Bigr)^z\\
&=
\ep^{1-\alpha}\inner{\psi}{\rho^{\alpha\over z}\psi}^z,
\end{align*}
that tends to $+\infty$ as $\ep\searrow 0$.
\end{proof}

We also introduce the notation
\begin{align*}
Q_{\alpha,z}(\rho\|\sigma):=\lim_{\ep\searrow 0}\Tr\Bigl(\rho^{\alpha\over 2z}(\sigma+\ep I)^{1-\alpha\over z}\rho^{\alpha\over 2z}\Bigr)^z
=
\begin{cases}
\Tr\Bigl(\rho^{\alpha\over 2z}\sigma^{1-\alpha\over z}\rho^{\alpha\over 2z}\Bigr)^z,& \alpha\in(0,1)\text{ or } \rho^0\le \sigma^0,\\
+\infty,&\text{otherwise},
\end{cases}
\end{align*}
so that 
\begin{align*}
D_{\alpha,z}(\rho\|\sigma)=\frac{1}{\alpha-1}\log Q_{\alpha,z}(\rho\|\sigma).
\end{align*}

The $\alpha$-$z$-R\'enyi relative entropies have the following monotonicity property:
For any $\rho,\sigma\in\B(\hil)_+$, and any CPTP map $\map:\,\B(\hil)\to\B(\kil)$, 
\begin{align}\label{az mon}
D_{\alpha,z}(\map(\rho)\|\map(\sigma))\le D_{\alpha,z}(\rho\|\sigma),
\end{align}
whenever
\begin{itemize}
\item[(a)]\ds $0<\alpha<1$,\ds $z\ge \max\{\alpha,1-\alpha\},$\ds or
\item[(b)]\ds $1<\alpha\le 2$,\ds $z=1$,\ds or
\item[(c)]\ds $1< \alpha=z$,\ds or
\item[(d)]\ds $1<\alpha\le 2$,\ds $z=\alpha/2$.
\end{itemize}
See \cite{Hi3} for the proof of (a), \cite{An} for (b), \cite{Beigi,FL13} for (c),
and \cite{CFL} for (d) (cf.~also \cite[Theorem 1]{AD}).

The {\it sandwiched R\'enyi divergence} introduced in \cite{Renyi_new,WWY13} is
$$
D_\alpha^*(\rho\|\sigma):={1\over\alpha-1}\log\Tr\Bigl(
\sigma^{1-\alpha\over2\alpha}\rho
\sigma^{1-\alpha\over2\alpha}\Bigr)^\alpha,
$$
which is the $z=\alpha$ case of the $\alpha$-$z$-R\'enyi relative entropy. From
\eqref{az mon} in cases (a) and (c) we have for any $\alpha\in[1/2,+\infty)\setminus\{1\}$
\begin{align}\label{sa mon}
D_\alpha^*(\map(\rho)\|\map(\sigma))\le D_\alpha^*(\rho\|\sigma).
\end{align}

As presented in \eqref{F-1.1} and \eqref{F-1.2} in the Introduction, the formulas of
$D_\alpha^*(\rho\|\sigma)$ and $D_{\alpha,z}(\rho\|\sigma)$ are often given with division
by $\Tr\rho$ inside the logarithm. However, the difference between with or without this
division is irrelevant to our discussions on the monotonicity inequality and the characterization
of its equality case. Thus, we here adopt, for the sake of simplicity, the definitions
without the division by $\Tr\rho$.

\renewcommand\theenumi{(E\arabic{enumi})}
In this section, we shall prove monotonicity \eqref{az mon} in some special cases
of $\rho,\sigma$ and $\Phi$, for some ranges of $\alpha,z$, including values not covered in previous works. 
Our main result is the characterization of equality in the monotonicity inequality \eqref{az mon}
in these cases.
For the latter, we will consider the following possible characterizations:
\begin{enumerate}
\item[(E0)]\ds $D_\alpha^*(\Phi(\rho)\|\Phi(\sigma))=D_\alpha^*(\rho\|\sigma)$,
\item\label{E1}\ds $D_{\alpha,z}(\map(\rho)\|\map(\sigma))= D_{\alpha,z}(\rho\|\sigma)$,
\item\label{E2}\ds $\map^*\bz\map(\rho)\jz=\rho,\ds\map^*\bz\map(\sigma)\jz=\sigma$,
\item\label{E3}\ds $\map_{\sigma}^*\bz\map(\rho)\jz=\rho,\ds\map_{\sigma}^*\bz\map(\sigma)\jz=\sigma$, \ds
(see \eqref{rverse-map} for the map $\map_{\sigma}^*$),
\item\label{E4}\ds $\map_{\rho}^*\bz\map(\rho)\jz=\rho,\ds\map_{\rho}^*\bz\map(\sigma)\jz=\sigma$,
\item\label{E5}\ds there exists a unitary $U$ such that  $\map(\rho)=U\rho U^*,\s\map(\sigma)=U\sigma U^*$.
\end{enumerate}
\renewcommand\theenumi{(\roman{enumi})}

\begin{thm}\label{thm:az mon rev}
Let $\rho,\sigma\in\B(\hil)_+$, and let $\map:\,\B(\hil)\to\B(\hil)$ be a bistochastic map.
The monotonicity inequality \eqref{az mon} holds if at least one of the following conditions is satisfied:
\begin{enumerate}
\item\label{az mon rev(i)} $\alpha\le z\le1$\ds and\ds $\sigma\in\fix{\map}$,
\item\label{az mon rev(ii)} $0<1-\alpha\le z\le1$\ds and\ds $\rho\in\fix{\map}$,
\item\label{az mon rev(iii)} $\alpha\ge z\ge\max\{1,\alpha/2\}$\ds and \ds $\sigma\in\fix{\map}$,
\item\label{az mon rev(iv)} $\alpha>1$, $z\ge\max\{1,\alpha-1\}$,\ds and\ds$\rho\in\fix{\map}$.
\end{enumerate}

If we also assume that $\map$ is $2$-positive, then we have the following characterizations of equality in the monotonicity inequality:
\begin{enumerate}
\renewcommand\theenumi{(\arabic{enumi})}
\item\label{az mon rev(1)}
If \ref{az mon rev(i)} holds with $\rho^0\le\sigma^0$, then we have
\ref{E1}\,$\iff$\,\ref{E2}\,$\iff$\,\ref{E3}\,$\iff$\,\ref{E5}.
\item\label{az mon rev(2)}
If \ref{az mon rev(i)} holds with $z\ne\alpha$ and $\sigma^0\le\rho^0$, then we have
\ref{E1}\,$\iff$\,\ref{E4}.
\item\label{az mon rev(3)}
If \ref{az mon rev(ii)} holds with $\sigma^0\le\rho^0$, then we have
\ref{E1}\,$\iff$\,\ref{E2}\,$\iff$\,\ref{E4}\,$\iff$\,\ref{E5}.
\item\label{az mon rev(4)}
If \ref{az mon rev(ii)} holds with $z\ne 1-\alpha$ and $\rho^0\le\sigma^0$, then we have
\ref{E1}\,$\iff$\,\ref{E3}.
\item\label{az mon rev(5)}
If \ref{az mon rev(iii)} holds with $\rho^0\le\sigma^0$, then we have
\ref{E1}\,$\iff$\,\ref{E2}\,$\iff$\,\ref{E3}\,$\iff$\,\ref{E5}.
\item\label{az mon rev(6)}
If \ref{az mon rev(iv)} holds with $z\ne\alpha-1$ and $\sigma^0=I$, then we have
\ref{E1}\,$\iff$\,\ref{E3}.
\end{enumerate}

Moreover, without the assumption that $\map$ is $2$-positive, we have
\ref{E1}\,$\iff$\,\ref{E2}\,$\iff$\,\ref{E3} in \ref{az mon rev(1)} and \ref{az mon rev(5)}, and
\ref{E1}\,$\iff$\,\ref{E2}\,$\iff$\,\ref{E4} in \ref{az mon rev(3)}.
\end{thm}
\medskip

Before giving the proof of Theorem \ref{thm:az mon rev}, we give some remarks and a corollary.

\begin{remark}
For $\alpha>1$, unless $\rho^0\le\sigma^0$, we have $D_{\alpha,z}(\rho\|\sigma)=+\infty$,
so that the monotonicity inequality \eqref{az mon} holds trivially, while the preservation
of $D_{\alpha,z}(\rho\|\sigma)$ has no implication on reversibility in general.
Therefore, in cases \ref{az mon rev(iii)} and \ref{az mon rev(iv)}, reversibility cannot be
obtained in general, if instead of the conditions in \ref{az mon rev(5)} and \ref{az mon rev(6)} above, one assumes $\sigma^0\le\rho^0$ as in \ref{az mon rev(2)} or
\ref{az mon rev(3)}.
\end{remark}

\begin{remark}
Note that in the cases \ref{az mon rev(2)}, \ref{az mon rev(4)}, and \ref{az mon rev(6)}
in Theorem \ref{thm:az mon rev}, we do not get \ref{E5} in general.
For instance, with the notations of Theorem \ref{thm:decomposition} where
$\hil_1=\hil_2=\bigoplus_{k=1}^r\hil_{k,L}\otimes\hil_{k,R}$, let $\rho,\sigma$, and
$\map$ be given as
\begin{align*}
\rho=\bigoplus_k \rho_k\otimes\omega_k,\ds\ds\ds
\sigma=\bigoplus_k \sigma_k\otimes\omega_k,\ds\ds\ds
\map(X_k\otimes Y_k)=U_kX_kU_k^*\otimes\eta_k(Y_k),
\end{align*}
 such that
the conditions of \ref{az mon rev(2)} are satisfied. Then it is clear that 
$\map^*_{\rho}(\map(\sigma))=\sigma$, i.e., \ref{E4} holds.
Now, assume that $\sigma_i=0$ for some $i$; then the condition 
$\sigma\in\fix{\map}$ imposes no restriction on $\eta_i$.
Hence, for this $i$, we can take $\omega_i$ and $\eta_i$ so that 
the spectrum of $\eta_i(\omega_i)$ 
is different from the spectrum of $\omega_i$, while for all $k\ne i$, 
$\eta_k(\cdot)=V_k\cdot V_k^*$ with some unitaries $V_k$.
 Then it is clear that there exists no unitary $U$ such that 
$\map(\rho)=U\rho U^*$, i.e., \ref{E5} does not hold.
\end{remark}

From the $z=\alpha$ case of the above theorem we have

\begin{cor}\label{cor:sa mon rev}
Let $\rho,\sigma\in\B(\hil)_+$, and let $\map:\,\B(\hil)\to\B(\hil)$ be a bistochastic map. The monotonicity inequality \eqref{sa mon} holds if one of the following conditions is satisfied:
\begin{enumerate}
\item\label{special monotonicity}
 $\sigma\in\fix{\Phi}$ (for arbitrary $\alpha\in(0,+\infty)\setminus\{1\}$),
\item\label{special monotonicity2}
 $1/2\le\alpha<1$ and $\rho\in\fix{\Phi}$.
\end{enumerate}

If $\Phi$ is $2$-positive, then we have the following characterizations of equality in the monotonicity inequality:
\begin{itemize}
\item[(1)] If (i) holds with $\rho^0\le\sigma^0$, then we have
(E0)\,$\iff$\,\ref{E2}\,$\iff$\,\ref{E3}\,$\iff$\,\ref{E5}.
\item[(2)] If (ii) holds with $\sigma^0\le\rho^0$, then we have
(E0)\,$\iff$\,\ref{E2}\,$\iff$\,\ref{E4}\,$\iff$\,\ref{E5}.
\item[(3)] If (ii) holds with $\alpha\ne1/2$ and $\rho^0\le\sigma^0$, then we have
(E0)\,$\iff$\,\ref{E3}\,$\iff$\,\ref{E5}.
\end{itemize}

Moreover, without the assumption that $\Phi$ is $2$-positive, we have
(E0)\,$\iff$\,\ref{E2}\,$\iff$\,\ref{E3} in (1) and (E0) $\iff$\,\ref{E2}\,$\iff$\,\ref{E4} in (2).
\end{cor}

\renewcommand\theenumi{(\alph{enumi})}
\begin{remark}
\begin{enumerate}
\item
Note that the monotonicity in \ref{special monotonicity2}, and in \ref{special monotonicity} for $\alpha\ge 1/2$ above are special cases of the general monotonicity \eqref{sa mon} for $\alpha\ge 1/2$, although they are derived in a different way than the known proofs of \eqref{sa mon}. On the other hand, the monotonicity 
\eqref{sa mon} does not hold in general for $\alpha\in(0,1/2)$ (see \cite[Section IV]{BFT_variational}), 
and hence for this range of $\alpha$, the monotonicity in \ref{special monotonicity} does not follow from known monotonicity results.
\item
A special case of the monotonicity in \ref{special monotonicity} of Corollary \ref{cor:sa mon rev} above is the monotonicity under the pinching by the spectral projections of $\sigma$, given in \cite[Proposition 14]{Renyi_new}.

\item
Concurrently to our paper, Jen\v{c}ov\'a \cite{Jencova_rev16} proved the characterization 
``(E0) for some $\alpha>1$ $\iff$ \ref{E3}'' when $\rho^0\le\sigma^0$, from which the characterizations in 
Corollary \ref{cor:sa mon rev} follow easily when $\alpha>1$.
\end{enumerate}
\end{remark}
\renewcommand\theenumi{(\roman{enumi})}

The case $\alpha=2$ is special, as reversibility can be obtained easily from the preservation of 
$D_2^*$, as has been shown very recently in \cite[Lemma 2]{Jencova_rev16}. 
Below we give a different proof.

\begin{prop}
Let $\rho,\sigma\in\BH_+$  with $\rho^0\le\sigma^0$ and 
$\map:\,\BH\to\BK$ be a $2$-positive trace-preserving (not necessarily bistochastic)
map. Then 
\begin{align*}
D_2^*(\Phi(\rho)\|\Phi(\sigma))=D_2^*(\rho\|\sigma)\ds\iff\ds
\map_{\sigma}^*(\map(\rho))=\rho.
\end{align*}
\end{prop}
\begin{proof}
 When $\rho,\sigma\in\BH_+$ are density operators with $\sigma>0$, since
$$
\bigl\<\rho-\sigma,L_\sigma^{-1/2}R_\sigma^{-1/2}(\rho-\sigma)\bigr\>_\HS
=\Tr(\rho-\sigma)\sigma^{-1/2}(\rho-\sigma)\sigma^{-1/2}
=\Tr\bigl(\sigma^{-1/4}\rho\sigma^{-1/4}\bigr)^2-1,
$$
one finds that $D_2^*(\Phi(\rho)\|\Phi(\sigma))=D_2^*(\rho\|\sigma)$ if and only if
$$
\bigl\<\Phi(\rho-\sigma),\Omega_{\Phi(\sigma)}^\kappa(\Phi(\rho-\sigma))\bigr\>_\HS
=\bigl\<\rho-\sigma,\Omega_\sigma^\kappa(\rho-\sigma)\bigr\>_\HS
$$
with $\kappa(x):=x^{-1/2}$. By virtue of \ref{rev cond10} of Theorem \ref{T-3.12}, this
implies that $\Phi$ is reversible on $\{\rho,\sigma\}$ if and only if
$D_2^*(\Phi(\rho)\|\Phi(\sigma))=D_2^*(\rho\|\sigma)$. This result can immediately be
extended to general $\rho,\sigma\in\BH_+$ with $\rho^0\le\sigma^0$ by normalizing
$\rho,\sigma$ and restricting $\Phi$ to $\sigma^0\BH\sigma^0=\cB(\sigma^0\cH)$.
%
\end{proof}
\medskip

Before we give the proof of Theorem \ref{thm:az mon rev}, we need some preparation, given below.
For any self-adjoint $X\in\B(\hil)$, we denote by $\lambda\down(X):=(\lambda\down_1(X),\ldots,\lambda\down_d(X))$ the vector of the decreasingly ordered eigenvalues of $X$, where $d:=\dim\hil$.
The following Lemmas \ref{lemma:maj1} and \ref{lemma:maj2} are standard; we include their proofs for readers' convenience.

\begin{lemma}\label{lemma:maj1}
For any $X\in\B(\hil)_{\sa}$, and any $k\in\{1,\ldots,d\}$,
\begin{align}\label{majorization1}
\lambda\down_1(X)+\ldots+\lambda\down_k(X)=\max\{\Tr XA:\,0\le A\le I,\,\Tr A\le k\}.
\end{align}
\end{lemma}
\begin{proof}
We have $X=\sum_{i=1}^d\lambda\down_i(X)\pr{e_i}$ for some orthonormal basis $\{e_i\}_{i=1}^d$, and hence
for any $0\le A\le I$ such that $\Tr A\le k$, 
we have $\Tr XA=\sum_{i=1}^d\lambda\down_i(X)\inner{e_i}{Ae_i}\le\sum_{i=1}^k\lambda\down_i(X)$,
since $\inner{e_i}{Ae_i}\le 1$ and $\sum_{i=1}^d\inner{e_i}{Ae_i}\le k$.
The equality in \eqref{majorization1} is attained by $A:=\sum_{i=1}^k\pr{e_i}$.
\end{proof}

\begin{lemma}\label{lemma:maj2}
Let $X\in\B(\hil)_{\sa}$ and $\map:\,\B(\hil)\to\B(\hil)$ be a bistochastic map.
Then $\lambda\down(\map(X))$ is majorized by $\lambda\down(X)$, in notation 
$\lambda\down(\map(X))\prec\lambda\down(X)$, i.e., 
for all $k=1,\ldots,d$,
\begin{align*}
\lambda\down_1(\map(X))+\ldots+\lambda\down_k(\map(X))
\le
\lambda\down_1(X)+\ldots+\lambda\down_k(X),
\end{align*}
with equality for $k=d$. Hence, there exist permutations $\pi_k\in S_d$ and probability weights 
$p_k>0$, $k=1,\ldots,r$, such that all the vectors
$\bigl(\lambda_{\pi_k(i)}^{\downarrow}(X)\bigr)_{i=1}^d$ for $k=1,\dots,r$ are different, and 
\begin{align}\label{BvN}
\lambda_i^{\downarrow}(\map(X))=\sum_{k=1}^r p_k\lambda_{\pi_k(i)}^{\downarrow}(X),\ds\ds\ds i=1,\ldots,d.
\end{align}
\end{lemma}
\begin{proof}
By \eqref{majorization1}, 
\begin{align*}
\lambda\down_1(\map(X))+\ldots+\lambda\down_k(\map(X))
&=
\max\{\Tr \map(X)A:\,0\le A\le I,\,\Tr A\le k\}\\
&=
\max\{\Tr X\map^*(A):\,0\le A\le I,\,\Tr A\le k\}\\
&\le
\sup\{\Tr XB:\,0\le B\le I,\,\Tr B\le k\}\\
&=
\lambda\down_1(X)+\ldots+\lambda\down_k(X),
\end{align*}
where we used that $0\le \map^*(A)\le\map^*(I)=I$, and $\Tr\map^*(A)=\Tr A\le k$. The
majorization relation just established yields immediately the second assertion (see, e.g.,
\cite[Theorem 4.1.1]{Hiai_book} for a proof).
\end{proof}

The following lemma can be considered as an analogue of Lemma \ref{Choi inequality}, where
operator convexity is relaxed to ordinary convexity, on the expense of replacing the
positive semidefinite order with the trace order, and requiring that $\map$ is also
trace-preserving.

\begin{lemma}\label{lemma:bistochastic strictly convex}
Let $\map:\,\B(\hil)\to\B(\hil)$ be a bistochastic map, $X\in\B(\hil)_{\sa}$, and 
let $f$ be a convex function on an interval containing $\spec (X)$. Then
\begin{align*}
\Tr f\bz\map(X)\jz\le\Tr f(X).
\end{align*}
If $f$ is strictly convex, then equality holds if and only if there exists a unitary $U$ such that $\map(X)=UXU^*$.
\end{lemma}
\begin{proof}
By \eqref{BvN}, we have
\begin{align}
\Tr f(\map(X))&=\sum_{i=1}^d f\bz\sum_{k=1}^r p_k\lambda_{\pi_k(i)}^{\downarrow}(X)\jz 
\le
\sum_{i=1}^d \sum_{k=1}^r p_k f\bz\lambda_{\pi_k(i)}^{\downarrow}(X)\jz \label{sand rev ineq}\\
&=
\sum_{k=1}^r p_k \sum_{i=1}^d f\bz\lambda_{\pi_k(i)}^{\downarrow}(X)\jz 
=
\Tr f(X).\nonumber
\end{align}
Moreover, if $f$ is strictly convex, then the inequality in \eqref{sand rev ineq} is strict, unless $r=1$. Hence, if equality holds in \eqref{sand rev ineq}, then $r=1$, which means
that $\lambda_i^{\downarrow}(\map(X))=\lambda_i^{\downarrow}(X)$, $1\le i\le d$. This
implies  the existence of a unitary $U$ such that $\map(X)=UXU^*$. Conversely, if
$\map(X)=UXU^*$ for some unitary $U$ then it is obvious that $\Tr f(\map(X))=\Tr f(X)$.
\end{proof}

\begin{lemma}\label{lemma:convexity equality}
Let $\map:\,\B(\hil)\to\B(\hil)$ be a bistochastic map and let $X\in\B(\hil)_{\sa}$. Then the following are equivalent:
\begin{enumerate}
\item\label{sand rev1}
$\Tr f\bz\map(X)\jz=\Tr f(X)$ for all real functions $f$ on $\spec (X)$.

\item\label{sand rev2}
$\Tr f\bz\map(X)\jz=\Tr f(X)$ for some strictly convex or strictly concave $f$ on an interval containing 
$\spec (X)$.

\item\label{sand rev2.1}
$\Tr\map(X)^2=\Tr X^2$.

\item\label{sand rev3}
$\map(X)=UXU^*$ for some unitary $U$.

\item\label{sand rev4}
$(\map^*\circ\map)(X)=X$.
\end{enumerate}
\end{lemma}
\begin{proof}
\ref{sand rev1}\,$\imp$\,\ref{sand rev2.1}\,$\imp$\,\ref{sand rev2} is trivial, \ref{sand rev2}\,$\imp$\,\ref{sand rev3} follows from Lemma \ref{lemma:bistochastic strictly convex}, and
\ref{sand rev3}\,$\imp$\,\ref{sand rev1} is obvious.
Hence, it is enough to show \ref{sand rev2.1}\,$\iff$\,\ref{sand rev4}.
Note that for any $X\in\B(\hil)_{\sa}$,
\begin{align*}
0&\le \Tr X^2-\Tr\map(X)^2
=\inner{X}{X}_{\HS}-\inner{\map(X)}{\map(X)}_{\HS}
=\inner{X}{(I-(\map^*\circ\map))X}_{\HS},
\end{align*}
where the first inequality is due to Lemma \ref{lemma:bistochastic strictly convex}.
Note that $\map^*\circ\map$ is positive semidefinite with respect to the Hilbert-Schmidt inner product, and 
the above inequality shows that $\map^*\circ\map\le I$ ($=I_{\BH}$).
Hence, $I-(\map^*\circ\map)$ is positive semidefinite, and thus the above can be written as
\begin{align*}
0&\le \Tr X^2-\Tr\map(X)^2=
\norm{\bz I-(\map^*\circ\map)\jz^{1/2}X}_{\HS}^2.
\end{align*}
Thus,
\begin{align*}
\Tr X^2=\Tr\map(X)^2\ds\iff\ds \bz I-(\map^*\circ\map)\jz^{1/2}X=0\ds\iff\ds
X=(\map^*\circ\map)X.
\end{align*}
\end{proof}

\begin{cor}\label{adjoint fixpoint}
For any bistochastic map $\map:\,\B(\hil)\to\B(\hil)$,
\begin{align*}
\fix{\map}=\fix{\map^*}.
\end{align*}
\end{cor}
\begin{proof}
We show that $\fix{\map}\subseteq\fix{\map^*}$, which implies $\fix{\map}=\fix{\map^*}$
since $(\map^*)^*=\map$. Assume that $X\in\fix{\map}$; then $X_1:=\half(X+X^*)$ and
$X_2:=\frac{1}{2i}(X-X^*)$ are also in $\fix{\map}$, and hence we can assume without loss
of generality that $X^*=X\in\fix{\map}$. Then, by  \ref{sand rev3}\,$\imp$\,\ref{sand rev4} of
Lemma \ref{lemma:convexity equality}, we get that $X=(\map^*\circ\map)(X)=\map^*(X)$, i.e.,
$X\in\fix{\map^*}$.
\end{proof}

When $\map$ is $2$-positive, the unitary in \ref{sand rev3} of Lemma \ref{lemma:convexity equality} can be chosen independently of $X$:
\begin{lemma}\label{lemma:convexity equality2}
Let $\map:\,\B(\hil)\to\B(\hil)$ be a $2$-positive bistochastic map. Then there exists a unitary $U$ such that 
\begin{align*}
\map(X)=UXU^*,\ds\ds\ds \map^*(X)=U^*XU,\ds\ds\ds X\in\fix{\map^*\circ\map}.
\end{align*}
\end{lemma}
\begin{proof}
Note that with the notations of Theorem \ref{thm:decomposition}, 
we have $\map_I=\map$ and $\map_I^*=\map^*$, and hence
$\fix{\map_I^*\circ\map}=\fix{\map^*\circ\map_I}$. This in turn yields that $\omega_k=I_{1,k,R}$ for all
$k$ in the decomposition \eqref{dec3}, and unitality of $\map$ and $\map^*$ yields 
$\eta_k(I_{1,k,R})=I_{2,k,R}$ and $\eta_k^*(I_{2,k,R})=I_{1,k,R}$
for all $k$. Defining 
$U:=\bigoplus_k U_k\otimes I_{1,k,R}$ then gives the desired unitary.
\end{proof}

\begin{remark}
The statement of the above lemma may not hold when $\map$ is only assumed to be positive, as one can easily see by choosing $\map$ to be the transposition in some 
orthonormal basis.
\end{remark}
\medskip

Now we are in a position to prove Theorem \ref{thm:az mon rev}.

\medskip\noindent
{\it Proof of Theorem \ref{thm:az mon rev}.}\enspace The proof is divided into several steps.

\medskip
(1a)\enspace
Assume \ref{az mon rev(i)} first.
Since $0<{\alpha\over z}\le1$, we have $\Phi(\rho^{\alpha\over z})\le\Phi(\rho)^{\alpha\over z}$ due to 
Lemma \ref{Choi inequality}, so that
\begin{align}\label{mainthm proof0}
\sigma^{1-\alpha\over2z}\Phi(\rho^{\alpha\over z})\sigma^{1-\alpha\over2z}
\le\sigma^{1-\alpha\over2z}\Phi(\rho)^{\alpha\over z}\sigma^{1-\alpha\over2z}.
\end{align}
Hence,
\begin{align}\label{mainthm proof1}
\Tr\Bigl(\sigma^{1-\alpha\over2z}\Phi(\rho^{\alpha\over z})\sigma^{1-\alpha\over2z}\Bigr)^z
\le
\Tr\Bigl(\sigma^{1-\alpha\over2z}\Phi(\rho)^{\alpha\over z}\sigma^{1-\alpha\over2z}\Bigr)^z
=
\Tr\Bigl(\map(\sigma)^{1-\alpha\over2z}\Phi(\rho)^{\alpha\over z}\map(\sigma)^{1-\alpha\over2z}\Bigr)^z,
\end{align}
due to $\sigma\in\fix{\map}$. Using also that $\map$ is bistochastic, 
Lemma \ref{lemma:fix-mult} yields that $\sigma^{1-\alpha\over2z}\in\fix{\map}\subseteq\M_{\map}$, and hence
\begin{align}\label{mainthm proof2}
\sigma^{1-\alpha\over2z}\Phi(\rho^{\alpha\over z})\sigma^{1-\alpha\over2z}
=
\Phi\bigl(\sigma^{1-\alpha\over2z}\bigr)\Phi(\rho^{\alpha\over z})
\Phi\bigl(\sigma^{1-\alpha\over2z}\bigr)
=
\Phi\bigl(\sigma^{1-\alpha\over2z}\rho^{\alpha\over z}\sigma^{1-\alpha\over2z}\bigr).
\end{align}
Using Lemma \ref{lemma:bistochastic strictly convex}, and that $x\mapsto x^z$ is concave for $0<z\le 1$, we get
\begin{align}\label{mainthm proof3}
\Tr\bigl(\sigma^{1-\alpha\over2z}\rho^{\alpha\over z}\sigma^{1-\alpha\over2z}\bigr)^z
\le
\Tr\Bigl(\Phi\bigl(\sigma^{1-\alpha\over2z}\rho^{\alpha\over z}
\sigma^{1-\alpha\over2z}\bigr)\Bigr)^z.
\end{align}
Putting together \eqref{mainthm proof1}--\eqref{mainthm proof3}, we get the desired monotonicity inequality
\eqref{az mon}.

If \eqref{az mon} holds with equality (i.e., \ref{E1} holds), we must have
equalities in  \eqref{mainthm proof1} and \eqref{mainthm proof3}. In particular, equality
in \eqref{mainthm proof1}, together with the strict monotonicity of
$X\in\BH_+\mapsto\Tr X^z$, implies that \eqref{mainthm proof0} holds with equality.
Multiplying both sides of  \eqref{mainthm proof0} with $\sigma^{\alpha-1\over2z}$ yields
\begin{align}\label{proof (1a)}
\sigma^{0}\Phi(\rho^{\alpha\over z})\sigma^{0}
=\sigma^{0}\Phi(\rho)^{\alpha\over z}\sigma^{0}.
\end{align}

\medskip
(1b)\enspace
We assume (i) with $\rho^0\le\sigma^0$.
If \ref{E1} holds then, by the above, we have \eqref{proof (1a)}, and using that $\map(\rho)^0\le\map(\sigma)^0=\sigma^0$, we get
\begin{align*}
\Phi(\rho^{\alpha\over z})=\Phi(\rho)^{\alpha\over z}.
\end{align*}
Since this gives $\Tr\Phi(\rho)^{\alpha\over z}=\Tr\rho^{\alpha\over z}$, it follows from
Lemma \ref{lemma:convexity equality} that
\begin{align*}
\Phi^*\circ\Phi(\rho)=\rho
\end{align*}
whenever $\alpha<z$. When $\alpha=z$, since $0<\alpha<1$, equality in \eqref{mainthm proof3} implies by Lemma \ref{lemma:convexity equality} again that
$$
\Phi^*\circ\Phi\bigl(\sigma^{1-\alpha\over2\alpha}\rho\sigma^{1-\alpha\over2\alpha}\bigr)
=\sigma^{1-\alpha\over2\alpha}\rho\sigma^{1-\alpha\over2\alpha}.
$$
Since $\sigma\in\fix{\Phi}=\fix{\Phi^*}$ by Corollary \ref{adjoint fixpoint}, it follows from Lemma \ref{lemma:fix-mult} that $\sigma^{1-\alpha\over2\alpha}\in\M_\Phi\cap\M_{\Phi^*}$. Hence we have
$$
\sigma^0(\Phi^*\circ\Phi(\rho))\sigma^0=\sigma^0\rho\sigma^0,
$$
so that $\Phi^*\circ\Phi(\rho)=\rho$ since
$(\Phi^*\circ\Phi(\rho))^0\le(\Phi^*\circ\Phi(\sigma))^0=\sigma^0$. This proves \ref{E2} since
$\Phi^*\circ\Phi(\sigma)=\sigma$. In the converse direction, assume that \ref{E2} holds.
Then $\map^*(\map(\sigma))=\sigma=\map(\sigma)$. Applying the above established monotonicity
to $\map^*$ and $\map(\rho),\map(\sigma)$ in place of $\map$ and $\rho,\sigma$, we get
\begin{align*}
D_{\alpha,z}(\rho\|\sigma)&=
D_{\alpha,z}(\map^*(\map(\rho))\|\map^*(\map(\sigma)))
\le
D_{\alpha,z}(\map(\rho)\|\map(\sigma))
\le
D_{\alpha,z}(\rho\|\sigma),
\end{align*}
proving the equality in \eqref{az mon}. Hence \ref{E1}\,$\iff$\,\ref{E2}.
Finally, notice that $\sigma^{1/2}\in\fix{\map^*}\subseteq\M_{\map^*}$ implies
\begin{align*}
\map_{\sigma}^*(Y)
&=
\sigma^{1/2}\map^*\bz \map(\sigma)^{-1/2}Y\map(\sigma)^{-1/2}\jz\sigma^{1/2}
=
\map^*\bigl(\sigma^{1/2}\bigr)\map^*\bz \sigma^{-1/2}Y\sigma^{-1/2}\jz
\map^*\bigl(\sigma^{1/2}\bigr)\\
&=
\map^*\bz \sigma^{1/2}\sigma^{-1/2}Y\sigma^{-1/2}\sigma^{1/2}\jz
=
\map^*\bz\sigma^0 Y\sigma^0\jz
\end{align*}
for any $Y\in\B(\hil)$. In particular, if $\rho^0\le\sigma^0$ then $\map_{\sigma}^*(\map(\rho))=\map^*(\map(\rho))$, proving \ref{E2}\,$\iff$\,\ref{E3}.

\medskip
(1c)\enspace
In (1) we also assume that $\map$ is $2$-positive; then \ref{E2} implies \ref{E5}
by Lemma \ref{lemma:convexity equality2}. 
Obviously, \ref{E5} implies \ref{E1}, completing the proof of \ref{az mon rev(1)}.

\medskip
(2)\enspace
Assume (i) with $\sigma^0\le\rho^0$ instead of $\rho^0\le\sigma^0$, and $z\ne\alpha$. Assume that \ref{E1} holds; then 
by the argument in (1a), we have \eqref{proof (1a)}.
Multiplying with $\map(\sigma)^{\frac{z-\alpha}{2z}}$ from both sides, and taking the trace, we get
\begin{align*}
\Tr\map(\sigma)^{1-\frac{\alpha}{z}}\map(\rho)^{\alpha\over z}
=
\Tr\map(\sigma)^{1-\frac{\alpha}{z}}\map(\rho^{\alpha\over z})
=
\Tr\map(\sigma^{1-\frac{\alpha}{z}}\rho^{\alpha\over z})
=
\Tr\sigma^{1-\frac{\alpha}{z}}\rho^{\alpha\over z},
\end{align*}
where we have used again that $\sigma\in\fix{\map}\subseteq\M_{\map}$. This means that the quantum $f$-divergence $S_f$ is preserved, i.e.,
$$
S_f(\Phi(\sigma)\|\Phi(\rho))=S_f(\sigma\|\rho),
$$
where $f(x):=-x^{1-{\alpha\over z}}$ with $1-{\alpha\over z}\in(0,1)$. Hence, when $\Phi$ is $2$-positive, by Theorem \ref{T-3.12} we have \ref{E4}. 
Assume now that \ref{E4} holds, i.e., $\rho,\sigma\in\fix{\map^*_{\rho}\circ\map}$.
Using then Theorem \ref{thm:decomposition} (with the role of $\rho$ and $\sigma$ interchanged), the decomposition in 
\eqref{dec3} yields immediately \ref{E1}.

\medskip
(3) \& (4)\enspace
Assume now that \ref{az mon rev(ii)} holds.  Note that for $0<\alpha<1$, we have 
$Q_{\alpha,z}(\rho\|\sigma)=Q_{1-\alpha,z}(\sigma\|\rho)$ so that $D_{\alpha,z}(\rho\|\sigma)=D_{1-\alpha,z}(\sigma\|\rho)$ for all $z$.
Hence the claim about monotonicity follows from the case \ref{az mon rev(i)} already proved in (1a). When the monotonicity inequality holds with equality, the assertions in (3) and (4) also follow from (1) and (2) proved above, by interchanging $\rho$ and $\sigma$ together with changing $\alpha$ to $1-\alpha$.

\medskip
(5)\enspace
Next, assume \ref{az mon rev(iii)}. Since
$1\le\alpha/z\le2$, the function $x\mapsto x^{\alpha/z}$ is operator convex, and 
$x\mapsto x^z$ is convex. Thus, the inequalities in \eqref{mainthm proof0}, \eqref{mainthm proof1}, and \eqref{mainthm proof3} hold in the opposite directions, proving the monotonicity inequality \eqref{az mon} as in the proof (1a). When $\rho^0\le\sigma^0$, the proof for the equality case goes the same way as in the proofs (1b) and (1c) above, completing the proof of \ref{az mon rev(5)}.

\medskip
(6)\enspace
Finally, assume \ref{az mon rev(iv)}. Since
$$
Q_{\alpha,z}(\rho\|\sigma)=\lim_{\eps\searrow0}Q_{\alpha,z}(\rho\|\sigma+\eps I),\qquad
Q_{\alpha,z}(\Phi(\rho)\|\Phi(\sigma))
=\lim_{\eps\searrow0}Q_{\alpha,z}(\rho\|\Phi(\sigma+\eps I)),
$$
we may assume that $\sigma^0=I$, to show the monotonicity inequality \eqref{az mon}. Since
$-1\le{1-\alpha\over z}<0$, Lemma \ref{Choi inequality} yields
\begin{equation}\label{main proof ineq}
\rho^{\alpha\over2z}\Phi(\sigma)^{1-\alpha\over z}\rho^{\alpha\over2z}
\le
\rho^{\alpha\over2z}\Phi(\sigma^{1-\alpha\over z})\rho^{\alpha\over2z}.
\end{equation}
As in the proof (1a), this implies
\begin{align*}
\Tr\Bigl(\map(\rho)^{\alpha\over2z}\map(\sigma)^{1-\alpha\over z}\map(\rho)^{\alpha\over2z}\Bigr)^z
&=
\Tr\Bigl(\rho^{\alpha\over2z}\Phi(\sigma)^{1-\alpha\over z}\rho^{\alpha\over2z}\Bigr)^z
\le 
\Tr\Bigl(\rho^{\alpha\over2z}\Phi(\sigma^{1-\alpha\over z})\rho^{\alpha\over2z}\Bigr)^z\\
&=\Tr\Bigl(\Phi\bigl(\rho^{\alpha\over2z}\sigma^{1-\alpha\over z}
\rho^{\alpha\over2z}\bigr)\Bigr)^z
\le\Tr\bigl(\rho^{\alpha\over2z}\sigma^{1-\alpha\over z}\rho^{\alpha\over2z}\bigr)^z,
\end{align*}
where the equalities are due to $\rho\in\fix{\map}\subseteq\M_{\map}$, and the last inequality
follows from  Lemma \ref{lemma:bistochastic strictly convex}. This proves monotonicity
\eqref{az mon}. 

Now, if \eqref{az mon} holds with equality, then we have equality in
\eqref{main proof ineq}, and hence, 
\begin{align*}
\rho^0\Phi(\sigma)^{1-\alpha\over z}\rho^0=\rho^0\Phi(\sigma^{1-\alpha\over z})\rho^0.
\end{align*}
Therefore, similarly to the above proof of (2), we have
$$
\Tr\Phi(\rho)^{1-\frac{1-\alpha}{z}}\Phi(\sigma)^{1-\alpha\over z}
=\Tr\rho^{1-\frac{1-\alpha}{z}}\sigma^{1-\alpha\over z}.
$$
Assuming that $z\ne\alpha-1$, we have $1-\frac{1-\alpha}{z}\in(1,2)$, and the above equality
means that the $f$-divergence $S_f(\rho\|\sigma)$, corresponding to
$f(t)=t^{1-\frac{1-\alpha}{z}}$, is preserved by $\map$. Hence, by Theorem \ref{T-3.12},
we get \ref{E3}. The implication \ref{E3}\,$\imp$\,\ref{E1} 
follows by observing that \ref{E3} means $\rho,\sigma\in\fix{\map^*_{\sigma}\circ\map}$, and using 
the decomposition \eqref{dec3} in Theorem \ref{thm:decomposition}, similarly to the proof of (2) above.
\hfill\qed

\begin{remark}
We remark that the exclusion of $z=1-\alpha$ is essential in the statement
\ref{az mon rev(4)} of Theorem \ref{thm:az mon rev}. Let $\Phi:\cB(\bC^2)\to\cB(\bC^2)$ be the
diagonal pinching by regarding $\cB(\bC^2)$ as the $2\times2$ matrices. Let
$\rho:=\begin{bmatrix}1&0\\0&0\end{bmatrix}$, and for $a>b>0$ and
$0<\theta<{\pi\over2}$, let
$$
\sigma:=\begin{bmatrix}\cos\theta&-\sin\theta\\\sin\theta&\cos\theta\end{bmatrix}
\begin{bmatrix}a&0\\0&b\end{bmatrix}
\begin{bmatrix}\cos\theta&\sin\theta\\-\sin\theta&\cos\theta\end{bmatrix}
=\begin{bmatrix}a\cos^2\theta+b\sin^2\theta&(a-b)\cos\theta\sin\theta\\
(a-b)\cos\theta\sin\theta&a\sin^2\theta+b\cos^2\theta\end{bmatrix}.
$$
Then
$$
S(\rho\|\sigma)=-\Tr\rho\log\sigma=-(\log a)\cos^2\theta-(\log b)\sin^2\theta
$$
while
$$
S(\Phi(\rho)\|\Phi(\sigma))=-\Tr\rho\log\Phi(\sigma)=-\log(a\cos^2\theta+b\sin^2\theta)
<S(\rho\|\sigma),
$$
and hence $\Phi$ is not reversible for $\rho,\sigma$. However, for any $\alpha\in(0,1)$ we have
\begin{align*}
\Tr\bigl(\sigma^{1/2}\rho^{\alpha\over1-\alpha}\sigma^{1/2}\bigr)^{1-\alpha}
=(a\cos^2\theta+b\sin^2\theta)^{1-\alpha} 
=\Tr\bigl(\Phi(\sigma)^{1/2}\Phi(\rho)^{\alpha\over1-\alpha}\Phi(\sigma)^{1/2}\bigr)^{1-\alpha},
\end{align*}
so that
$D_{\alpha,z}(\Phi(\rho)\|\Phi(\sigma))=D_{\alpha,z}(\rho\|\sigma)$ for $z=1-\alpha$.
In particular when $\alpha=z=1/2$, since
$$
D_{1/2,1/2}(\rho\|\sigma)=-2\log F(\rho,\sigma)
$$
with the fidelity $F(\rho,\sigma):=\Tr|\rho^{1/2}\sigma^{1/2}|$, we notice that when
$\rho^0\le\sigma^0$ and $\rho\in\fix{\Phi}$, the equality
$F(\Phi(\rho),\Phi(\sigma))=F(\rho,\sigma)$ does not imply the reversibility of $\Phi$ on
$\rho,\sigma$ in general (cf.~also \cite[Corollary A.9]{MO13}). But it does so when
$\rho^0\le\sigma^0$ and $\sigma\in\fix{\Phi}$, as follows from \ref{az mon rev(1)} of
Theorem \ref{thm:az mon rev}.
\end{remark}

\begin{remark}
The max-relative entropy \cite{Datta} is defined as the limit of the sandwiched R\'enyi
divergences:
$$
D_{\max}(\rho\|\sigma):=\lim_{\alpha\to+\infty}D_\alpha^*(\rho\|\sigma)
=\inf\{\gamma:\rho\le e^\gamma\sigma\}.
$$
It is known \cite[Corollary A.9]{MO13} that preservation of the max-relative entropy
does not imply reversibility. Below we give an example that shows that reversibility
does not follow from the preservation of the max-relative entropy even in the very special
case where the second state is a fixed point of the map (cf.~\ref{az mon rev(1)} of Theorem
\ref{thm:az mon rev}). Consider $3\times3$ invertible density matrices
$\sigma=\diag(\mu_1,\mu_2,\mu_3)$ and
$$
\rho=\begin{bmatrix}\lambda&0&0\\0&a&c\\0&\overline c&b\end{bmatrix},\quad
\lambda,a,b>0,\ \ c\ne0,\ \ |c|^2<ab,\ \ \lambda+a+b=1.
$$
Let $\Phi:\bM_3\to\bM_3$ be the diagonal pinching, i.e., $\Phi(X)$ is the diagonal part of
$X$. Then
$$
\Tr\Phi(\rho)^2\Phi(\sigma)^{-1}={\lambda^2\over\mu_1}+{a^2\over\mu_2}+{b^2\over\mu_3},
$$
$$
\Tr\rho^2\sigma^{-1}={\lambda^2\over\mu_1}+{a^2+|c|^2\over\mu_2}+{b^2+|c|^2\over\mu_3},
$$
so that 
$$
\Tr\Phi(\rho)^2\Phi(\sigma)^{-1}<\Tr\rho^2\sigma^{-1}.
$$
By Theorem \ref{T-3.21}, this is equivalent to that
$\sigma^{-1/2}\rho\sigma^{-1/2}\not\in\cM_{\Phi_\sigma}$,
and it implies that $\Phi$ is not reversible on $\{\rho,\sigma\}$. On the other hand,
$$
\big\|\Phi(\sigma)^{-1/2}\Phi(\rho)\Phi(\sigma)^{-1/2}\big\|_\infty
=\max\biggl\{{\lambda\over\mu_1},{a\over\mu_2},{b\over\mu_3}\biggr\},
$$
\begin{align*}
\big\|\sigma^{-1/2}\rho\sigma^{-1/2}\big\|_\infty
&=\max\Biggl\{{\lambda\over\mu_1},\Bigg\|
\begin{bmatrix}\mu_2^{-1/2}&0\\0&\mu_3^{-1/2}\end{bmatrix}
\begin{bmatrix}a&c\\\overline c&b\end{bmatrix}
\begin{bmatrix}\mu_2^{-1/2}&0\\0&\mu_3^{-1/2}\end{bmatrix}\Bigg\|_\infty\Biggr\} \\
&=\max\Biggl\{{\lambda\over\mu_1},
{1\over2}\Biggl({a\over\mu_2}+{b\over\mu_3}
+\sqrt{\biggl({a\over\mu_2}-{b\over\mu_3}\biggr)^2
+{4|c|^2\over\mu_2\mu_3}}\Biggr)\Biggr\}.
\end{align*}
When the $\mu_i$'s are fixed and $\lambda\nearrow1$ (hence $a,b,|c|\searrow0$), we have
$$
\big\|\Phi(\sigma)^{-1/2}\Phi(\rho)\Phi(\sigma)^{-1/2}\big\|_\infty
={\lambda\over\mu_1}=\big\|\sigma^{-1/2}\rho\sigma^{-1/2}\big\|_\infty,
$$
which means that $D_{\max}(\Phi(\rho)\|\Phi(\sigma))=D_{\max}(\rho\|\sigma)$.
\end{remark}

\section{Closing remarks}

\begin{remark}\label{R-3.24}\rm
In this paper we treat $f$-divergences for general positive operators. 
We note that restricting to density operators would make no essential difference. Indeed,
for $\rho,\sigma\in\BH_+$, $\sigma\ne 0$, write $\rho=\alpha\rho_1$ and
$\sigma=\beta\sigma_1$, where $\rho_1,\sigma_1$ are density operators and $\alpha=\Tr \rho$,
$\beta=\Tr \sigma>0$. For any operator convex function $f$, since
$$
S_f(\rho\|\sigma)=S_{f_1}(\rho_1\|\sigma_1),\qquad
\maxdiv{f}(\rho\|\sigma)=\maxdiv{f_1}(\rho_1\|\sigma_1)
$$
with $f_1(x):=\beta f(\alpha\beta^{-1}x)$, one can easily obtain properties of $S_f$ and
$\maxdiv{f}$ for general positive operators from those restricting to density operators.
\end{remark}

\begin{remark}\label{R-3.25}\rm
We may treat trace-preserving positive linear maps $\Phi:\cA_1\to\cA_2$ between general
finite-dimensional $C^*$-algebras $\cA_1$, $\cA_2$. When $\cA_1\subseteq\BH$ and
$\cA_2\subseteq\BK$, we can extend $\Phi$ to $\widetilde\Phi:\BH\to\BK$ by
$\widetilde\Phi:=\Phi\circ\cE_{\cA_1}$, where $\cE_{\cA_1}:\BH\to\cA_1$ is the
trace-preserving conditional expectation onto $\cA_1$. Then
$\widetilde\Phi^*=\Phi^*\circ\cE_{\cA_2}$. It is straightforward to reformulate the results
of this paper for $\rho,\sigma\in\cA_1$ and $\widetilde\Phi$ into those for $\rho,\sigma$
and $\Phi$. Thus the generalization to the setting of finite-dimensional $C^*$-algebras is
automatic.
\end{remark}

\subsection*{Acknowledgments}
The work of FH\ was supported in part by Grant-in-Aid for Scientific Research 
(C)26400103 and (C)17K05266.
MM acknowledges support from 
the Spanish MINECO (Project No. FIS2013-40627-P), the Generalitat de Catalunya CIRIT (Project No. 2014 SGR 966),
the Hungarian Research Grant OTKA-NKFI K104206,
and the Technische Universit\"at M\"unchen -- Institute for Advanced Study, funded by the German Excellence Initiative and the European Union Seventh Framework Programme under grant agreement no. 291763. 
Part of this work was done while MM was with the Institute for Advanced Studies, and with the 
Zentrum Mathematik, M5, at the 
Technische Universit\"at M\"unchen.
The authors are grateful to Anna Jen\v cov\'a for advice on the extremal decomposition of POVMs, and to an anonymous referee for a 
very careful reading of the manuscript and for various suggestions that helped to improve both the content and the presentation 
of the paper.
\appendix

\section{Extension of Lemma \ref{lemma:persp properties}}
\label{sec:persp properties}

\begin{prop}\label{P-A.1}
Let $f$ be a real function on $(0,+\infty)$. The following conditions are equivalent:
\begin{itemize}
\item[(i)] $(A,B)\in\BH_{++}\times\BH_{++}\mapsto P_f(A,B)$ is jointly operator convex for any
finite-dimensional Hilbert space $\cH$;
\item[(ii)] for every $B\in\BH_{++}$, $A\in\BH_{++}\mapsto\Tr P_f(A,B)$ is convex for any
finite-dimensional Hilbert space $\cH$;
\item[(iii)] for every $A\in\BH_{++}$, $B\in\BH_{++}\mapsto\Tr P_f(A,B)$ is convex for any
finite-dimensional Hilbert space $\cH$;
\item[(iv)] $f$ is operator convex on $(0,+\infty)$;
\item[(v)] $\widetilde f$ is operator convex on $(0,+\infty)$.
\end{itemize}
\end{prop}

\begin{proof}
That (i) implies (ii) and (iii) is trivial. Since $\Tr P_f(B^{1/2}AB^{1/2},B)=
\Tr B^{1/2}f(A)B^{1/2}$ for $A,B\in\BH_{++}$, we have (ii)\,$\imp$\,(iv). The proof of
(iv)\,$\imp$\,(i) is in \cite{ENG} or \cite{EH}. Apply these to $\widetilde f$ and use Lemma
\ref{lemma:transpose perspective} to prove that (iii)\,$\imp$\,(v)\,$\imp$\,(i).

Although the equivalence of (iv) and (v) has already been shown, it may be worth giving a
different proof based on Kraus' characterization of operator convex functions, see
\cite[Corollary 2.7.8]{Hiai_book}. Indeed, since
$$
{\widetilde f(x)-\widetilde f(1)\over x-1}=-{f(x^{-1})-f(1)\over x^{-1}-1}+f(1),
$$
the operator convexity of $\widetilde f$ follows from that of $f$ and vice versa by Kraus'
theorem.
\end{proof}

\begin{prop}\label{P-A.2}
Let $h$ be a real function on $(0,+\infty)$. The following conditions are equivalent
\begin{itemize}
\item[(i)] $(A,B)\in\BH_{++}\times\BH_{++}\mapsto P_h(A,B)$ is jointly operator monotone
increasing for any finite-dimensional Hilbert space $\cH$;
\item[(ii)] $(A,B)\in\BH_{++}\times\BH_{++}\mapsto\Tr P_h(A,B)$ is jointly monotone increasing
for any finite-dimensional Hilbert space $\cH$;
\item[(iii)] $h$ is non-negative and operator monotone on $(0,+\infty)$;
\item[(iv)] $h$ is operator monotone on $(0,+\infty)$ and $\widetilde h$ is numerically
increasing on $(0,+\infty)$;
\item[(v)] $\widetilde h$ is non-negative and operator monotone on $(0,+\infty)$;
\item[(vi)] $\widetilde h$ is operator monotone on $(0,+\infty)$ and $h$ is numerically
increasing on $(0,+\infty)$.
\end{itemize}
\end{prop}

To prove the proposition, we give a lemma.

\begin{lemma}\label{L-A.3}
Let $h$ be an operator monotone function on $(0,+\infty)$. Then $h$ is non-negative on
$(0,+\infty)$ if and only if $\widetilde h$ is numerically increasing on $(0,+\infty)$.
\end{lemma}

\begin{proof}
According to \cite[Theorem 1.9]{FHR}, $h$ admits the integral representation
$$
h(x)=h(1)+\gamma(x-1)+\int_{[0,+\infty)}{x-1\over x+s}\,d\mu(s),
\qquad x\in(0,+\infty),
$$
with a constant $\gamma\ge0$ and a positive measure $\mu$ on $[0,+\infty)$ such that
$\int_{[0,+\infty)}(1+s)^{-1}\,d\mu(s)<+\infty$. Note that
$$
{d^2\over dx^2}\biggl({x-1\over x+s}\biggr)={1+s\over(x+s)^2}\ge0.
$$
Since $(1-x)/(x+s)\nearrow1/s$ for $s\ge0$ (where $1/0=+\infty$) as $x\searrow0$, the monotone
convergence theorem yields that
$$
\lim_{x\searrow0}\int_{[0,+\infty)}{1-x\over x+s}\,d\mu(s)
=\int_{[0,+\infty)}{1\over s}\,d\mu(s),
$$
so that
$$
h(0+)=\lim_{x\searrow0}h(x)=h(1)-\gamma-\int_{[0,+\infty)}{1\over s}\,d\mu(s)
\in[-\infty,+\infty).
$$
This implies that $h$ is non-negative on $(0,+\infty)$, i.e., $h(0+)\ge0$ if and only if
\begin{equation}\label{F-A.1}
\int_{[0,+\infty)}s^{-1}\,d\mu(s)<+\infty\quad\mbox{and}\quad
h(1)-\gamma-\int_{[0,+\infty)}{1\over s}\,d\mu(s)\ge0.
\end{equation}
Moreover, note that 
$$
\widetilde h(x)=xh(x^{-1})=h(1)x+\gamma(1-x)+\int_{[0,+\infty)}{1-x\over x^{-1}+s}\,d\mu(s).
$$
If \eqref{F-A.1} is satisfied, then
$$
\widetilde h(x)=\gamma+\biggl(h(1)-\gamma
-\int_{[0,+\infty)}{1\over s}\,d\mu(s)\biggr)x
+\int_{[0,+\infty)}{1+s\over s(x^{-1}+s)}\,d\mu(s)
$$
is obviously increasing on $(0,+\infty)$.

Conversely, assume that \eqref{F-A.1} is not satisfied. If
$\int_{[0,+\infty)}s^{-1}\,d\mu(s)=+\infty$, then
$$
\widetilde h(x)=h(1)x+\gamma(1-x)
-(x-1)\int_{[0,+\infty)}{1\over x^{-1}+s}\,d\mu(s)\longrightarrow-\infty
$$
as $x\to+\infty$, so obviously $\widetilde h$ is not increasing on $(0,+\infty)$. If
$\int_{[0,+\infty)}s^{-1}\,d\mu(s)<+\infty$ but
$h(1)-\gamma-\int_{[0,+\infty)}s^{-1}\,d\mu(s)<0$,
then
$$
\widetilde h(x)=\gamma
+\biggl(h(1)-\gamma-\int_{[0,+\infty)}{1\over x^{-1}+s}\,d\mu(s)\biggr)x
+\int_{[0,+\infty)}{1\over x^{-1}+s}\,d\mu(s)
$$
is not increasing on $(0,+\infty)$.
\end{proof}

\noindent{\it Proof of Proposition \ref{P-A.2}.}\enspace
It is trivial that (i) implies (ii). If (ii) holds, then as in the proof of Proposition
\ref{P-A.1}, one can see that $h$ and $\widetilde h$ are operator monotone on $(0,+\infty)$, so
(iv) follows. By Lemma \ref{L-A.3}, (iii)\,$\iff$\,(iv) and (v)\,$\iff$\,(vi) are obvious.
(iii)\,$\iff$\,(v) is well-known, see, e.g., \cite[Corollary 2.5.6]{Hiai_book}. Although
(iii)\,$\imp$\,(i) is also well-known in the theory of operator means \cite{KA}, we give a proof
for convenience. Since (iii)\,$\iff$\,(v), (iii) means that $h$ and $\widetilde h$ are operator
monotone on $(0,+\infty)$. If $A_1,A_2,B_1,B_2\in\BH_{++}$ are such that $A_1\le A_2$ and
$B_1\le B_2$, then
\begin{align*}
P_h(A_1,B_1)&=B_1^{1/2}h(B_1^{-1/2}A_1B_1^{-1/2})B^{1/2}
\le B_1^{1/2}h(B_1^{-1/2}A_2B_1^{-1/2})B^{1/2} \\
&=A_2^{1/2}\widetilde h(A_2^{-1/2}B_1A_2^{-1/2})A_2^{1/2}
\le A_2^{1/2}\widetilde h(A_2^{-1/2}B_2A_2^{-1/2})A_2^{1/2}=P_h(A_2,B_2)
\end{align*}
by Lemma \ref{lemma:transpose perspective} for the second equality.\qed

\begin{remark}
Note that neither Proposition \ref{P-A.2} nor Lemma \ref{L-A.3} hold when ``operator monotone'' is replaced with ``numerically increasing''. Indeed, for $f(t):=t^2+1$ we have $f\ge 0$ and $f$ is numerically increasing, but 
$\wtilde f(s)=\frac{1}{s}+s$ is neither increasing nor decreasing, and so $P_f$ is neither increasing nor decreasing in its second variable.
\end{remark}

\section{Examples for $\fix{\map}$ and $\M_\map$}
\label{sec:examples}

\begin{example}
Let $U$ be a unitary on $\hil$, which determines the bistochastic map $\map(\cdot)=U(\cdot)U^*$. Then it is trivial to check that
\begin{align*}
\fix{\map}=\{X\in\B(\hil):\,UX=XU\}\subseteq \B(\hil)=\M_{\map},
\end{align*}
and the inclusion is strict unless $U\in\bC I$.
\end{example}

\begin{prop}\label{prop:classical mult domain}
Let $e_1,\ldots,e_d$ be an orthonormal basis in a Hilbert space $\hil$, and let $T\in\bR_+^{d\times d}$ be a stochastic matrix, i.e.,
$\sum_{y=1}^d T_{xy}=1$, $1\le x\le d$. Define 
\begin{align*}
\map(X):=\sum_{x=1}^d\pr{e_x}\sum_{y=1}^d T_{xy}\inner{e_y}{Xe_y},\ds\ds\ds X\in\B(\hil).
\end{align*}
Then $\map$ is a unital CP map, which is trace-preserving if and only if $T$ is bistochastic, i.e., $\sum_{x=1}^dT_{xy}=1$, $1\le y\le d$, as well. If none of the columns of $T$ are zero then
\begin{align*}
\M_{\map}=\big\{X\in\B(\hil):\,&\inner{e_x}{Xe_y}=0,\s x\ne y,\text{ and }\\
&\inner{e_y}{Xe_y}=\inner{e_z}{Xe_z}\s\text{if}\s T_{xy}T_{xz}>0\s\text{for some}\s x\big\}.
\end{align*}
In particular, if $T$ has a strictly positive row then $\M_{\map}=\bC I$.
\end{prop}
\begin{proof}
It is clear from the definition that $\map$ is CPTP.
For every $a,b\in\{1,\ldots,d\}$ and $X\in\B(\hil)$, we have
\begin{align*}
\map(X\diad{e_b}{e_a})&=\map(\diad{Xe_b}{e_a})=\sum_{x=1}^d\pr{e_x} T_{xa}\inner{e_a}{Xe_b}.
\end{align*}
Now, if $a\ne b$ then $\map(\diad{e_a}{e_b})=0$, and if $X\in\M_{\map}$, we get
\begin{align*}
\map(X\diad{e_b}{e_a})&=\map(X)\map(\diad{e_a}{e_b})=0.
\end{align*}
Thus, if the $a$-th column of $T$ is not zero, then we get $\inner{e_a}{Xe_b}=0$ for all
$b\ne a$. In particular,  if none of the columns of $T$ are zero, then all elements in
$\M_{\map}$ are diagonal in the basis $\{e_x\}_{x=1}^d$.

Now, let $F\in\B(\hil)$ be diagonal in the given basis, so that it can be written as $F=\sum_{x=1}^d f(x)\pr{e_x}$. Then
\begin{align*}
\inner{e_x}{\left[\map(F^*F)-\map(F)^*\map(F))\right]e_x}&=
\sum_{y}T_{xy}|f(y)|^2-\sum_{y,z}T_{xy}\ol f(y)T_{xz}f(z)\\
&=
\sum_{y}(T_{xy}-T_{xy}^2)|f(y)|^2-\sum_{y,z:\,y\ne z}T_{xy}\ol f(y)T_{xz}f(z).
\end{align*}
Note that $T_{xy}-T_{xy}^2=T_{xy}(1-T_{xy})=T_{xy}\sum_{z:\,z\ne y}T_{xz}$, and hence the first term above is
\begin{align*}
\sum_{y}(T_{xy}-T_{xy}^2)|f(y)|^2
=
\sum_{y,z:\,y\ne z}T_{xy}T_{xz}|f(y)|^2
=
\half\sum_{y,z:\,y\ne z}T_{xy}T_{xz}(|f(y)|^2+|f(z)|^2).
\end{align*}
The second term can be written as 
\begin{align*}
\sum_{y,z:\,y\ne z}T_{xy}\ol f(y)T_{xz}f(z)
=
\half\sum_{y,z:\,y\ne z}T_{xy}T_{xz}(\ol f(y)f(z)+f(y)\ol f(z)).
\end{align*}
Thus,
\begin{align*}
\inner{e_x}{\left[\map(F^*F)-\map(F)^*\map(F))\right]e_x}&=
\half\sum_{y,z:\,y\ne z}T_{xy}T_{xz}\left[|f(y)|^2+|f(z)|^2-\ol f(y)f(z)-f(y)\ol f(z)\right]\\
&=
\half\sum_{y,z:\,y\ne z}T_{xy}T_{xz}|f(y)-f(z)|^2,
\end{align*}
from which the remaining assertions follow.
\end{proof}

\begin{example}\label{ex:reversed inclusion}
In the setting of Proposition \ref{prop:classical mult domain}, let 
$\map$ be given by the stochastic matrix
\begin{align*}
T:=\begin{bmatrix}
1 &  & & \\
1/4 & 1/4 &1/4 &1/4 \\
1/4 & 1/4 &1/4 &1/4 \\
 & & & 1
\end{bmatrix}\in\bR_+^{4\times 4}.
\end{align*}
By Proposition \ref{prop:classical mult domain}, we see that $\M_{\map}=\bC\egy$. On the
other hand, it is easy to see that both $I$ and $A:=\pr{e_1}+2\pr{e_2}+2\pr{e_3}+3\pr{e_4}$
are fixed points of $\map$, and hence
\begin{align*}
\M_{\map}\subsetneqq\fix{\map}.
\end{align*}
Moreover, $A^2$ is not a fixed point of $\map$, and hence $\fix{\map}$ is not an algebra.

By Lemma \ref{lemma:fix-mult}, $\map$ cannot have a faithful invariant state. Indeed, 
it is easy to see that $\map^*(\rho)=\rho$ if and only if $\rho$ is diagonal in the given
basis, and $\inner{e_2}{\rho e_2}=\inner{e_3}{\rho e_3}=0$.
\end{example}

\section{Example for $\widetilde S_f(\rho\|\sigma)$}
\label{sec:counterex}

For any positive definite $\rho,\sigma\in\B(\hil)$, and any function $f:\,(0,+\infty)\to\bR$, let 
\begin{align*}
\widetilde S_f(\rho\|\sigma):=\Tr\sigma f(\rho^{1/2}\sigma^{-1}\rho^{1/2})
\end{align*}
as in \eqref{bad fdiv def}. In this section we show that there exists a non-linear operator convex function $f$ such that 
$\widetilde S_f$ is neither monotone increasing, nor monotone decreasing under CPTP maps. 

To this end, let 
\begin{align*}
f_{\delta}(x):=1-x+\delta(1-x)^2,\ds\ds\ds x\in[0,+\infty),\s \delta>0.
\end{align*}
For every $\delta>0$, $f_{\delta}$ is a non-linear operator convex function on $(0,+\infty)$; in particular, it is strictly convex. According to the theory of classical $f$-divergences, if $\rho,\sigma\in\B(\hil)_{++}$ commute, and $\map$ is a trace-preserving positive map such that 
$\map(\rho)$ and $\map(\sigma)$ commute, then 
\begin{align}\label{bad fdiv mon} 
\widetilde S_{f_\delta}(\map(\rho)\|\map(\sigma))\le\widetilde S_{f_\delta}(\rho\|\sigma),
\end{align}
and the inequality is in general strict; see, e.g., \cite[Proposition A.3]{HMPB} for details. 

For every $\ep>0$ and every $t\in(0,1)$, let 
\begin{align*}
\rho_0:=\begin{bmatrix}1/2&1/2\\1/2&1/2\end{bmatrix},\ds\ds
\rho_{\ep}:=\frac{1}{1+2\ep}(\rho_0+\ep I),\ds\ds\text{and}\ds\ds
\sigma_t:=\begin{bmatrix}t&0\\0&1-t\end{bmatrix}.
\end{align*}
Then 
\begin{align*}
\widetilde S_{f_\delta}(\rho_{\ep}\|\sigma_t)
&=
\Tr\sigma_t\left[I-\rho_{\ep}^{1/2}\sigma_t\inv\rho_{\ep}^{1/2}+\delta(I-\rho_{\ep}^{1/2}\sigma_t\inv\rho_{\ep}^{1/2})^2\right]\\
&\xrightarrow[\ep,\delta\searrow 0]{}
\Tr\sigma_t\left[I-\rho_0^{1/2}\sigma_t\inv\rho_0^{1/2}\right]=
1-\frac{1}{4t(1-t)}<0,\ds\ds\ds t\in(0,1)\setminus\{1/2\},
\end{align*}
with equality in the last inequality for $t=1/2$.
Taking $\map$ to be the diagonal pinching, we have $\map(\rho_{\ep})=\half I$, $\map(\sigma_t)=\sigma_t$, and 
\begin{align*}
\widetilde S_{f_\delta}(\Phi(\rho_{\ep})\|\Phi(\sigma_t))
&=
tf_{\delta}\bz \frac{1/2}{t}\jz+(1-t)f_{\delta}\bz \frac{1/2}{1-t}\jz \\
&\xrightarrow[\delta\searrow 0]{}t\bz 1-\frac{1}{2t}\jz+(1-t)\bz 1-\frac{1}{2(1-t)}\jz=0.
\end{align*}
Thus, for every $t\in(0,1)\setminus\{1/2\}$, there exist $\ep_t,\delta_t>0$ such that for all 
$0<\ep<\ep_t$ and $0<\delta<\delta_t$, 
\begin{align}\label{opposite-mono}
\widetilde S_{f_{\delta}}(\map(\rho_{\ep})\|\map(\sigma_t))>\widetilde S_{f_{\delta}}(\rho_{\ep}\|\sigma_t).
\end{align}
Together with \eqref{bad fdiv mon}, this shows that
$\widetilde S_{f_{\delta}}$ is neither increasing nor decreasing under CPTP maps. 

\begin{remark}
The above also implies that $\rho\mapsto\widetilde S_{f_\delta}(\rho\|\sigma_t)$ is not convex on invertible density operators
for any $\delta\in(0,\delta_t)$.
Indeed, suppose that $\widetilde S_{f_\delta}(\rho\|\sigma_t)$ is convex in $\rho$. Since
$\Phi(\rho_\eps)=\int_\Gamma U\rho_\eps U^*\,dU$, where $\Gamma$ is the group of diagonal
$2\times2$ unitaries and $dU$ is the Haar probability measure on $\Gamma$, one has a
contradiction to \eqref{opposite-mono} as follows:
$$
\widetilde S_{f_\delta}(\Phi(\rho_\eps)\|\sigma_t)
\le\int_\Gamma\widetilde S_{f_\delta}(U\rho_\eps U^*\|\sigma_t)\,dU
=\int_\Gamma\widetilde S_{f_\delta}(U\rho_\eps U^*\|U\sigma_tU^*)\,dU
=\widetilde S_{f_\delta}(\rho_\eps\|\sigma_t),
$$
where the unitary invariance $\widetilde S_{f_\delta}(U\rho_\eps U^*\|U\sigma_tU^*)=
\widetilde S_{f_\delta}(\rho_\eps\|\sigma_t)$ is obvious.
\end{remark}

\section{Continuity properties of the standard $f$-divergences}
\label{sec:cont}

\noindent\textit{Proof of Proposition \ref{prop:standard fdiv cont}:}

\ref{fdiv cont i}\enspace
Arrange the eigenvalues of $\rho$ and $\sigma$ in decreasing order, counted with multiplicities, as
\begin{align*}
&a_1=\dots=a_{i_1}>a_{i_1+1}=\dots=a_{i_2}>\dots>a_{i_{s-1}+1}=\dots=a_{i_s},\\
&b_1=\dots=b_{j_1}>b_{j_1+1}=\dots=b_{j_2}>\dots>b_{j_{r-1}+1}=\dots=b_{j_r},
\end{align*}
where $i_s=j_r=d=\dim\hil$.
Let $P_l$ be the spectral projection of $\rho$ corresponding to $a_{i_l}$, and 
$Q_k$ be the spectral projection of
$\sigma$ corresponding to the eigenvalue $b_{j_k}$.
Let $(\rho_n)_{n\in\bN}$ and $(\sigma_n)_{n\in\bN}$ be two sequences in $\BH_+$ such that
$\rho_n\to\rho$, $\sigma_n\to\sigma$ as $n\to+\infty$. Choose spectral decompositions 
$\rho_n=\sum_{i=1}^da_i^{(n)}P_i^{(n)}$ with eigenvalues
$a_1^{(n)}\ge\dots\ge a_d^{(n)}$ and orthogonal rank one projections $P_i^{(n)}$, and similarly, 
$\sigma_n=\sum_{j=1}^db_j^{(n)}Q_j^{(n)}$
with eigenvalues $b_1^{(n)}\ge\dots\ge b_d^{(n)}$ and orthogonal rank one projections
$Q_j^{(n)}$. Then, as $n\to\infty$,
$a_i^{(n)}\to a_i$ and $\sum_{i=i_{l-1}+1}^{i_l}P_i^{(n)}\to P_l$ for every $l=1,\dots,s$, 
and similarly, $b_j^{(n)}\to b_j$ and $\sum_{j=j_{k-1}+1}^{j_k}Q_j^{(n)}\to Q_k$
for every $k=1,\dots,r$, where $i_0:=j_0:=0$.
Then one has
\begin{align}\label{cont proof1}
S_f(\rho_n\|\sigma_n)=\sum_{i,j=1}^d P_f\bz a_i^{(n)},b_j^{(n)}\jz\Tr P_i^{(n)}Q_j^{(n)}.
\end{align}
Under the assumption that both $f(0^+)$ and $f'(+\infty)$ are finite, 
the perspective function $P_f$ is continuous on $[0,+\infty)\times[0,+\infty)$, and hence
\eqref{cont proof1} converges to the expression in \eqref{standard fdiv1}.

\ref{fdiv cont ii}\enspace
Assume that $f$ is operator convex on $(0,+\infty)$, and let $L_n$ be given as stated.
First, by the joint convexity in Proposition \ref{P-3.6} and Remark \ref{rem:subadd} we see that
$$
S_f(\rho+L_n\|\sigma+L_n)\le S_f(\rho\|\sigma)+S_f(L_n\|L_n)
=S_f(\rho\|\sigma)+f(1)\Tr L_n
$$
so that $\limsup_{n\to\infty}S_f(\rho+L_n\|\sigma+L_n)\le S_f(\rho\|\sigma)$. Thus, to obtain the
result, it remains to prove that
$\liminf_{n\to\infty}S_f(\rho+L_n\|\sigma+L_n)\ge S_f(\rho\|\sigma)$. To do this, we use the
integral expression \eqref{F-2.3} that we rewrite as
\begin{align}\label{Eq-1}
f(x)=a+bx+cx^2+dx^{-1}+\int_{(0,+\infty)}\psi_s(x)\,d\lambda(s),\qquad x\in(0,+\infty),
\end{align}
where $a,b\in\bR$, $c,d\ge0$, and
$$
\psi_s(x):={(x-1)^2\over x+s},\qquad s>0.
$$
Now we consider the functions $f_1(x):=a+bx$, $f_2(x):=x^2$, $f_{-1}(x):=x^{-1}$, and $\psi_s(x)$ for
$s>0$, separately. For $f_1$ and $\psi_s$ we have by the previous point
\begin{align}
\lim_{n\to\infty}S_{f_1}(\rho+L_n\|\sigma+L_n)
&=S_{f_1}(\rho\|\sigma), \label{Eq-2}\\
\lim_{n\to\infty}S_{\psi_s}(\rho+L_n\|\sigma+L_n)
&=S_{\psi_s}(\rho\|\sigma), \label{Eq-3}
\end{align}
where \eqref{Eq-2} is also obvious since $S_{f_1}(\rho\|\sigma)=a\Tr \sigma+b\Tr\rho$.
For $f_2$ and $f_{-1}$ we prove the following:
\begin{align}
\liminf_{n\to\infty}S_{f_2}(\rho+L_n\|\sigma+L_n)
&\ge S_{f_2}(\rho\|\sigma), \label{Eq-4}\\
\liminf_{n\to\infty}S_{f_{-1}}(\rho+L_n\|\sigma+L_n)
&\ge S_{f_{-1}}(\rho\|\sigma). \label{Eq-5}
\end{align}
When these have been proved, combining \eqref{Eq-2}--\eqref{Eq-5} yields that
\begin{align*}
\liminf_{n\to\infty}S_f(\rho+L_n\|\sigma+L_n)
&\ge 
S_{f_1}(\rho\|\sigma)+cS_{f_2}(\rho\|\sigma)+dS_{f_{-1}}(\rho\|\sigma) \\
&\ds+\liminf_{n\to\infty}\int_{(0,+\infty)}S_{\psi_s}(\rho+L_n\|\sigma+L_n)\,d\lambda(s) \\
&\ge S_{f_1}(\rho\|\sigma)+cS_{f_2}(\rho\|\sigma)+dS_{f_{-1}}(\rho\|\sigma)
+\int_{(0,+\infty)}S_{\psi_s}(\rho\|\sigma)\,d\lambda(s) \\
&=S_f(\rho\|\sigma),
\end{align*}
where the second inequality in the above is due to Fatou's lemma since
$S_{\psi_s}(\rho+L_n\|\sigma+L_n)\ge0$ for all $s\in(0,+\infty)$ and $n\in\bN$. Thus
\eqref{Eq-1} follows. Hence, we are left to prove \eqref{Eq-4} and \eqref{Eq-5}.

Proof of \eqref{Eq-4}:\enspace
By (3.10) and Corollary \ref{cor:fdiv infty}, note that
$$
S_{f_2}(\rho\|\sigma)=\begin{cases}\Tr\rho^2\sigma^{-1} & \text{if $\rho^0\le\sigma^0$}, \\
+\infty & \text{if $\rho^0\not\le\sigma^0$},\end{cases}
$$
and
$$
S_{f_2}(\rho+L_n\|\sigma+L_n)=\Tr(\rho+L_n)^2(\sigma+L_n)^{-1}.
$$
Assume that $\rho^0\le\sigma^0$. Apply the monotonicity property (Proposition \ref{P-3.8}) to the
pinching $\Phi(X):=\sigma^0X\sigma^0+(I-\sigma^0)X(I-\sigma^0)$ to obtain
\begin{align*}
&S_{f_2}(\rho+L_n\|\sigma+L_n) \\
&\quad\ge S_{f_2}(\Phi(\rho+L_n)\|\Phi(\sigma+L_n)) \\
&\quad=S_{f_2}\bigl((\rho+\sigma^0L_n\sigma^0)+(I-\sigma^0)L_n(I-\sigma^0)\|
(\sigma+\sigma^0L_n\sigma^0)+(I-\sigma^0)L_n(I-\sigma^0)\bigr) \\
&\quad=\Tr(\rho+\sigma^0L_n\sigma^0)^2(\sigma+\sigma^0L_n\sigma^0)^{-1}
+\Tr(I-\sigma^0)L_n(I-\sigma^0).
\end{align*}
Since $\Tr(\rho+\sigma^0L_n\sigma^0)^2(\sigma+\sigma^0L_n\sigma^0)^{-1}\to\Tr\rho^2\sigma^{-1}$
and $\Tr(I-\sigma^0)L_n(I-\sigma^0)\to0$ as $n\to\infty$,  \eqref{Eq-4} holds in the case
$\rho^0\le\sigma^0$.

Next, assume that $\rho^0\not\le\sigma^0$ (hence $L_n\ne0$ for all $n$). Since
$$
(\sigma+L_n)^{-1}\ge(\sigma+\|L_n\|_\infty I)^{-1}\ge\|L_n\|_\infty^{-1}(I-\sigma^0)
$$
with the operator norm $\|L_n\|_\infty$, one has
$$
\Tr(\rho+L_n)^2(\sigma+L_n)^{-1}
\ge\|L_n\|_\infty^{-1}\Tr(\rho^2+\rho L_n+L_n\rho+L_n^2)(I-\sigma^0),
$$
which implies \eqref{Eq-4} in this case too, by noting that $\|L_n\|_\infty^{-1}\to+\infty$,
$\Tr\rho^2(I-\sigma^0)>0$ due to $\rho^0\not\le\sigma^0$, and $\rho L_n+L_n\rho+L_n^2\to0$.

Proof of \eqref{Eq-5}:\enspace
The proof is immediate from the above and Proposition \ref{P-3.4}, since
$f_{-1}=\widetilde f_2$.
\hfill\qed
\medskip

Unlike in the classical case, the continuity property stated in \ref{fdiv cont ii} of Proposition \ref{prop:standard fdiv cont} may not hold when $f$ is only assumed to be convex, as the following example shows.

\begin{example}\label{ex:cont counter}
Here we give an example where $\rho=\sigma$ and $\rho_\eps=\rho+\eps K$, $\sigma_\eps=\sigma+\eps L$ such that $K,L\ge0$ and $\rho+K,\sigma+L>0$, but 
$\lim_{\ep\searrow 0}S_f(\rho+\eps K\|\sigma+\eps L)\ne S_f(\rho\|\sigma)$ for a convex but not operator convex $f$.

Let
$\rho=\sigma=\begin{bmatrix}1&0\\0&0\end{bmatrix}$ and $K=\begin{bmatrix}0&0\\0&1\end{bmatrix}$,
$L=\begin{bmatrix}1/2&1/2\\1/2&1/2\end{bmatrix}$ so that $\rho+K,\sigma+L>0$. 
 The eigenvalues of $\sigma+\eps L$
are
$$
b_1^{(\eps)}={1+\eps+\sqrt{1+\eps^2}\over2},\qquad
b_2^{(\eps)}={1+\eps-\sqrt{1+\eps^2}\over2},
$$
whose unit eigen-vectors are
$$
y_1^{(\eps)}=\begin{bmatrix}{\eps\over\bigl[2(1+\eps^2-\sqrt{1+\eps^2})\bigr]^{1/2}}
\\{\sqrt{1+\eps^2}-1\over\bigl[2(1+\eps^2-\sqrt{1+\eps^2})\bigr]^{1/2}}\end{bmatrix},\qquad
y_2^{(\eps)}=\begin{bmatrix}{\eps\over\bigl[2(1+\eps^2+\sqrt{1+\eps^2})\bigr]^{1/2}}
\\-{\sqrt{1+\eps^2}+1\over\bigl[2(1+\eps^2+\sqrt{1+\eps^2}\bigr]^{1/2}}\end{bmatrix},
$$
respectively. Therefore, with $P_1=\begin{bmatrix}1&0\\0&0\end{bmatrix}$,
$P_2=\begin{bmatrix}0&0\\0&1\end{bmatrix}$, $Q_1^{(\eps)}=|y_1^{(\eps)}\>\<y_1^{(\eps)}|$ and
$Q_2^{(\eps)}=|y_2^{(\eps)}\>\<y_2^{(\eps)}|$, we have
\begin{align}\label{cont counterex1}
S_f(\rho+\eps K\|\sigma+\eps L)
&=b_1^{(\eps)}f(1/b_1^{(\eps)})\Tr P_1Q_1^{(\eps)}
+b_2^{(\eps)}f(1/b_2^\eps)\Tr P_1Q_2^{(\eps)} \nonumber\\
&\quad+b_1^{(\eps)}f(\eps/b_1^{(\eps)})\Tr P_2Q_1^{(\eps)}
+b_2^{(\eps)}f(\eps/b_2^{(\eps)})\Tr P_2Q_2^{(\eps)}.
\end{align}

Since $b_1^{(\eps)}\to1$ and $y_1^{(\eps)}\to\begin{bmatrix}1\\0\end{bmatrix}$ so that
$Q_1^{(\eps)}\to P_1$, we have as $\eps\searrow0$,
$$
b_1^{(\eps)}f(1/b_1^{(\eps)})\Tr P_1Q_1^{(\eps)}\ \longrightarrow\ f(1)=S_f(\rho\|\sigma).
$$
On the other hand,
\begin{align*}
b_2^{(\eps)}f(1/b_2^{(\eps)})\Tr P_1Q_2^{(\eps)}
&={1+\eps-\sqrt{1+\eps^2}\over2}\,f\biggl({2\over1+\eps-\sqrt{1+\eps^2}}\biggr)
{\eps^2\over2(1+\eps^2+\sqrt{1+\eps^2})} \\
&={\eps^3\over2(1+\eps+\sqrt{1+\eps^2})(1+\eps^2+\sqrt{1+\eps^2})}
\,f\biggl({1+\eps+\sqrt{1+\eps^2}\over\eps}\biggr).
\end{align*}
For example, when $f(x)=x^\alpha$ where $\alpha>0$, we find that
$$
b_2f(1/b_2^\eps)\Tr P_1Q_2^{(\eps)}=
{\eps^{3-\alpha}(1+\eps+\sqrt{1+\eps^2})^{\alpha-1}\over2(1+\eps^2+\sqrt{1+\eps^2})}
\ \longrightarrow\ \begin{cases}0 & \text{if $0<\alpha<3$} \\ 1 & \text{if $\alpha=3$} \\
+\infty & \text{if $\alpha>3$}
\end{cases}
$$
Since the other two terms in \eqref{cont counterex1} are non-negative, we have 
\begin{align*}
\lim_{\ep\searrow 0}S_f(\rho+\eps K\|\sigma+\eps L)=+\infty\ne S_f(\rho\|\sigma)
\end{align*}
for $f(x)=x^\alpha$ with $\alpha>3$.
\end{example}
\medskip

We also have the following one-sided continuity result:
\begin{prop}\label{prop:onesided cont}
Let $\rho,\sigma\in\B(\hil)_+$ and let $f$ be an operator convex function on $(0,+\infty)$.
\begin{enumerate}
\item\label{onesided cont1}
If $f(0^+)<+\infty$ and $\rho^0\le\sigma^0$, then
$$
S_f(\rho\|\sigma)=\lim_{n\to\infty}S_f(\rho\|\sigma+L_n)
$$
for any sequence $L_n\in\BH_+$ with $L_n\to0$.
\item\label{onesided cont2}
If $f'(+\infty)<+\infty$ and $\sigma^0\le\rho^0$, then
$$
S_f(\rho\|\sigma)=\lim_{n\to\infty}S_f(\rho+K_n\|\sigma)
$$
for any sequence $K_n\in\BH_+$ with $K_n\to0$.
\end{enumerate}
\end{prop}
\begin{proof}
By Remark \ref{rem:subadd},
$$
S_f(\rho\|\sigma+L_n)\le S_f(\rho\|\sigma)+S_f(0\|L_n)=S_f(\rho\|\sigma)+f(0^+)\Tr L_n,
$$
so it is enough to prove that $\liminf_{n\to\infty}S_f(\rho\|\sigma+L_n)\ge S_f(\rho\|\sigma)$.
The assumption $f(0^+)<+\infty$ guarantees that the $f_{-1}$ term does not appear in the integral representation \eqref{Eq-1}.
Moreover, by \ref{fdiv cont i} of Proposition \ref{prop:standard fdiv cont}, we have
\begin{align}
\lim_{n\to\infty}S_{f_1}(\rho\|\sigma+L_n)
&=S_{f_1}(\rho\|\sigma),\ds\ds\ds
\lim_{n\to\infty}S_{\psi_s}(\rho\|\sigma+L_n)
=S_{\psi_s}(\rho\|\sigma).
\end{align}
Thus, it is enough to show show that $\liminf_{n\to\infty}S_{f_2}(\rho\|\sigma+L_n)\ge S_{f_2}(\rho\|\sigma)$.
This is easy as in the above proof
of \eqref{Eq-4}:
\begin{align*}
S_{f_2}(\rho\|\sigma+L_n)
&\ge S_{f_2}\bigl(\rho\|(\sigma+\sigma^0L_n\sigma^0)+(I-\sigma^0)L_n(I-\sigma^0)\bigr) \\
&=\Tr\rho^2(\sigma+\sigma^0L_n\sigma^0)^{-1}
\ \longrightarrow\ \Tr\rho^2\sigma^{-1}=S_{f_2}(\rho\|\sigma),
\end{align*}
where the convergence holds due to the assumption $\rho^0\le\sigma^0$.
This proves \ref{onesided cont1}, and \ref{onesided cont2} follows by using 
$S_f(\rho\|\sigma)=S_{\wtilde f}(\sigma\|\rho)$.
\end{proof}

Again, assuming only convexity of $f$ is not sufficient for the above proposition.

\begin{example}
Let $\rho,\sigma$ and $L$ as in Example \ref{ex:cont counter}, and $f(x):=x^{\alpha}$ with some fixed $\alpha>3$.
Then $f(0^+)=0$, $\rho^0\le\sigma^0$, and the same calculation as in Example \ref{ex:cont counter} shows that 
\begin{align*}
\lim_{\ep\searrow 0}S_f(\rho\|\sigma+\eps L)=+\infty\ne S_f(\rho\|\sigma)
\end{align*}
for $f(x)=x^\alpha$ with $\alpha>3$.
\end{example}


\section{Proof of Proposition \ref{prop:persp extension2}}
\label{sec:extension proof}

\noindent
{\it Proof of \ref{op conv cond4}.}\enspace
Since $f$ is operator convex, $g(x):=(f(x)-f(1))/(x-1)$ where $g(1):=f'(1)$ is an operator
monotone function on $(0,+\infty)$, and $f(0^+)<+\infty$ implies that $g(0):=g(0^+)$ is finite.
Thus, $h(x):=g(x)-g(0^+)$ is a non-negative operator monotone function on $[0,+\infty)$, so that
$$
f(x)=\alpha+\beta x+(x-1)h(x),\qquad x\in(0,\infty),
$$
where $\alpha,\beta\in\bR$. If $h(1)=0$ then $h$ is identically zero, and the assertion is
trivial, so for the rest we can assume that $h(1)=1$, by possibly replacing $h$ with $h/h(1)$.
Then we can write
\begin{align*}
&P_f(\rho+K_n,\sigma+K_n) \\
&\quad=\alpha(\sigma+K_n)+\beta(\rho+K_n) \\
&\qquad+(\sigma+K_n)^{1/2}\bigl[(\sigma+K_n)^{-1/2}(\rho+K_n)
(\sigma+K_n)^{-1/2}-I\bigr] \\
&\qquad\qquad\times h\bigl((\sigma+K_n)^{-1/2}(\rho+K_n)
(\sigma+K_n)^{1/2}\bigr)(\sigma+K_n)^{1/2} \\
&\quad=\alpha(\sigma+K_n)+\beta(\rho+K_n)
+\bigl[(\rho+K_n)(\sigma+K_n)^{-1}-I\bigr]\bigl[(\sigma+K_n)\,\tau_h\,(\rho+K_n)\bigr].
\end{align*}
On the other hand, we write
\begin{align*}
\sigma^{1/2}f(\sigma^{-1/2}\rho\sigma^{-1/2})\sigma^{1/2}
&=\alpha\sigma+\beta\rho+\sigma^{1/2}(\sigma^{-1/2}\rho\sigma^{-1/2}-P)
h(\sigma^{-1/2}\rho\sigma^{-1/2})\sigma^{1/2} \\
&=\alpha\sigma+\beta\rho+(\rho\sigma^{-1}-P)(\sigma\,\tau_h\,\rho),
\end{align*}
where $P:=\sigma^0$ and the operator mean $\sigma\,\tau_h\,\rho$ is defined as an operator in
$\cB(P\cH)$. Set
\begin{equation}\label{F-B.1}
Y_n:=\bigl[(\rho+K_n)(\sigma+K_n)^{-1}-I\bigr]\bigl[(\sigma+K_n)\,\tau_h\,(\rho+K_n)\bigr],
\end{equation}
and write $Y_n$ in the form of $2\times2$ block matrices under the decomposition
$\cH=P\cH\oplus(I-P)\cH$ as
$$
Y_n=\begin{bmatrix}Y_{11}^{(n)}&Y_{12}^{(n)}\\Y_{21}^{(n)}&Y_{22}^{(n)}\end{bmatrix},
\qquad(Y_{12}^{(n)})^*=Y_{21}^{(n)}.
$$
What we need to prove is that, as $n\to\infty$,
\begin{equation}\label{F-B.2}
Y_{11}^{(n)}\longrightarrow(\rho\sigma^{-1}-P)(\sigma\,\tau_h\,\rho),\quad
Y_{12}^{(n)}\longrightarrow0\ \ (\mbox{hence}\ \ Y_{21}^{(n)}\longrightarrow0),\quad
Y_{22}^{(n)}\longrightarrow0.
\end{equation}
We also write
$$
K_n:=\begin{bmatrix}K_{11}^{(n)}&K_{12}^{(n)}\\K_{21}^{(n)}&K_{22}^{(n)}\end{bmatrix},
\qquad(K_{12}^{(n)})^*=K_{21}^{(n)}.
$$
Since $K_n\ge0$, we note (see, e.g., \cite[Proposition 1.3.2]{Bhatia2}) that
\begin{equation}\label{F-B.3}
K_{12}^{(n)}=(K_{11}^{(n)})^{1/2}W_n(K_{22}^{(n)})^{1/2},\quad
K_{21}^{(n)}=(K_{22}^{(n)})^{1/2}W_n^*(K_{11}^{(n)})^{1/2}\quad
\mbox{with}\quad\|W_n\|\le1,
\end{equation}
and $K_n\to0$ means that $K_{ij}^{(n)}\to0$ ($i,j=1,2$) as $n\to\infty$. Furthermore, we write
\begin{equation}\label{F-B.4}
(\rho+K_n)(\sigma+K_n)^{-1}
=\begin{bmatrix}X_{11}^{(n)}&X_{12}^{(n)}\\X_{21}^{(n)}&X_{22}^{(n)}\end{bmatrix}.
\end{equation}

We will perform the following computations, for the sake of brevity, with disregarding the
superscript $^{(n)}$. Since
$$
\begin{bmatrix}\rho+K_{11}&K_{12}\\K_{21}&K_{22}\end{bmatrix}
=\begin{bmatrix}X_{11}&X_{12}\\X_{21}&X_{22}\end{bmatrix}
\begin{bmatrix}\sigma+K_{11}&K_{12}\\K_{21}&K_{22}\end{bmatrix},
$$
we have
\begin{align}
\rho+K_{11}&=X_{11}(\sigma+K_{11})+X_{12}K_{21}, \label{F-B.5}\\
K_{12}&=X_{11}K_{12}+X_{12}K_{22}, \label{F-B.6}\\
K_{21}&=X_{21}(\sigma+K_{11})+X_{22}K_{21}, \label{F-B.7}\\
K_{22}&=X_{21}K_{12}+X_{22}K_{22}. \label{F-B.8}
\end{align}
Note that $K_{22}>0$ since $\sigma+K_n>0$. By \eqref{F-B.3} and \eqref{F-B.6} one finds
\begin{align*}
X_{12}K_{21}&=X_{12}K_{22}^{1/2}W_n^*K_{11}^{1/2}
=(I_1-X_{11})K_{12}K_{22}^{-1/2}W_n^*K_{11}^{1/2} \\
&=(I_1-X_{11})K_{11}^{1/2}W_nW_n^*K_{11}^{1/2}.
\end{align*}
Therefore, by \eqref{F-B.5},
$$
\rho+K_{11}=X_{11}(\sigma+K_{11})+(I_1-X_{11})K_{11}^{1/2}W_nW_n^*K_{11}^{1/2}
$$
so that
\begin{equation}\label{F-B.9}
X_{11}=\Bigl(\rho+K_{11}-K_{11}^{1/2}W_nW_n^*K_{11}^{1/2}\Bigr)
\Bigl(\sigma+K_{11}-K_{11}^{1/2}W_nW_n^*K_{11}^{1/2}\Bigr)^{-1}
\end{equation}
for sufficiently large $n$. Here, note that the operator in the second bracket above is
invertible for large $n$. By \eqref{F-B.7}, \eqref{F-B.3} and \eqref{F-B.8} one further finds
\begin{align*}
X_{21}(\sigma+K_{11})
&=K_{21}-X_{22}K_{21}=K_{21}-X_{22}K_{22}^{1/2}W_n^*K_{11}^{1/2} \\
&=K_{21}-(K_{22}-X_{21}K_{12})K_{22}^{-1/2}W_n^*K_{11}^{1/2} \\
&=K_{21}-K_{22}^{1/2}W_n^*K_{11}^{1/2}+X_{21}K_{11}^{1/2}W_nW_n^*K_{11}^{1/2},
\end{align*}
which implies that
\begin{equation}\label{F-B.10}
X_{21}=\Bigl(K_{21}-K_{22}^{1/2}W_n^*K_{11}^{1/2}\Bigr)
\Big(\sigma+K_{11}-K_{11}^{1/2}W_nW_n^*K_{11}^{1/2}\Bigr)^{-1}
\end{equation}
for sufficiently large $n$. Furthermore, by \eqref{F-B.6} and \eqref{F-B.3} one has
\begin{equation}\label{F-B.11}
X_{12}K_{22}^{1/2}=(I_1-X_{11})K_{12}K_{22}^{-1/2}=(I_1-X_{11})K_{11}^{1/2}W_n.
\end{equation}

Now, via \eqref{F-B.9}, \eqref{F-B.10}, \eqref{F-B.11} and \eqref{F-B.8} we obtain the convergences
\begin{equation}\label{F-B.12}
X_{11}^{(n)}\longrightarrow\rho\sigma^{-1},\quad
X_{21}^{(n)}\longrightarrow0,\quad
X_{12}^{(n)}(K_{22}^{(n)})^{1/2}\longrightarrow0,\quad
X_{22}^{(n)}K_{22}^{(n)}\longrightarrow0.
\end{equation}
We next write
\begin{equation}\label{F-B.13}
(\sigma+K_n)\,\tau_h\,(\rho+K_n)
=\begin{bmatrix}A_{11}^{(n)}&A_{12}^{(n)}\\A_{21}^{(n)}&A_{22}^{(n)}\end{bmatrix},
\qquad (A_{12}^{(n)})^*=A_{21}^{(n)},
\end{equation}
and note as in \eqref{F-B.3} that
\begin{equation}\label{F-B.14}
A_{21}^{(n)}=(A_{22}^{(n)})^{1/2}V_n(A_{11}^{(n)})^{1/2}\quad\mbox{with}\quad\|V_n\|\le1.
\end{equation}
Since $(\sigma+K_n)\,\tau_h\,(\rho+K_n)\to\sigma\,\tau_h\,\rho$, we have
\begin{equation}\label{F-B.15}
A_{11}^{(n)}\longrightarrow\sigma\,\tau_h\,\rho,\qquad
A_{12}^{(n)}\longrightarrow0,\qquad A_{22}^{(n)}\longrightarrow0.
\end{equation}
Moreover, since
$$
(\sigma+K_n)\,\tau_h\,(\rho+K_n)\ge K_n\,\tau_h\,K_n=K_n,
$$
we have $A_{22}^{(n)}\ge K_{22}^{(n)}$. On the other hand, by the transformer inequality,
\begin{align*}
A_{22}^{(n)}&=P^\perp\bigl[(\sigma+K_n)\,\tau_h\,(\rho+K_n)\bigr]P^\perp \\
&\le\bigl(P^\perp(\sigma+K_n)P^\perp\bigr)\,\tau_h\,\bigl(P^\perp(\rho+K_n)P^\perp\bigr) \\
&=(P^\perp K_nP^\perp)\,\tau_h\,(P^\perp K_nP^\perp)=K_{22}^{(n)}.
\end{align*}
Therefore,
\begin{equation}\label{F-B.16}
A_{22}^{(n)}=K_{22}^{(n)}.
\end{equation}

Finally, since \eqref{F-B.1}, \eqref{F-B.4} and \eqref{F-B.13} give
\begin{align*}
Y_n=\begin{bmatrix}X_{11}^{(n)}-I_1&X_{12}^{(n)}\\X_{21}^{(n)}&X_{22}^{(n)}-I_2\end{bmatrix}
\begin{bmatrix}A_{11}^{(n)}&A_{12}^{(n)}\\A_{21}^{(n)}&A_{22}^{(n)}\end{bmatrix},
\end{align*}
we have, by \eqref{F-B.14} and \eqref{F-B.16},
\begin{align*}
Y_{11}^{(n)}&=(X_{11}^{(n)}-I_1)A_{11}^{(n)}
+X_{12}^{(n)}(K_{22}^{(n)})^{1/2}V_n(A_{11}^{(n)})^{1/2}, \\
Y_{12}^{(n)}&=(X_{11}^{(n)}-I_1)A_{12}^{(n)}+X_{12}^{(n)}K_{22}^{(n)}, \\
Y_{22}^{(n)}&=X_{21}^{(n)}A_{12}^{(n)}+(X_{22}^{(n)}-I_2)K_{22}^{(n)}.
\end{align*}
Together with \eqref{F-B.12} and \eqref{F-B.15} these yield the required convergences in
\eqref{F-B.2}.

\medskip\noindent
{\it Proof of \ref{op conv cond5}.}\enspace
This follows from \ref{op conv cond4} by Lemma \ref{lemma:transpose perspective}.

\medskip\noindent
{\it Proof of (iii).}\enspace
From the integral expression \eqref{F-2.3} we define
\begin{align*}
f_0(x)&:=f(1)+f'(1)(x-1)+c(x-1)^2+\int_{[1,+\infty)}{(x-1)^2\over x+s}\,d\lambda(s), \\
f_1(x)&:=\int_{[0,1)}{(x-1)^2\over x+s}\,d\lambda(s),\qquad x\in(0,+\infty).
\end{align*}
Then $f_0$ and $f_1$ are operator convex functions on $(0,+\infty)$ such that $f=f_0+f_1$. Since
$\int_{[0,+\infty)}(1+s)^{-1}\,d\lambda(s)<+\infty$, note that
$$
\int_{[1,+\infty)}{1\over s}\,d\lambda(s)<+\infty,\qquad
\int_{[0,1)}\,d\lambda(s)<+\infty.
$$
Hence it is easy to see that $f_0(0^+)<+\infty$ and $f_1'(+\infty)<+\infty$. So one can apply (i)
to $f_0$ and (ii) to $f_1$ to obtain
\begin{align}
\lim_{n\to\infty}\per{f_0}(\rho+K_n,\sigma+K_n)
&=\sigma^{1/2}f_0\bz \sigma^{-1/2}\rho\sigma^{-1/2}\jz\sigma^{1/2} \label{F-B.17}\\
\lim_{n\to\infty}\per{f_1}(\rho+K_n,\sigma+K_n)
&=\rho^{1/2}\widetilde f_1\bz \rho^{-1/2}\sigma\rho^{-1/2}\jz\rho^{1/2}. \label{F-B.18}
\end{align}
Since the assumption $\rho^0=\sigma^0$ gives
$$
\sigma^{1/2}f_k\bz \sigma^{-1/2}\rho\sigma^{-1/2}\jz\sigma^{1/2}
=\rho^{1/2}\widetilde f_k\bz \rho^{-1/2}\sigma\rho^{-1/2}\jz\rho^{1/2},\qquad k=0,1,
$$
the conclusion of (iii) follows by adding \eqref{F-B.17} and \eqref{F-B.18} together.\qed

\bibliography{bibliography160615}
\end{document}